\documentclass[10pt, english]{article}
\usepackage[utf8]{inputenc}
\usepackage[a4paper,
            left=1in,
            right=1in,
            top=1in,
            bottom=1in,
            footskip=.25in]{geometry}

\usepackage{tikz-cd}
\usepackage{amsmath, bm, stmaryrd, bbm}
\usepackage{amsthm}
\usepackage{amssymb}
\usepackage{hyperref}
\usepackage{dsfont, graphicx}
\usepackage{enumitem}
\usepackage{mathrsfs}
\usepackage{mathtools}

\DeclareMathOperator{\supp}{supp}

\renewcommand{\P}[1]{\operatorname{\mathbb{P}}\left(#1\right)}
\newcommand{\PP}[2]{\operatorname{\mathbb{P}}_{#1}\left(#2\right)}

\newcommand{\E}[1]{\operatorname{\mathbb{E}}\left[#1\right]}

\newcommand{\EE}[2]{\operatorname{\mathbb{E}}_{#1}\left[ #2\right]} 

\newcommand{\VVar}[2]{\operatorname{\text{Var}}_{#1}\left[#2\right]}

\newcommand{\RR}{\mathbb{R}}

\newcommand{\ZZ}{\mathbb{Z}}

\newcommand{\Y}{\mathcal{Y}}

\newtheorem{theorem}{Theorem}
\newtheorem{lemma}[theorem]{Lemma}
\newtheorem{proposition}[theorem]{Proposition}
\newtheorem{assumption}[theorem]{Assumption}
\newtheorem{corollary}[theorem]{Corollary}

\title{Evolution of Fast and Slow Life Histories in Resource-Constrained Populations with Mass-Mortality Events}
\author{\'Eloi Martin and David Steinsaltz}
\date{\today}

\begin{document}

\maketitle

\begin{abstract}
    We study the evolution of the speed of life history in populations competing for a single growth-limiting resource subject to demographic stochasticity and mass mortality events. We focus on a quasi-neutral regime in which competing types have equal resource-use efficiency but differ in life-history speed. 
    
    In the large-carrying-capacity scaling limit, we reduce the model to a one-dimensional jump-diffusion supported on the manifold of ecological equilibria. The drift and diffusion components capture the joint effects of demographic stochasticity and density regulation, while the jump component is driven by the catastrophic mortality events. Using this limiting process, we derive a first-order approximation for the fixation probability of an invading type that differs slightly in life history speed from the resident population. Calculations show that small demographic events tend to favor slower life histories, whereas catastrophic mortality events create transient periods of resource abundance that benefit faster types. The interaction of these two evolutionary forces allows for the existence of a nontrivial evolutionarily attractive life-history speed, which is an increasing function of the frequency and intensity of the catastrophic mortality events.
    
    These results provide a rigorous mathematical framework for some classical r/K-selection arguments, and furthermore demonstrate how mass mortality events can maintain selection for faster life histories even in resource-constrained populations.
\end{abstract}

\noindent\textbf{MSC 2020 subject classifications:} Primary 92D15; Secondary 60J76, 60F17, 92D25.

\noindent\textbf{Keywords:} Quasi-neutrality; life-history evolution; mass mortality events; separation of time scales; Meyer--Zheng topology; jump-diffusion limit; risk-dominant strategy; resource-constrained population.

\medskip

\section{Introduction}\label{sec:intro}

A superficial consideration of the evolution of life histories leads to the recognition that even small advantages in growth rate for one species, or one subpopulation, rapidly turns into overwhelming numerical dominance. 
What tradeoffs prevent organisms from being driven toward the maximum physically possible reproduction rate? 
One answer is suggested in \textit{The Theory of Island Biogeography} \cite{macarthur1967theory}, where MacArthur and Wilson argue that the strain put on resources due to population density reverses the direction of selective pressure from reproduction speed to efficient use of resources:
\begin{quote}
    {\it In an environment with no crowding (\(r\) selection) genotypes which harvest the most food will [\ldots] be most fit. [\ldots] At the other extreme, in a crowded area (\(K\) selection), genotypes which can at least replace themselves [\ldots] at the lowest food level will win.} 
\end{quote}
As any growing population will eventually reach carrying capacity, this flips our initial question on its head: why doesn't evolution drive species towards the lowest possible mortality? Yet, what one observes is neither a trend towards ever increasing or decreasing reproduction speed and mortality, but a wide diversity of life histories. The missing piece, we argue, is the effect of exceptional mass-mortality events. 

In the spirit of \emph{r}/\emph{K}-selection theory, we shall emphasize the differences in evolutionary strategies between ``slow'' (\emph{K}-selected) and ``fast'' (\emph{r}-selected) species. The former thrive in saturated, highly competitive environments, while the latter have the greatest advantage when the population density is low with respect to the carrying capacity. The use of the letters \emph{r} and \emph{K} is derived from the famous logistic equation 
\begin{equation}\label{eq:logisticGrowth}
    \frac{dX}{dt} = r X\bigg(1 - \frac{X}{K}\bigg),
\end{equation}
which expresses the change in the size of a population \(X\) in terms of a growth rate \(r\) and a carrying capacity \(K\). The \emph{r}-selected individuals have a large growth rate, while the \emph{K}-selected individuals have a large carrying capacity. 

The carrying capacity, and more generally the logistic equation, abstract away underlying dynamics of resource production and consumption. In our model we will make those assumptions explicit by introducing an intermediary variable \(R\), representing the quantity of available resources. Introducing \(R\) also allows us to situate the problem with respect to competitive exclusion principles. The resource-ratio theory of competitive exclusion \cite{tilman1985resource} summarizes the efficiency of a species in a single statistic \(\nu\), representing the lowest level of available resources at which the species can maintain itself. A genetic type or species that exploits the resources more efficiently, which may be achieved either by a larger growth rate or by a larger carrying capacity, will unsurprisingly come to replace a type that is less efficient. 
What is more interesting is that different strategies can lead to the same value of \(\nu\), and in our case different life history strategies. The competitive exclusion principle is then inconclusive, and this defines the space where \emph{r}/\emph{K}-selection theory is relevant. This regime where species exploit the resources with the same efficiency but differ in their life history strategies has been dubbed \textit{quasi-neutrality} by Parsons, Quince and Plotkin \cite{parsons2010some, parsons2008absorption}. While Parsons, Quince and Plotkin correctly recognize the importance of this regime, their computation of several key statistics rely on a separation of time scales argument, similar to \cite{katzenberger1991solutions}, that is not given persuasive mathematical justification. 
(\cite{parsons2010some} deferred this technical argument to a preprint that was never completed, and that has since been withdrawn.) Our results, among other things, provide a rigorous foundation for the arguments used by Parsons \emph{et al.}

While there are mathematical results showing the benefit of a larger carrying capacity (or equivalently being of a ``slower'' type) for invading mutations in models combining density regulation and demographic fluctuations\cite{macarthur1962some, lin2012features, balasekaran2022quasi}, the fact that \emph{r}-selected individuals benefit from an unstable environment cannot be verified in the static environment of~\eqref{eq:logisticGrowth}. 
In one mathematical interpretation of an unreliable and suddenly shifting environment, we introduce so-called ``catastrophic events'', where a significant fraction of all living individuals are simultaneously killed off. The ecological impact of such Mass Mortality Events (MMEs) has been the subject of considerable investigation~\cite{di2008note, bansaye2013extinction, fey2019consequences,cattiaux2022random, baeckens2025evolutionary, murray2026disturbing}.

In summary, the population rapidly grows to its carrying capacity when undisturbed, but is periodically pulled away from the carrying capacity by rare catastrophic events. In the aftermath of such an event, the resources made available by the death of a large number of individuals are rapidly consumed by those who have survived, until carrying capacity is reached once again. Fast types will draw a larger benefit from this recovery phase. Between such events, the proportion of fast and slow individuals fluctuates through the effect of demographic stochasticity, which benefits the slow type on average. Hence, the balance of slow and fast individuals is being pulled in opposite directions by these two forces.
Understanding which force will prevail is key to refining our understanding of \emph{r}/\emph{K}--selection.   

In all cases, if the carrying capacity is sufficiently large or, equivalently, if resources are sufficiently abundant, the fraction of time spent far from the theoretical carrying capacity is small. In fact, Theorem \ref{thm:convergenceToManifold} shows that as the carrying capacity grows to infinity, the (properly rescaled) population process converges to the solution of an SDE supported on the sub-manifold of admissible population states at carrying capacity. The SDE possesses a drift-diffusion term, resembling a Wright--Fisher diffusion with selection, and corresponding to the joint effect of demographic stochasticity and density constraints, which we find to accord with \cite{parsons2010some}, and a jump term that corresponds to the rare catastrophic events. 

Under the assumption of quasi-neutrality, the vital rates of each types are related in such a way that they can be parametrized by a single number \(\beta\). 
The larger \(\beta\) is, the faster an individual reproduces and dies. We do not explicitly model a process of mutation and fixation, but we appeal to an intuitive process whereby small changes to \(\beta\) arise as mutant subpopulations, and either vanish or become fixed.
Then it is natural to characterize the life-history speeds favored by selection in terms of risk dominance \cite{nowak2004emergence}.
As we discuss in Section \ref{sec:risk} a strategy A is {\em risk-dominant} over strategy B when A is more likely to successfully invade B (and go to fixation) than vice versa.
For a certain category of catastrophes, we find that there is at least one value of life-history speed \(\beta_{\mathrm{RDS}}\) that is attractive, in the sense that moves toward \(\beta_{\mathrm{RDS}}\) will be more likely to succeed than movements away.
The long-term result, loosely speaking, is a stochastic drift toward the RDS.

Since Katzenberger's foundational article \cite{katzenberger1991solutions}, the technique which consists in exploiting diverging time-scales to reduce the dimensionality of the model has been successfully applied to several problems of mathematical biology, such as the evolution of seed banks \cite{etheridge2026seed}, genealogies of branching processes \cite{pra2025multi}, effective population size for structured populations \cite{forien2025stochastic}, Fisher--KPP models \cite{etheridge2026fluctuations}, to cite the most recent preprints and publications. 
We are particularly interested in \cite{adeosun2026markov}, as it gives a template for extending Katzenberger's results to discontinuous limit processes.

\textbf{Outline of the paper.} In Section~\ref{sec:deterministicPrologue}, we define a fully deterministic model of birth, death, resource production and consumption, for two types sharing a single resource type. In addition to serving as gentle introduction to the more complex model of the following section, it provides a description of the quasi-neutrality regime by way of Theorem~\ref{thm:competitiveExclusionPrinciple}. 
In Section~\ref{sec:ModelDefinition}, assuming quasi-neutrality, we introduce the stochastic model that will be the focus of all following sections. 
In Section~\ref{sec:topology} we give a brief exposition of the finer technical points that will be necessary to state precisely Theorem~\ref{thm:convergenceToManifold}, our main result. In particular, we review the Meyer--Zheng topology. In Section~\ref{sec:mainResults} we state and prove Theorem~\ref{thm:convergenceToManifold}, showing that, under quasi-neutrality and in the limit of a large carrying capacity, the joint population process is supported on a one-dimensional manifold, and we give an expression for the generator of this limiting process. In Section~\ref{sec:fixationProbability}, we use the scale and speed methods for one-dimensional diffusions to find an approximation for the fixation probability of each type, in the form of Theorem \ref{thm:fixationProbability}. 
In Sections~\ref{sec:proof} and \ref{sec:fixationProof} we prove the technical results required for Sections~\ref{sec:mainResults} and~\ref{sec:fixationProbability}. Section~\ref{sec:discussion} contains our conclusions. 

Readers concerned primarily with the biological interpretation may be most interested in Sections \ref{sec:deterministicPrologue} and \ref{sec:ModelDefinition}; the definitions of the types of catastrophes (Assumption \ref{ass:survivalProbability}) and the main result on fixation probabilities (Theorem \ref{thm:fixationProbability}) in Section \ref{sec:fixationProbability}; and the interpretative discussion in Section \ref{sec:discussion}.

\textbf{Notation.} A complete mathematical glossary of the notation and definitions used can be found in Appendix \ref{sec:glossary}. We use \(\odot\) to denote component-wise multiplication of vectors, that is \((v_1,v_2) \odot (w_1, w_2) = (v_1w_1, v_2w_2)\). The absolute value of a two-dimensional vector \({\bm v}\) is understood to equal the taxicab or \(L^1\) norm \(|{\bm v}| = |v_1| + |v_2|\). The maximum between two numbers \(a\) and \(b\) is written as \(a \vee b\), and the minimum \(a \wedge b\). 

For subset \(V\) of \(\RR^d\), \(C^2(V)\) is the set of twice continuously differentiable functions \(V \mapsto \RR\), and \(C^2_b(V)\) is the set of those functions which are also bounded.

\section{Deterministic Model and Quasi-Neutrality}\label{sec:deterministicPrologue}

Consider the system of differential equations: 
\begin{equation}\label{eq:chemostatEquation}
    \begin{split}
        dX_1 &= X_1(\beta_1F(R)-\mu_1)dt,\\
        dX_2 &= X_2(\beta_2F(R)-\mu_2)dt,\\
        dR &= (\gamma G(R)-F(R){\bm \beta} \cdot {\bf X})dt. 
    \end{split}
\end{equation}
The quantities \({\bf X} = (X_1, X_2)\) correspond to the total biomass of types 1 and 2, respectively, while \(R\) is the available quantity of resource. The coefficient \(F(R)\) controls the rate at which a single unit of biomass of type 1, respectively of type 2, converts one unit of resource into one unit of biomass, given that the current amount of available resources is \(R\). Meanwhile, \(\mu_1\) (resp. \(\mu_2\)) gives the rate at which type 1 (resp. type 2) loses biomass, and \(\beta_1\) (resp. \(\beta_2\)) gives the maximal rate at which resources are converted into biomass. The function \(\gamma G\) gives the renewal rate of the resource. We shall assume the specific functional forms 
\begin{equation*}
    F(R) = \frac{R}{a + R}, \quad G(R) = R_{\max} - R. 
\end{equation*}
for some constants \(a\) and \(R_{\max}\), so that both types exploit resources according to a Holling's type II functional response \cite{holling1959components}, and that the maximum amount of resources that can be available for use at any time is \(R_{\max}\). Equation \eqref{eq:chemostatEquation} is strikingly similar to the dynamics of a chemostat where two organisms compete for a single growth-limiting resource \cite{young1970dynamic, hsu1977mathematical, butler1985mathematical}, the difference being that in our case the death rate of the organisms is uncoupled from the production rate of the resources. 

A central assumption throughout this article will be that \(\gamma\) is large relative to other parameters. Let us define \(\hat{\bf X} = (\hat{X}_1, \hat{X}_2) := \gamma^{-1}{\bf X}\). We then can rewrite \eqref{eq:chemostatEquation} as 
\begin{equation}\label{eq:chemostatEquationRescaled}
    \begin{split}
        d\hat {X}_1 &= \hat{X}_1(F(R) \beta_1 - \mu_1)dt,\\
        d\hat{X}_2 &= \hat{X}_2(F(R) \beta_2 - \mu_2)dt,\\
        dR &= \gamma [G(R)- F(R){\bm \beta} \cdot \hat{\bf X}]dt. 
    \end{split}
\end{equation}
When \(\gamma\) is large, \(R\) moves very quickly relative to \(\bf X\). Hence, we approximate \(R\) at any given time by the unique \(R^* \in (0,R_{\max})\) such that \(G(R^*)= F(R^*){\bm \beta} \cdot {\bf X}\) (this idea is made more rigorous in Theorem~\ref{thm:convergenceToManifold}). \(R^*\) is given by 
\begin{equation*}
    R^*(\mathcal{E}) := \frac{R_{\max} - \mathcal{E} - a}{2} + \frac{\sqrt{4aR_{\max} + (R_{\max} -a - \mathcal{E})^2}}{2}. 
\end{equation*}
where \(\mathcal{E} = {\bm \beta} \cdot {\hat{\bf X}}\). Inserting this quantity back into \eqref{eq:chemostatEquationRescaled} we obtain the reduced equation
\begin{equation}\label{eq:chemostatEquationReduced}
    \begin{split}
        d\hat{X}_1 &= \hat{X}_1\left(\beta_1F \circ R^*(\bm{\beta} \cdot \hat{\bf X}) -\mu_1\right)dt,\\
        d\hat{X}_2 &= \hat{X}_2\left(\beta_2F \circ R^*(\bm{\beta} \cdot \hat{\bf X}) -\mu_2\right)dt. 
    \end{split}
\end{equation}
We introduce the quantities \(\nu_1 = \beta_1/\mu_1\) and \(\nu_2 = \beta_2/\mu_2\), which we name \textit{intrinsic lifetime reproductive success} because they equal the expected number of offspring produced by an individual dying at rate \(\mu_1\) (resp. \(\mu_2\)) and giving birth at a rate \(\beta_1\) (resp. \(\beta_2\)). The long-term behavior of the types under \eqref{eq:chemostatEquationReduced} is now tractable, as made explicit in the following theorem: 
\begin{theorem}\label{thm:competitiveExclusionPrinciple}
    Suppose \(\nu_1 > \nu_2\) and let \(\hat{\bf X}\) be a solution to \eqref{eq:chemostatEquationReduced} such that \(\hat{\bf X}(0) \in \RR_{>0}^2\). Then: 
    \begin{enumerate}
        \item If \(\nu_1 < 1/F(R_{\max})\), then 
        \begin{equation*}
            \lim_{t \rightarrow \infty} \hat{\bf X}(t) = {\bf 0};
        \end{equation*}
        and
        \item If \(\nu_1 > 1/F(R_{\max})\), then 
        \begin{equation*}
            \lim_{t \rightarrow \infty} \hat{\bf X}(t) = (\mathcal{E}_{(1,0)}/\beta_1,0),
        \end{equation*}
        where \(\mathcal{E}_{(1,0)} = (F \circ R^*)^{-1}(1/\nu_1)\).
    \end{enumerate}
\end{theorem}
The proof of Theorem~\ref{thm:competitiveExclusionPrinciple} is provided in Appendix~\ref{sec:stability}. Theorem~\ref{thm:competitiveExclusionPrinciple} is in agreement with other similar results for chemostat equations found in~\cite{young1970dynamic, hsu1977mathematical, butler1985mathematical}, and is more broadly an instance of the resource-ratio hypothesis which says that, when multiple species compete for a single resource, the species that can maintain itself at the lowest resource level will come to dominate \cite{tilman1985resource}. 
Indeed, the point \(R^0_i = F^{-1}(1/\nu_i)\) is precisely the threshold where resources become too scarce to fuel the growth of types \(i\), and the condition \(\nu_1 > \nu_2\) is equivalent to \(R^0_1 < R_2^0\). 

The case \(\nu_1 = \nu_2 =: \nu\) eludes Theorem~\ref{thm:competitiveExclusionPrinciple}, despite occurring naturally under the empirically verified~\cite{marba2007allometric} allometric scaling
\begin{equation*}
    \beta \propto M^{-1/4}, \quad \mu \propto M^{-1/4},
\end{equation*}
where \(M\) denotes the adult mass of an individual. 
In such a case \(\nu\) becomes the ratio of the allometric prefactors. 
Substituting the condition \(\nu = \nu_1 = \nu_2\) into~\eqref{eq:chemostatEquationReduced}, one finds 
\begin{equation}\label{eq:chemostatEquationReducedDegenerate}
    \begin{split}
        dX_1 &= \beta_1 X_1\left(F \circ R^*({\bm \beta \cdot {\bf X}})-\frac{1}{\nu}\right)dt,\\
        dX_2 &= \beta_2X_2\left(F \circ R^*({\bm \beta \cdot {\bf X}})-\frac{1}{\nu}\right)dt.
    \end{split}
\end{equation}
Notice that~\eqref{eq:chemostatEquationReducedDegenerate} is stationary on the manifold \({\Gamma} = \{{\bm x} \in \RR_{\geq 0}^2 : \bm{\beta} \cdot {\bm x} = {\mathcal E}_{\Gamma}\}\) where the constant \({\mathcal E}_{\Gamma}\) is given by 
\begin{equation*}
    \mathcal{E}_{\Gamma} = (F \circ R^*)^{-1}(\nu^{-1}) = \nu R_{\max} - \frac{\nu}{\nu - 1}a.  
\end{equation*}
We see that this quantity is positive if \(F(R_{\max}) > 1/\nu\). When \(\beta_1 = \beta_2 = \beta\), \eqref{eq:chemostatEquationReduced} mirrors~\eqref{eq:logisticGrowth}, where the growth rate is \(r = \beta\) and the carrying capacity \(K = \beta^{-1}\mathcal{E}_{\Gamma}\). We thus see that \(\beta_1\) and \(\beta_2\) act as a parameters situating each type on the \emph{r}/\emph{K} spectrum. 

To ascertain more exactly \textit{where} on \(\Gamma\) the process will end up is a more delicate question, and the answer depends on the initial conditions of the system. 
Let \({\bf X}^{{\bm x}}\) be a solution of \eqref{eq:chemostatEquationReducedDegenerate} such that \({\bf X}^{\bm x}(0) = {\bm x}\), and define
\begin{equation}\label{eq:piLimit}
    {\bm \pi}({\bm x}) := \lim_{t \rightarrow \infty} {\bf X}^{\bm x}(t). 
\end{equation}
While \({\bm\pi}\) does not admit a closed form, it is implicitly defined by the following system of equations: 
\begin{equation}\label{eq:implicitSystem}
    \begin{split}
        \bm{\beta} \cdot \bm{\pi}({\bm x}) &= \mathcal{E}_{\Gamma},\\
        E(\bm{\pi}({\bm x})) &= E({\bm x}), 
    \end{split}
\end{equation}
where \(E({\bm x}) := x_2^{\beta_1}/x_1^{\beta_2}\), owing to the fact that \(\frac{d}{dt}E({\bf X}(t)) = 0\). 
If \(\beta_1 = \beta_2 = \beta\) it follows immediately from~\eqref{eq:implicitSystem} that \({\bm \pi}({\bm x}) = \frac{\mathcal{E}_{\Gamma}}{\beta}\frac{\bm x}{\|{\bm x}\|_1}\). Similarly, we can compute a first-order approximation when \(\beta_1\) and \(\beta_2\) are almost equal. Let \(\beta_2 = \beta_1 + \varepsilon = \beta + \varepsilon\), for some small \(\varepsilon > 0\) and suppose we can write 
\begin{equation}\label{eq:firstOrderProjectionExpansion}
    {\bm \pi}({\bm x}) = \frac{\mathcal{E}_{\Gamma}}{\beta} \frac{\bm x}{|{\bm x}|} + \varepsilon {\bf h}({\bm x}) + O(\varepsilon^2),
\end{equation} 
for a certain map \({\bf h}\). Then, by substituting~\eqref{eq:firstOrderProjectionExpansion} into~\eqref{eq:implicitSystem} and matching coefficients of first order we obtain the approximation 
\begin{equation*}
    {\bf h}({\bm x}) = \frac{\mathcal{E}_{\Gamma}}{\beta^2|{\bm x}|^2}\begin{pmatrix}
        x_1x_2 - x_1x_2\log\frac{\mathcal{E}_{\Gamma}}{\beta|{\bm x}|}  \\
        -x_2^2 - x_1x_2\log \frac{\mathcal{E}_{\Gamma}}{\beta|{\bm x}|}
    \end{pmatrix}. 
\end{equation*}
Define the frequency map \(\Phi({\bm x}) := \frac{x_1}{x_1 + x_2}\). By composing \(\Phi\) with~\eqref{eq:firstOrderProjectionExpansion} and expanding the resulting fraction to the first order, the limit frequency of type 1, i.e. \(\Phi \circ {\bm\pi}({\bm x})\), can be expressed as
\begin{equation}\label{eq:phiPiFirstOrderExpansion}
    \Phi \circ {\bm \pi}({\bm x}) = \Phi({\bm x}) - \varepsilon \beta^{-1}\Phi({\bm x})(1-\Phi({\bm x}))\log \frac{\mathcal{E}_{\Gamma}}{\beta|{\bm x}|} + O(\varepsilon^2).
\end{equation}
We see that the relative frequency of type 1 will have increased in the limit if and only if the initial population count is above the carrying capacity of \(K = \beta^{-1}\mathcal{E}_{\Gamma}\). The interpretation is very natural: a faster life history is beneficial when there are a lot of resources available to fuel rapid growth, but detrimental when there are few resources and populations are shrinking at a pace dictated by their intrinsic rate of mortality. 
This shows that a faster life history is beneficial when colonizing new territory, up to the point where resources become saturated. But this alone does not paint the full picture of how life history speed is selected.
In particular it does not address the crucial scenario where a small number of mutants invade a population that is already occupying an ecological niche and exploiting its resources. 

\section{Definition of the Model}\label{sec:ModelDefinition}

As in Section \ref{sec:deterministicPrologue}, we write \(R_t^{\gamma}\) for the quantity of available resources, while \({\bf X}^{\gamma} = (X_1^{\gamma}, X_2^{\gamma})\) keeps track of the number of individuals of each type. But unlike Section \ref{sec:deterministicPrologue}, we now assume that \(\{(R^{\gamma}_t, {\bf X}^{\gamma}_t)\}_{t\geq 0}\) is a continuous-time non-deterministic Markov process supported on \(E := [0,R_{\max}] \times \ZZ_{\geq 0}^2\).  

Let us write \({\bf G}^{\gamma}\) for the generator of \((R^{\gamma}, {\bf X}^{\gamma})\), acting on \(C^2(E)\). 
In our model, the evolution of \((R^{\gamma}, {\bf X}^{\gamma})\) can be decomposed into three components.

{\bf Continuous Resource Evolution.} The resource level \(R^{\gamma}\) evolves continuously according to the piecewise deterministic ODE
\begin{equation*}
    dR_t^{\gamma} = (\gamma G(R_t^{\gamma}) - F(R_t^{\gamma}){\bm{\beta} \cdot {\bf X}^{\gamma}_t})dt, 
\end{equation*}
tracking the renewal and depletion of available resources, where we recall \(F(r) = \frac{r}{a+r}\) and \(G(r) = R_{\max} - r\). The resulting contribution to \({\bf G}^{\gamma}\) is given by the operator \({\mathscr G}^{\gamma}_{\bf RES}\) defined as 
\begin{equation*}
    \mathscr{G}^{\gamma}_{\textbf{RES}} f(r, {\bm x}) := [{\gamma}G(r)- F(r) \bm{\beta} \cdot{\bm x}] \frac{\partial f (r, {\bm x})}{\partial r}, 
\end{equation*}
for all \(f \in C^2(E)\) and \((r, {\bm x}) \in E\).

{\bf Birth-Death Process.} Each individual of first type (respectively second type) dies independently of all other individuals at a rate \(\beta_1/\nu\) (resp. \(\beta_2/\nu\)). Assuming the resource level is currently \(r\), individuals of first type (respectively second type) produce offspring independently of all other individuals at a rate \(\beta_1F(r)\) (resp. \(\beta_2 F(r)\)). The resulting contribution to \({\bf G}^{\gamma}\) is given by the operator \({\mathscr G}^{\gamma}_{\bf DEM}\) defined as
\begin{align*}
    \mathscr{G}^{\gamma}_{\textbf{DEM}} f(r, {\bm x}) := \sum_{i = 1,2}\bigg( x_i\beta_iF(r)[f(r, {\bm x} + {\bf e}_i)-f(r, {\bm x}) ] + x_i\beta_i\nu^{-1}[ f(r, {\bm x} - {\bf e}_i)-f(r, {\bm x}) ] \bigg),  
\end{align*}
for all \(f \in C^2(E)\) and \((r, {\bm x}) \in E\).

{\bf Rare and Catastrophic Mortality Events.} Additionally, we suppose that at some rate \(\gamma^{-1}\kappa\) a rare and catastrophic \textit{mass-mortality event} occurs, which leads to the sudden death of a large number of individuals \textit{at the same time}. When such an event occurs, a pair of survival probabilities \({\bf U} = (U_1, U_2)\) is sampled from a probability law \(\Upsilon\) on \([0,1]^2\). Then, each individual of type 1 (type 2), independently of all other individuals, is killed off with probability \(1-U_1\) (probability \(1 - U_2\)). The resulting contribution to \({\bf G}^{\gamma}\) is given by the operator \({\mathscr G}^{\gamma}_{\bf CAT}\) defined as: 
\begin{equation*}
    \mathscr{G}^{\gamma}_{\textbf{CAT}} f(r, {\bm x}) := \frac{\kappa}{\gamma} \int_{[0,1]^2} \bigg(\EE{{\bm x},{\bf u}}{f(r, {\bf B})} - f(r, {\bm x}) \bigg)\Upsilon(d{\bf u})
\end{equation*}
for all \(f \in C^2(E)\), where \(\EE{{\bf m}, {\bf u}}{ \cdot}\) indicates expectation with respect to the law of two independent binomial random variables with parameters \({\bf m}\) and \(\bf u\), that is 
\begin{equation*}
    \EE{{\bf m}, {\bf u}}{g({\bf B})} = \sum_{i=0}^{m_1} \sum_{j=0}^{m_2} \binom{m_1}{i}\binom{m_2}{j}u_1^i(1-u_1)^{m_1-i}u_2^j(1-u_2)^{m_2-j}g(i,j). 
\end{equation*}

The generator of \((R^{\gamma}, {\bf X}^{\gamma})\) is therefore
\begin{equation*}
    {\bf G}^{\gamma} := \mathscr{G}^{\gamma}_{\textbf{RES}} + \mathscr{G}^{\gamma}_{\textbf{DEM}} + \mathscr{G}^{\gamma}_{\textbf{CAT}}.
\end{equation*}
We synthesize all jump dynamics in Table \ref{tab:modelTransitionRates}. To contain the probability that the population goes extinct, we introduce the following assumption: 
\begin{assumption}\label{ass:survivalProbability}
    The survival probability is bounded away from zero, in the sense that there exists \(a > 0\) such that:
    \begin{equation*}
        \Upsilon(\{{\bf u} \in [0,1]^2 : u_1 \wedge u_2 \leq a\}) = 0. 
    \end{equation*}
\end{assumption}
Consider the time-accelerated and rescaled process \((\hat{R}^{\gamma}_t, \hat{\bf X}_t^{\gamma}) := (R^{\gamma}_{\gamma t}, \gamma^{-1}\bf X^{\gamma}_{\gamma t})\), defined on \(\hat{E} = [0,R_{\max}] \times \RR_{\geq0}^2\).  
As \(\gamma\) increases, \(\hat{R}^{\gamma}\) is strongly pulled towards the equilibrium value \(R^*\); and when \(\hat{R}^{\gamma} = R^*\), a strong drift drives \(\hat{\bf X}^{\gamma}\) towards \(\Gamma\), the stable manifold for the reduced equation \eqref{eq:chemostatEquationReducedDegenerate}. Hence we may expect that \((\hat{R}^{\gamma}, \hat{\bf X}^{\gamma})\) converges to a stochastic process supported on \(\{R^*(\mathcal{E}_{\Gamma})\} \times \Gamma\) as \(\gamma \uparrow \infty\). This is made rigorous in Theorem \ref{thm:convergenceToManifold}. However, the full statement of Theorem \ref{thm:convergenceToManifold} cannot be given without first introducing the notion of convergence that we will be using, which is slightly weaker than the Skorokhod convergence that is more commonly used.

\begin{table}
    \centering
    \begin{tabular}{c c c}
    Transition & Rate & Ecological Meaning \\ \hline
    \({\bf X}  \rightarrow {\bf X} + {\bf e}_1\) & \(\beta_1 X_1 F(R^{\gamma}) \) & Birth of individual of first type\\
    \({\bf X}  \rightarrow {\bf X} + {\bf e}_2\) & \(\beta_2 X_2 F(R^{\gamma}) \) & Birth of individual of second type\\
    \({\bf X}  \rightarrow {\bf X} - {\bf e}_1\) & \(\beta_1X_1\nu^{-1} \) & Death of individual of first type\\
    \({\bf X}  \rightarrow {\bf X} - {\bf e}_2\) & \(\beta_2 X_2\nu^{-1} \) & Death of individual of second type\\
    \({\bf X}  \rightarrow \text{Binomial}({\bf X}, {\bf U})\) & \(\kappa \gamma^{-1}\) & Catastrophic event\\
    \end{tabular}
    \caption{Discontinuous transitions that can occur in the model and their ecological interpretation}\label{tab:modelTransitionRates}
\end{table}

\section{Background on the Skorokhod Topology and Pseudopaths}\label{sec:topology}

It is well-known that the space of c\`adl\`ag paths \(D_{\hat{E}}[0,T] = \{{\xi} : [0,T] \to \hat{E} \text{ such that } \xi \text{ is }\text{c\`adl\`ag}\}\) is a Polish space when endowed with the \(J_1\) Skorokhod metric: 
\begin{equation*}
    d^{J_1}(\xi, \eta) := \inf\{\varepsilon >0 : \exists \lambda \in \Lambda \text{ such that}\sup_{t \in [0,T]}|\lambda(t) - t| <\varepsilon \text{ and} \sup_{t \in [0,T]}|\xi \circ \lambda(t) - \eta(t)| <\varepsilon\}
\end{equation*}
where \(\Lambda\) is the set of continuous bijections \(\lambda : [0,T] \to [0,T]\), and \(\xi, \eta \in D_{\hat{E}}[0,T]\) (cf. \cite{billingsley2013convergence}, Chapter 3). When a sequence of paths \(\{\xi_n\}_{n \geq 1}\) converges to some limiting path \(\xi\) in this metric, we shall write 
\begin{equation*}
    \xi_n \xrightarrow{J_1} \xi.  
\end{equation*}
For a sequence of random variables \(\{\Xi_n\}_{n \geq 1}\) taking value in \(D_T\), the \(J_1\) topology induces a notion of weak convergence. We shall write 
\begin{equation*}
    \Xi_n \xRightarrow{J_1} \Xi
\end{equation*}
if \(\E{f(\Xi_n)} \rightarrow \E{f(\Xi)}\) for every bounded functional \(f : D_T \to \RR\) that is continuous in the \(J_1\) topology.

The \(J_1\) topology, however, is too fine for some of our purposes. Indeed, the limiting trajectories of the rescaled process \(\hat{\bf X}^{\gamma}\) have two--sided discontinuities at the times where large catastrophic jumps occur. 
Therefore limiting trajectories are not c\`adl\`ag with positive probability. 
Those two-sided discontinuities, however, are not so pervasive; they occupy a set of null Lebesgue measure. Therefore, rather than considering the convergence of \(\hat{\bf X}^{\gamma}\) as a sequence of trajectories, we consider the coarser problem of the convergence of the occupation measure \(dt \otimes \delta_{{\bf X}^{\gamma}(t)}\), since the set of those two--sided discontinuities is negligible for this measure. 

For a Polish space \(\Y\) we write \(\mathscr{P}(\Y)\) for the set of all probability measures on \([0,T] \times \Y\). 

For a sequence \(\{m_n\}_{n\geq 1}\subseteq \mathscr{P}(\Y)\) we write 
\begin{equation*}
    m_n \xrightarrow{\mathscr{P}(\Y)} m
\end{equation*}
when convergence occurs to the measure \(m\in \mathscr{P}(\Y)\). 
That is, for every bounded and continuous function \(f\), one has 
\begin{equation*}
    \lim_{n \uparrow \infty} \int_{[0,T] \times \Y} fdm_n = \int_{[0,T] \times \Y} f dm.
\end{equation*}
Likewise for a sequence of \(\mathscr{P}(\Y)\)-valued random variables \(\{M_n\}_{n \geq 1}\) and \(M\) a \(\mathscr{P}(\Y)\)-valued random variable, we write 
\begin{equation*}
    M_n \xRightarrow{\mathscr{P}(\Y)} M
\end{equation*}
whenever \(\E{F(M_n)} \rightarrow \E{F(M)}\) for all continuous functionals \(F : \mathscr{P}(\Y) \to \RR\).

We abbreviate \(\mathscr{P}:= \mathscr{P}(\hat E)\).

The Meyer--Zheng pseudopath space is the subset \({\bf \Psi} \subseteq \mathscr{P}\) such that \(\psi \in {\bf \Psi}\) if and only if there exists a Borel function \(\xi^{\psi} : [0,T] \to \hat E\) such that \(\psi = T^{-1}dt \otimes \delta_{\xi^{\psi}(t)}\). 
We say that \(\xi^{\psi}\) is a path associated to \(\psi\), and conversely we say that \(\psi\) is the occupation measure associated to \(\xi^{\psi}\).
Two Borel functions determine the same \(\psi\) precisely when they agree Lebesgue-almost everywhere, so \(\xi^\psi\) is determined only up to such equivalence; every statement involving \(\xi^\psi\) is invariant under modification on null sets.

There is an embedding of \(D_{\hat{E}}[0,T]\) into \(\bm\Psi\) but, as not every Borel function agrees Lebesgue-almost-everywhere with a c\`adl\`ag function, \(\bm \Psi\) is larger than \(D_{\hat{E}}[0,T]\).
For a c\`adl\`ag process \(\Xi\) on \([0,T]\) taking values in \(\hat E\) we write \(\Xi \xRightarrow{\mathscr{P}} M\) to mean that the associated occupation measures \(T^{-1} dt \otimes\delta_{\Xi(t)}\), which are \(\bm \Psi\)-valued random elements of \(\mathscr{P}\), converge weakly to \(M\).

The following proposition relates time-averaged statistics to the topology of uniform convergence.

\begin{proposition}\label{prop:ppConvergenceImpliesSkorConvergence}
    Let \(S : [0,T] \times \hat{E} \to \RR\) be a continuous and bounded function. Let \(\{M_n\}_{n \geq 1} \subseteq {\bf \Psi}\), \(M \in {\bf \Psi}\). Suppose that \(M_n \xRightarrow{\mathscr{P}} M\). Define \(I_n(t) = \int_{[0,t] \times \hat{E}} S({\bf z})M_n(d{\bf z})\) and \(I(t) =\int_{[0,t] \times \hat{E}} S({\bf z})M(d{\bf z})\). Then \(I_n\) converges weakly to \(I\) in the uniform topology. In particular \(I_n \xRightarrow{J_1} I\). 
\end{proposition}
\begin{proof}[Proof of Proposition \ref{prop:ppConvergenceImpliesSkorConvergence}]
    For \(\psi \in {\bf \Psi}\), write \(\xi^{\psi}\) for the associated path and define the map \(\phi : \Psi \to C^0([0,T])\), where \(C^0([0,T])\) is endowed with the uniform topology, as 
    \begin{equation*}
        \phi(\psi)(t) = \int_0^t S(s, \xi^{\psi}(s))ds. 
    \end{equation*}
    Note that \(I_n(t) = \phi(M_n)(t)\) and \(I(t) = \phi(M)(t)\). We claim that for any sequence such that \(\psi_n \xrightarrow{\mathscr{P}} \psi\), we have \(\phi(\psi_n) \xrightarrow{C^0} \phi(\psi)\). By the triangle inequality 
    \begin{align*}
        \sup_{0 \leq t \leq T}|\phi(\psi_n)(t) - \phi(\psi)(t)| &\leq \int_0^T|S(t, \xi^{\psi_n}(t)) - S(t,\xi^{\psi}(t))|dt.
    \end{align*}
    
    Fix any $\eta>0$.
    By Lusin's Theorem there is a a continuous $\tilde\xi:[0,T]\to \hat E$ with $\lambda\{t\, :\, \tilde\xi(t)\neq \xi^\psi(t)\}<\eta$.
    Set $\Theta(t,\bm y):= 1\wedge |\bm y - \tilde\xi(t)|$.
    Since $\Theta:[0,T]\times\hat E\to \RR$ is continuous and bounded by 1, the Continuous Mapping Theorem gives us
    \begin{align*}
    \limsup_{n\to\infty}\frac1T\int_0^T 1\wedge |\xi^{\psi_n}(t)-\xi^\psi(t)| dt & \le 
      \limsup_{n\to\infty}\frac1T\int_0^T 1\wedge |\xi^{\psi_n}(t)-\tilde\xi(t)| dt
      + \frac1T\int_0^T 1\wedge |\xi^{\psi}(t)-\tilde\xi(t)| dt\\
      &= \limsup_{n\to\infty }
        \int \Theta d\psi_n + \frac1T\int_0^T 1\wedge |\xi^{\psi}(t)-\tilde\xi(t)| dt\\
        &=
        \int \Theta d\psi + \frac1T\int_0^T 1\wedge |\xi^{\psi}(t)-\tilde\xi(t)| dt\\
        &=\frac2T\int_0^T 1\wedge |\xi^{\psi}(t)-\tilde\xi(t)| dt\\
        &\le \frac{2\eta}{T} .
    \end{align*}
    Since $\eta>0$ was arbitrary, it follows that $\frac1T\int_0^T 1\wedge |\xi^{\psi_n}(t)-\xi^\psi(t)| dt\to 0$.
    By the Riesz Subsequence Theorem \cite[Theorem 3.12]{rudin1970} there exists a subsequence $(n_k)$ such that
    \(\xi^{\psi_{n_k}}(t) -\xi^\psi(t)\) converges to 0 for almost every $t\in [0,T]$, hence \(S(\cdot, \xi^{\psi_{n_k}}(\cdot))\) converges almost everywhere to \(S(\cdot, \xi^{\psi}(\cdot))\). 
    As the \(\{S(\cdot, \xi^{\psi_n}(\cdot))\}_{n \geq 1}\) are uniformly bounded, dominated convergence implies 
    \begin{equation*}
        \lim_{k \uparrow \infty}\int_0^T|S(t, \xi^{\psi_{n_k}}(t)) - S(t,\xi^{\psi}(t))|dt = 0
    \end{equation*}
    along this subsequence. 
    As we could have used the same argument for any arbitrary subsequence, we conclude convergence of the sequence as a whole. 
    Hence \(\phi\) is continuous as a map from \({\bf \Psi}\) to \(C^0([0,T])\), and we can use the Continuous Mapping Theorem to conclude.
\end{proof}

\begin{proposition}
    \label{prop:J1convergenceimpliesppConvergence}
    If $\Xi^n \xRightarrow{J_1} \Xi$ in $D_{\hat E}[0,T]$, then the associated occupation measures converge, i.e.,
    \begin{equation}
        \label{eq:J1impliespp}
        T^{-1}dt \otimes \delta_{\Xi^n(t)}
          \xRightarrow{\mathscr{P}}
          T^{-1}dt \otimes \delta_{\Xi(t)}
    \end{equation}
\end{proposition}

\begin{proof}[Proof of Proposition \ref{prop:J1convergenceimpliesppConvergence}]
Consider the map $D_{\hat E}[0,T]\to \mathscr{P}$ taking $\xi$ to its occupation measure $T^{-1}dt\otimes \delta_{\xi(t)}$.
The result will follow from the Continuous Mapping Theorem once we show that this map is continuous.

If $\xi^n\to \xi$ in $J_1$ then there exist time-changes $\lambda^n:[0,T]\to[0,T]$ converging uniformly to the identity, such that $\xi^n\circ\lambda^n\to \xi$ uniformly.
Thus $\xi^n(t)\to \xi(t)$ at continuity points $t$ of $\xi$, hence Lebesgue-a.e., since a c\`adl\`ag path has at most countably many discontinuities.
Hence for any bounded continuous $F:[0,T]\times \hat E\to \RR$
\begin{equation*}
  \int_0^T F(t,\xi^n(t))dt \to \int_0^T F(t,\xi(t))dt,
\end{equation*}
by dominated convergence, meaning that the map is continuous.
\end{proof}

Furthermore, the description of compact sets for the topology we have imposed on pseudopaths is quite straightforward.
\begin{proposition}\label{prop:psiCompactSets}
    A set \(K \subseteq \mathscr{P}(\Y)\) is relatively compact in \(\mathscr{P}(\Y)\) if and only if for every positive \(\varepsilon\) there exists a compact set \(C_{\varepsilon} \subseteq \Y\) such that 
    \begin{equation*}
        \sup_{m \in K} m\bigl( [0,T] \times C_{\varepsilon}^\complement\bigr)\leq \varepsilon. 
    \end{equation*}
    In particular, \(K\subseteq \bm\Psi\) is relatively compact if and only if for every positive \(\varepsilon\) there exists \(C_\varepsilon\subseteq \hat E\) such that
    \begin{equation*}
        \sup_{\psi \in K} \lambda\{t \in [0,T] : \xi^{\psi}(t) \not\in C_{\varepsilon}\}\leq \varepsilon, 
    \end{equation*}
    where \(\lambda\) is the Lebesgue measure on \([0,T]\). 
\end{proposition}

\begin{proof}[Proof of Proposition \ref{prop:psiCompactSets}]
    By Prokhorov's theorem, \(K\) is relatively compact if and only if \(K\) is tight, which is exactly the statement of Proposition \ref{prop:psiCompactSets}.  
\end{proof}

Note that this condition is much weaker than the compactness condition found in \cite{meyer1984tightness}. The reason is that Meyer and Zheng ask for sets that are relatively compact in \(\bf\Psi\), such that any sequence admits a subsequence converging to a pseudopath. Here we merely require convergence to a probability measure.

\section{Convergence to a Limit Supported on \texorpdfstring{\(\Gamma\)}{\text{Gamma}}}\label{sec:mainResults}

Throughout this section and all following sections we assume that \(\nu > 1 / F(R_{\max})\) and that Assumption \ref{ass:survivalProbability} holds. We can now state our main result. 

\begin{theorem}\label{thm:convergenceToManifold}
    Fix a time interval \([0,T]\). Fix a family of initial states \(\{(R_0^{\gamma},\mathbf{X}_0^\gamma)\}_{\gamma>1} \subset \hat E\), and suppose \(\gamma^{-1}\mathbf{X}_0^\gamma \to \mathbf{x}_0 \in \mathbb{R}_{>0}^2\) and \(R^{\gamma}_0 \rightarrow r_0 \in [0,R_{\max}]\) as \(\gamma\to\infty\). Define the time-accelerated rescaled process
    {\((\hat{R}_t^{\gamma}, \hat{{\bf X}}_t^{\gamma}) = (R^{\gamma}_{\gamma t}, \gamma^{-1}{\bf X}^{\gamma}_{\gamma t})\)}. Also define the projected process \(\hat{\bf Z}^{\gamma}_{t} = {\bm \pi}(\hat{\bf X}_t^{\gamma})\). Then
    \begin{equation*}
        {\hat{\bf Z}}^{\gamma}\xRightarrow{J_1} \hat{\bf Z}, 
    \end{equation*}
    as \(\gamma \uparrow \infty\), where \(\{\hat{\bf Z}_t\}_{t \geq 0}\) is a c\`adl\`ag Markov process on \(\Gamma\) which solves the martingale problem associated to \(\{(f, {\bf L}f) : f \in C^2(\Gamma)\}\) with initial condition \(\hat{\mathbf{Z}}_0={\bm\pi}(\mathbf{x}_0)\), where \({\bf L}\) is the generator 
    \begin{align*}
        {\bf L} f({\bm z}) = \frac{1}{\nu}\sum_{i,j=1}^2\beta_iz_i \frac{\partial^2 \pi_j}{\partial x_i^2}\frac{\partial f}{\partial z_j}({\bm z}) + \frac{1}{\nu} &\sum_{i,j,k=1}^2\beta_i z_i \frac{\partial \pi_j}{\partial x_i}\frac{\partial \pi_k}{\partial x_i} \frac{\partial^2 f}{\partial z_k \partial z_j}({\bm z})\\
        &+ \kappa\int\displaylimits_{[0,1]^2} [f \circ {\bm \pi}({\bf u} \odot {\bm z}) - f \circ {\bm \pi}({\bm z})] \Upsilon(d{\bf u}).
    \end{align*}
    Furthermore
    \begin{equation*}
        (\hat{R}^{\gamma}, \hat{{\bf X}}^{\gamma}) \xRightarrow{\mathscr{P}} (R^*(\mathcal{E}_{\Gamma}), \hat{{\bf Z}}),
    \end{equation*}
    as \(\gamma \uparrow \infty\).  
\end{theorem}

Theorem \ref{thm:convergenceToManifold} holds thanks to a few key properties of the pre-limits \(\hat{R}^{\gamma}, \hat{\bf X}^{\gamma}\) and \(\hat{\bf Z}^{\gamma}\), which are collected in Lemma \ref{lem:keystoneLemma}. 

\begin{lemma}\label{lem:keystoneLemma}
    The following key properties are satisfied: 
    \begin{enumerate}[label= (A\arabic*)]
        \item\label{it:decomposition} The generator \(\hat{\bf G}^{\gamma}\) of \((\hat{R}^{\gamma}, \hat{\bf X}^{\gamma})\) can be decomposed as 
        \begin{equation*}
            \hat{\bf G}^{\gamma} = \gamma^2 \hat{\bf G}^{\gamma}_2 + \gamma\hat{\bf G}^{\gamma}_1 + \hat{\bf G}_0^{\gamma}, 
        \end{equation*}
        where for each \(i \in \{0,1,2\}\) the operators \(\hat{\bf G}^{\gamma}_i : C^2(\hat E) \to C^0(\hat E)\) satisfy the following continuity property: if \(M_{\gamma} \xRightarrow{\mathscr{P}} M\) as \(\gamma \uparrow \infty\), then for all $f\in C^2(\hat E)$ the pushforwards of $M_\gamma$ under $(t,(r,\bm x))\mapsto (t, (\hat{\bf G}^{\gamma}_i f)(r,\bm x))$ converge in $\mathscr{P}(\RR)$ to the corresponding pushforward of $M$.
Furthermore, this decomposition is explicitly given by \eqref{eq:generatorDecompositionFormulae} and \eqref{eq:gZeroFormula};
        \item\label{it:piCancellation} Let \({\bm \pi}\) be defined as in \eqref{eq:piLimit}. Then \(\bm \pi \in C^2(\hat{E} \setminus \{{\bf 0}\})\) and
        \begin{equation*}
            \hat{\bf G}^{\gamma}(h \circ {\bm \pi}) \equiv \hat{\bf G}^{\gamma}_0(h \circ {\bm \pi}), 
        \end{equation*}
        for all \(h \in C^2(\Gamma)\); 
        \item\label{it:psiTightness} The sequence \(\{(\hat R^{\gamma}, \hat{\bf X}^{\gamma})\}_{\gamma > 1}\) is tight in $\mathscr{P}$; 
        \item\label{it:nonExtinction} There exists \(m_0 > 0\) such that 
        \begin{equation*}
            \lim_{\gamma \uparrow \infty} \P{\inf_{t \in [0,T]}|\hat{\bf X}^{\gamma}_t| < m_0} = 0;
        \end{equation*} and
        \item\label{it:skorokhodTightness} The sequence \(\{\hat{\bf Z}^{\gamma}\}_{\gamma > 1}\) is tight for the \(J_1\) topology. 
        
    \end{enumerate}
\end{lemma}

\begin{proof}[Proof of Theorem \ref{thm:convergenceToManifold}]

For all \(f \in C_b^2(\hat{E})\), define the martingale
\begin{equation*}
    \{A_t^f\}_{t \geq 0} := \bigg\{f(\hat{R}^{\gamma}_t, \hat{{\bf X}}^{\gamma}_t) - \int_0^t (\hat{\bf G}^{\gamma} f) (\hat{R}_s^{\gamma}, \hat{{\bf X}}_s^{\gamma}) ds\bigg\}_{t \geq 0}.
\end{equation*}
By Prokhorov's theorem and \ref{it:psiTightness}, there exists a weak limit \(M\) to \((\hat{R}^{\gamma}, \hat{\bf X}^{\gamma})\) in \(\mathscr{P}\) (passing to a subsequence if necessary), which \textit{a priori} is a random probability measure that may or may not have an occupation measure representation. Then on the one hand 
\begin{equation*}
    \E{\gamma^{-2}A_t^f} \rightarrow \E{-\int_{[0,t] \times \hat{E}} (\hat{\bf G}_2f)dM},
\end{equation*}
by \ref{it:decomposition} (as the functional \(H(m) = \int_{[0,T] \times \hat{E}} (\hat{\bf G}_2f)dm\) is continuous and bounded). 
On the other hand using the Martingale property 
\begin{equation*}
    \E{\gamma^{-2}A^f_t} = \gamma^{-2} f(\hat R_0^\gamma, \hat{\bf X}_0^\gamma) \rightarrow 0, 
\end{equation*}
as \(f\) is bounded.
Arguing similarly for \(\gamma^{-1}A_t^g\) where \(g \in C_b^2(\hat{E})\) is constant in the first variable, we find that 
\[
\E{\int_{[0,t] \times \hat{E}}(\hat{\bf G}_1g)dM} = 0.
\] 
Since the set of continuous and bounded functions is a measure-determining class, using \eqref{eq:generatorDecompositionFormulae} this implies \(M\) is almost surely supported on \([0,T]\times\Sigma \subseteq [0,T]\times \hat{E}\) defined as
\begin{equation}\label{eq:nullSetSystem}
    \Sigma = \{(r, {\bm x}) \in \hat{E} : 
        G(r) - F(r) {\bm \beta} \cdot {\bm x} = 0 \text{ and }
        {\bm \beta} \cdot {\bm x}(F(r) - \nu^{-1}) = 0\}.   
\end{equation}
Using \ref{it:nonExtinction} we can eliminate \({\bm x} = {\bm 0}\) from the support of \(M\). Therefore from solving \eqref{eq:nullSetSystem} we can conclude: 
\begin{equation*}
    \supp(M) \subseteq [0,T] \times \{{R^*(\mathcal{E}}_{\Gamma})\} \times \Gamma. 
\end{equation*}
In particular, \({\bm \pi}_*M = M\), where \({\bm \pi}_*M\) is the pushforward measure of \(M\) by \({\bm \pi}\).

By \ref{it:skorokhodTightness}, there exists a c\`adl\`ag \(\hat{\bf Z}\) such that \(\hat{\bf Z}^{\gamma} \xRightarrow{J_1} \hat{\bf Z}\) (passing to a further sublimit if necessary). 
Further, \({\bm \pi}\) is continuous on $\hat E\setminus \{\bm 0\}$; hence the pushforward $f_*$ by $f(t,(r,\bm x)):= (t,\bm\pi(\bm x))$ is continuous at every measure that assigns full mass to $\hat E\setminus \{\bm 0\}$, a condition that applies to almost every $M$.
The Continuous Mapping Theorem then gives immediately   \(\bm\pi\circ \hat{\bf X}^\gamma \xRightarrow{\mathscr{P}} \bm \pi_* M\) in the occupation-measure sense of Section \ref{sec:topology}.
Consider now the diagram below:
\[
\begin{tikzcd}[row sep=large, column sep=large]
{\bm \pi} \circ \hat{\bf X}^{\gamma}
  \arrow[r, Rightarrow, "\mathscr{P}"]
& {\bm \pi}_* M
  \arrow[r, equal, "(\ast)"']
& M
  \arrow[d, dashed, no head, "?"]
\\
\hat{\bf Z}^{\gamma}
  \arrow[u, equal, "\text{def.}"]
  \arrow[rr, Rightarrow, "J_1"']
&
& T^{-1} dt \otimes \delta_{(R^*(\mathcal{E}_{\Gamma}),\, \hat{\bf Z}_t)}
\end{tikzcd}
\]
Here equalities are to be understood in the sense of the occupation measure formulation. 
By Proposition \ref{prop:J1convergenceimpliesppConvergence} convergence in $J_1$ implies convergence of the associated occupation measures. 
Thus $\hat{\bf Z}^\gamma$ converges in $\mathscr{P}$ both to $M$ in the top row, and to $dt \otimes \delta_{(R^*(\mathcal{E}_{\Gamma}),\, \hat{\bf Z}_t)}$. 
By uniqueness of weak limits, the dashed line is an equality; namely, $M$ has the occupation-measure representation $T^{-1}dt \otimes \delta_{(R^*(\mathcal{E}_{\Gamma}),\, \hat{\bf Z}_t)}$.

Let \(h \in C^2(\Gamma)\). Picking the stopping time \({\tau^{\gamma}: = \inf\{t \geq 0 : |\hat{\bf X}^{\gamma}_t| < m_0\} }\), where \(m_0\) is chosen as in \ref{it:nonExtinction}, we have that  
\begin{equation*}
    \{A^{h,\gamma}_{t \wedge \tau^{\gamma}}\}_{t \geq 0} := \bigg\{h \circ {\bm \pi}(\hat{{\bf X}}^{\gamma}_{t\wedge \tau^{\gamma}}) - \int_0^{t\wedge \tau^{\gamma}} \hat{\bf G}^{\gamma}_0 (h \circ {\bm \pi})(\hat{R}^{\gamma}_s,\hat{{\bf X}}^{\gamma}_s)ds\bigg\}_{t \geq 0} 
\end{equation*}
is a martingale. Using \ref{it:piCancellation} we can rewrite 
\begin{equation*}
    A^{h,\gamma}_{t \wedge \tau^{\gamma}}= h(\hat{\bf Z}^{\gamma}_{t \wedge \tau^{\gamma}}) - \int_0^{t \wedge \tau^{\gamma}} \hat{\bf G}^{\gamma}_0(h \circ {\bm \pi})(\hat{R}^{\gamma}_s, \hat{\bf X}^{\gamma}_s)ds.
\end{equation*}
By \ref{it:nonExtinction} we have \(\tau^{\gamma} \wedge T \rightarrow T\) a.s. as \(\gamma \uparrow \infty\). As \(\hat{\bf Z}^{\gamma} \xRightarrow{J_1} \hat{\bf Z}\), the first term converges to \(h(\hat{\bf Z})\) in \(J_1\).
Furthermore, as we have just shown that \((\hat{R}^{\gamma}, \hat{\bf X}^{\gamma}) \xRightarrow{\mathscr{P}} (R^*(\mathcal{E}_{\Gamma}), \hat{\bf Z})\), using~\ref{it:decomposition} and Proposition~\ref{prop:ppConvergenceImpliesSkorConvergence}, the second term converges to \(-\int_0^{\cdot} \hat{\bf G}_0(h \circ {\bm \pi})(R^*(\mathcal{E}_{\Gamma}), \hat{\bf Z}_s)ds\) in \(J_1\). 
As \(\hat{\bf G}^{\gamma}_0(h \circ \bm{\pi})\) is uniformly bounded on \(\{{\bm x} \in \hat{E} : |{\bm x}| \geq m_0\}\), by Proposition~9.1.1 of~\cite{jacod2013limit}, the weak limit of \(\{A^{h, \gamma}_{t \wedge \tau^{\gamma}}\}_{t \geq 0}\) as \(\gamma \uparrow \infty\) must also be a martingale, i.e., \(\{h(\hat{\bf Z}_{t}) - \int_0^t \hat{\bf G}(h \circ {\bm \pi})(R^*(\mathcal{E}_{\Gamma}), \hat{\bf Z}_s)ds\}_{t \geq 0}\) is a martingale. 
Hence we obtain a characterization of \(\hat{\bf Z}\) as the solution to the martingale problem associated to \({\bf L}\), where \({\bf L}\) satisfies 
\begin{equation}\label{eq:generatorFormula}
    {\bf L} f({\bm z}) = \hat{\bf G}_0(f \circ {\bm \pi})(R^*(\mathcal{E}_{\Gamma}), {\bm z}). 
\end{equation}
The full statement of Theorem \ref{thm:convergenceToManifold} follows from expanding \eqref{eq:generatorFormula} according to \eqref{eq:gZeroFormula}. 
\end{proof}

\section{Fixation Probability}\label{sec:fixationProbability}

Now that we have established convergence of \(\hat{\bf X}^{\gamma}\) to \(\hat{\bf Z}\), we dedicate this section to the study of this limiting process \(\hat{\bf Z}\). Let us choose a parametrization of \(\Gamma\) that makes evident the fact that \(\hat{\bf Z}\) is a one-dimensional process. 

\begin{corollary}\label{cor:frequencyProcess}
    Define the frequency map \(\Phi : \Gamma \to [0,1]\) as 
    \begin{equation*}
        \Phi({\bm z}) := \frac{z_1}{z_1 + z_2}. 
    \end{equation*}
    It is a bijection \(\Gamma \leftrightarrow [0,1]\), with inverse 
    \begin{equation*}
        \Phi^{-1}(w) = \frac{\mathcal{E}_{\Gamma}}{\beta_1w + \beta_2(1-w)}(w,1-w).
    \end{equation*}
    The frequency process \(\{W_t\}_{t\geq0} :=\{\Phi(\hat{\bf Z})\}_{t \geq 0}\) is a Markov process with generator 
    \begin{equation}\label{eq:relativeAbundanceGenerator}
        {\bf A} f(w) = \mu(w)f'(w) + \frac{\sigma(w)^2}{2}f''(w) + \kappa\int\displaylimits_{[0,1]^2} \bigg[f \circ \Phi \circ {\bm \pi}({\bf u} \odot \Phi^{-1}(w)) - f(w)\bigg]\Upsilon(d{\bf u}),
    \end{equation}
    where 
    \begin{align*}
    \mu(w) &= \frac{1}{\nu {\mathcal{E}_{\Gamma}}}\frac{\beta_1\beta_2(\beta_2-\beta_1)w(1-w)(\beta_1w + \beta_2(1-w))^2}{(\beta_1^2w + \beta_2^2(1-w))^3}Q(w)\\
    Q(w)
    &= \beta_2^2(2\beta_1-\beta_2)(1-w) + \beta_1^2(2\beta_2-\beta_1)w + 2(\beta_1 + \beta_2)(\beta_1 - \beta_2)^2w(1-w)\\
    \sigma(w)^2 &= \frac{2}{\nu {\mathcal{E}_{\Gamma}}}\frac{\beta_1 \beta_2w(1-w)(\beta_1w + \beta_2(1-w))^4}{(\beta_1^2w + \beta_2^2(1-w))^2}.
    \end{align*}
\end{corollary}

\begin{proof}[Proof of Corollary \ref{cor:frequencyProcess}]
    This a straightforward consequence of Theorem \ref{thm:convergenceToManifold}, by applying \({\bf L}\) to \(f \circ \Phi\). The intermediate steps are detailed in Section~\ref{sec:intermediateGenerator}. The only difficulty is computing the first and second order partial derivatives of \({\bm \pi}\) on \(\Gamma\). But the conservation law given in \eqref{eq:implicitSystem} enables us to evaluate these derivatives through implicit differentiation. We do so in Proposition~\ref{prop:firstDerivatives} (see more specifically Corollary~\ref{cor:gradientOnManifold}).
\end{proof}

\begin{assumption}
    We consider two possibilities for the survival probability distribution:
    \begin{enumerate}[label= (H\arabic*)]
        \item\label{it:diagonalDistribution} \(\Upsilon\) is such that \({\bf U} = (1,1)U^*\), for some random variable \(U^*\) with law \(\upsilon_*\) on \([a,1]\), where \(a > 0\) is some constant;
        \item\label{it:proportionalDistribution} \(\Upsilon\) is such that \({\bf U} = (e^{-T^*\beta_1},e^{-T^*\beta_2})\), for some random variable \(T^*\) with law \(\upsilon_*\) on \([0,b]\), where \(b < \infty\) is some constant.
    \end{enumerate}
\end{assumption}

    Assumption \ref{it:diagonalDistribution} formalizes a fire- or flood-type catastrophe, that annihilates a significant portion of the population in a short time, without respect to its life-history: both types are affected indiscriminately. 
    Assumption \ref{it:proportionalDistribution}, on the other hand, could represent a drought or famine that places pressure on the population over a period of duration \(\nu T^*\), so that reproduction is halted and mortality continues as usual:
    then the catastrophe mortality is correlated to the ordinary demographic mortality.

    To measure the evolutionary success of each type in this setting, we introduce the fixation probability \(\theta(w) = \P{\tau_1 < \tau_0 \:|\: W_0 = w}\), where we use the absorption times \(\tau_i := \inf\{t \geq 0 : W_t = i\}\) for \(i \in \{0,1\}\). We will be particularly interested in the case where \(\beta_1 = \beta\) and \(\beta_2 = \beta + \varepsilon\), for a small quantity \(\varepsilon\) which may be positive or negative, to represent invasion by a mutant that is only slightly different from the resident population. 

    The case of \ref{it:proportionalDistribution} is the simpler of the two. First note that \({\bm\pi}(e^{-t\beta_1}x_1, e^{-t\beta_2}x_2) = {\bm \pi}(x_1, x_2)\),  since \(E(e^{-t\beta_1}x_1, e^{-t\beta_2}x_2) = E(x_1, x_2)\), for any \(t\). Hence the jump term in \eqref{eq:relativeAbundanceGenerator} completely vanishes. What we are left with is a one-dimensional diffusion to which we can apply the standard speed and scale theory~\cite{ewens2012mathematical}, and we find
    \begin{equation}\label{eq:H2FixationProbability}
        \theta(w) = \frac{\int_0^w \exp\bigg(-\int_0^y \frac{2\mu(z)}{\sigma(z)^2}dz \bigg) dy}{\int_0^1 \exp\bigg(-\int_0^y \frac{2\mu(z)}{\sigma(z)^2}dz \bigg) dy}
        = w + \frac{1}{2\beta}\varepsilon w(1-w) + o(\varepsilon). 
    \end{equation}
    This approximation agrees exactly with results of Parsons, Quince and Plotkin \cite[Table 1]{parsons2010some}. 
    
    For the case of \ref{it:diagonalDistribution}, we can estimate the fixation probability when \(\beta_1\) and \(\beta_2\) differ by a small quantity by linearizing the jump term, as stated in Theorem~\ref{thm:fixationProbability} below. 

\begin{theorem}\label{thm:fixationProbability}
    Assume \ref{it:diagonalDistribution} holds. Suppose \(\beta_1 = \beta\) and \(\beta_2 = \beta + \varepsilon\), where \(\varepsilon\) may be positive or negative. Then
        \begin{equation*}
            \theta(w) = w + \varepsilon w(1-w)\bigg[\frac{1}{2\beta} - \frac{\kappa\nu\mathcal{E}_{\Gamma}}{2\beta^3}\E{\log \frac{1}{U^*}}\bigg] + o(\varepsilon),
    \end{equation*}
    uniformly over \(w \in [0,1]\).
\end{theorem}

To prove Theorem~\ref{thm:fixationProbability} we require Lemma~\ref{lem:fixationProbabilityApproximation} below.  

\begin{lemma}\label{lem:fixationProbabilityApproximation}
    Assume \ref{it:diagonalDistribution}. To emphasize the dependence in \(\varepsilon\) we will write \({\bf A}\) as \({\bf A}^{\varepsilon}\). Then:  
    \begin{enumerate}[label= (B\arabic*)]
        \item\label{it:operatorDifferenceEstimate} Introduce the linearization around \(\varepsilon = 0\) 
        \begin{equation*}
            \tilde{\bf A}^{\varepsilon} f(w):= w(1-w)\bigg[a_{1,2}f''(w) + \varepsilon a_{2,1}f'(w) + \varepsilon a_{2,2}f''(w)\bigg]
        \end{equation*}
        where 
        \begin{align*}
            a_{1,2} &= \frac{\beta^2}{\nu{\mathcal{E}_{\Gamma}}},\\
            a_{2,1} &= \frac{\beta}{\nu {\mathcal E}_{\Gamma}} - \frac{\kappa}{\beta}\E{\log \frac{1}{U^*}},\\
            a_{2,2} &= \frac{\beta}{\nu{\mathcal E}_{\Gamma}}.
        \end{align*}
        Then there exists a constant \(C\) such that 
        \begin{equation*}
            \int_0^1 \frac{|{\bf A}^{\varepsilon}f(w) - \tilde{\bf A}^{\varepsilon}f(w)|}{w(1-w)}dw \leq \varepsilon^2C\int_0^1|f'(w)|+|f''(w)|dw,
        \end{equation*}
        for all \(f \in C^2[0,1]\); 
        \item\label{it:stabilityBound} There exists a constant \(C\) such that for all \(f \in C^2([0,1])\) such that \(f(0)=f(1)=0\), one has  
        \begin{equation*}
            \sup_{w \in [0,1]} |f(w)| \leq  C\int_0^1 \frac{|\tilde{\bf A}^{\varepsilon}f(w)|}{w(1-w)}dw; 
        \end{equation*}
        \item\label{it:aPrioriEstimate} The fixation probability \(\theta(w)\) is increasing in \(w\). 
    \end{enumerate}
\end{lemma}

\begin{proof}[Proof of Theorem \ref{thm:fixationProbability}]
    To emphasize the dependence in \(\varepsilon\) we will write \({\bf A}\) as \({\bf A}^{\varepsilon}\) and \(\theta\) as \(\theta^{\varepsilon}\). Notice \(\{\theta^{\varepsilon}(W_t)\}_{t \geq 0}\) is a martingale (to see this write \(\theta^{\varepsilon}(W_t) = \E{\mathds{1}_{\tau_1 < \tau_0} \:|\: \mathcal{F}_t}\) where \(\{\mathcal{F}_t\}_{t \geq 0}\) is the natural filtration of \(\{W_t\}_{t \geq 0}\)). Therefore \(\{\int_0^t {\bf A}^{\varepsilon} \theta^{\varepsilon}(W_s)ds\}_{t \geq 0}\) is also a martingale. As it is a continuous martingale with finite variation, it must be constant. Hence \({\bf A}^{\varepsilon}\theta^{\varepsilon}(W_t) = 0\) for all \(t \geq 0\). As \(W_t\) visits any point of \([0,1]\) with positive probability, we conclude that the fixation probability must solve the integro-differential equation 
    \begin{equation*}
            \begin{cases}
            {\bf A}^{\varepsilon}\theta^{\varepsilon} \equiv 0,\\
            \theta^{\varepsilon}(0) = 0,\\
            \theta^{\varepsilon}(1) = 1. 
        \end{cases}
    \end{equation*}
    Similarly, the fixation probability \(\tilde{\theta}^{\varepsilon}\) associated to the linearized generator solves the ordinary differential equation \begin{equation}\label{eq:fixationProbabilityLinearizedSystem}
        \begin{cases}
        \tilde{\bf A}^{\varepsilon}\tilde{\theta}^{\varepsilon} \equiv 0,\\
        \tilde{\theta}^{\varepsilon}(0) = 0,\\
        \tilde{\theta}^{\varepsilon}(1) = 1, 
    \end{cases}
    \end{equation}
    where \(\tilde{\bf A}^{\varepsilon}\) is the differential operator defined in Lemma \ref{lem:fixationProbabilityApproximation}. As the linearized operator \(\tilde{\bf A}^{\varepsilon}\) has no jump component, \eqref{eq:fixationProbabilityLinearizedSystem} can be solved explicitly and we find
    \begin{equation*}
        \tilde{\theta}^{\varepsilon}(w) = \frac{\int_0^w \psi^{\varepsilon}(z)dz}{\int_0^1 \psi^{\varepsilon}(z)dz},
    \end{equation*}
    where
    \begin{equation*}
        \psi^{\varepsilon}(z) := \exp\left(- \frac{\varepsilon a_{2,1}}{a_{1,2} + \varepsilon a_{2,2}} z\right). 
    \end{equation*}
    And so 
    \begin{equation*}
        \tilde{\theta}^{\varepsilon}(w) = w + \varepsilon\frac{a_{2,1}}{2a_{1,2}}w(1-w) + O(\varepsilon^2). 
    \end{equation*}
    Define \(h(w) = \tilde{\theta}^{\varepsilon}(w) - \theta^{\varepsilon}(w)\). Notice it solves 
    \begin{equation*}
        \begin{cases}
            \tilde{\bf A}^{\varepsilon}h \equiv (\tilde{\bf A}^{\varepsilon} - {\bf A}^{\varepsilon}){\theta}^{\varepsilon},\\
            h(0) = h(1) = 0. 
        \end{cases}
    \end{equation*}
    Hence, using \ref{it:stabilityBound}, there exists a constant \(C\) such that  
    \begin{align*}
        \sup_{w \in [0,1]}|h(w)| &\leq C\int_0^1\frac{|\tilde{\bf A}^{\varepsilon}h(w)|}{w(1-w)}dw\\
        &= C \int_0^1 \frac{|(\tilde{\bf A}^{\varepsilon} - {\bf{A}}^{\varepsilon})\theta^{\varepsilon}(w)|}{w(1-w)}dw.
    \end{align*}
    Then, using \ref{it:operatorDifferenceEstimate}, there exists another constant \(C\) such that 
    \begin{equation*}
        \int_{0}^1 \frac{|(\tilde{\bf A}^{\varepsilon} - {\bf{A}}^{\varepsilon})\theta^{\varepsilon}(w)|}{w(1-w)}dw \leq C \varepsilon^2\int_0^1 |(\theta^{\varepsilon})''(w)| + |(\theta^{\varepsilon})'(w)|dw.
    \end{equation*}
    We can write 
    \begin{equation*}
        (a_{1,2} + \varepsilon a_{2,2})(\theta^{\varepsilon})''(w) = -\varepsilon a_{2,1}(\theta^{\varepsilon})'(w) + \frac{(\tilde{\bf A}^{\varepsilon} - {\bf A}^{\varepsilon})\theta^{\varepsilon}(w)}{w(1-w)}.
    \end{equation*}
    Therefore, using \ref{it:operatorDifferenceEstimate} once again 
    \begin{align*}
        |a_{1,2} + \varepsilon a_{2,2}|\int_0^1 |(\theta^{\varepsilon})''(w)|dw &\leq \varepsilon |a_{2,1}|\int_0^1 |(\theta^{\varepsilon})'(w)|dw + \int_{0}^1 \frac{|(\tilde{\bf A}^{\varepsilon} - {\bf{A}}^{\varepsilon})\theta^{\varepsilon}(w)|}{w(1-w)}dw\\
        &\leq \varepsilon|a_{2,1}|\int_0^1 |(\theta^{\varepsilon})'(w)|dw + \varepsilon^2 C\int_0^1 |(\theta^{\varepsilon})''(w)|dw.
    \end{align*}
    Hence by rearranging, for \(\varepsilon\) small enough, there exists yet another constant \(C\) such that 
    \begin{equation*}
        \int_0^1 |(\theta^{\varepsilon})''(w)|dw \leq C\int_0^1 |(\theta^{\varepsilon})'(w)|dw. 
    \end{equation*}
    By \ref{it:aPrioriEstimate}, \(\theta^{\varepsilon}\) is increasing and so 
    \begin{equation*}
        \int_0^1 |(\theta^{\varepsilon})'(w)|dw = \theta^{\varepsilon}(1) - \theta^{\varepsilon}(0) = 1.
    \end{equation*}
    for all \(\varepsilon\) small enough. And so 
    \begin{equation*}
        \theta^{\varepsilon}(w) = \tilde{\theta}^{\varepsilon}(w) + O(\varepsilon^2),
    \end{equation*}
    uniformly in \(w \in [0,1]\), which is the desired conclusion. 
\end{proof}

\section{Proof of Lemma \ref{lem:keystoneLemma}}\label{sec:proof}

According to the description of the process \((R^{\gamma}, {\bf X}^{\gamma})\) given in Section \ref{sec:ModelDefinition}, the generator of the time-accelerated processes \((\hat{R}^{\gamma}, \hat{{\bf X}}^{\gamma})\) can be decomposed as the sum of three operators acting on \(C^2(\hat{E})\): 
\begin{equation*}
    \hat{\bf G}^{\gamma} = \hat{\mathscr G}^{\gamma}_{\textbf{RES}} + \hat{\mathscr G}^{\gamma}_{\textbf{DEM}} + \hat{\mathscr G}^{\gamma}_{\textbf{CAT}},  
\end{equation*}
such that 
\begin{align*}
    \hat{\mathscr{G}}^{\gamma}_{\textbf{RES}} f(r, {\bm x}) &:= {\gamma}^2(G(r)- F(r) \bm{\beta} \cdot{\bm x}) \frac{\partial f (r, {\bm x})}{\partial r},\\
    \hat{\mathscr{G}}^{\gamma}_{\textbf{DEM}} f(r, {\bm x}) &:= \gamma^2\sum_{i = 1,2}\bigg( x_i\beta_iF(r)[f(r, {\bm x} + \gamma^{-1}{\bf e}_i)-f(r, {\bm x}) ]\\ 
    &+ x_i\beta_i\nu^{-1}[ f(r, {\bm x} - \gamma^{-1}{\bf e}_i)-f(r, {\bm x}) ] \bigg),\\
    \hat{\mathscr{G}}^{\gamma}_{\textbf{CAT}} f(r, {\bm x}) &:= \kappa \int_{[0,1]^2} \bigg(\EE{\lfloor \gamma{\bm x}\rfloor,{\bf u}}{f(r, \gamma^{-1}{\bf B})} - f(r, {\bm x}) \bigg)\Upsilon(d{\bf u})
\end{align*}
for all \(f \in C^2(\hat{E})\).

In this section we will need at some points to emphasize the initial conditions.
For \((r,{\bm x})\in \hat E\) we write \(\mathbb{P}_{(r,{\bm x})}\) and \(\mathbb{E}_{(r,{\bm x})}\) for probability and expectation when the process in question is started from \((r,{\bm x})\). 
Both rescaled processes are well defined from any such starting point (after replacing all instances of \(x_1\) and \(x_2\) with \(x_1 \vee 0\) and \(x_2 \vee 0\) in the above generator).

\subsection{Proof of \ref{it:decomposition} and \ref{it:piCancellation}}

The claim is that the generator naturally separates into three timescales with respect to \(\gamma\), that is one may write 
\begin{equation*}
    \hat{\bf G}^{\gamma} = \gamma^2\hat{\bf G}_2 + \gamma\hat{\bf G}_1 + \hat{\bf G}_0^{\gamma},
\end{equation*}
such that \({\bf G}_0^{\gamma}\) converges in a suitable sense to some non-degenerate generator \({\bf G}_0\) as \(\gamma \uparrow \infty\). To make this more explicit, consider the candidates: 
\begin{equation}\label{eq:generatorDecompositionFormulae}
    \begin{split}
        \hat{\bf G}_2 &= \gamma^{-2}\hat{\mathscr G}^{\gamma}_{\bf RES},\\
        \hat{\bf G}_1 &= \hat{\mathscr G}_{\bf \overline{DEM}},\\
        \hat{\bf G}_0^{\gamma} &= \hat{\mathscr G}^{\gamma}_{\bf DEM} - \gamma\hat{\mathscr G}_{\bf\overline {DEM}} + \hat{\mathscr G}^{\gamma}_{\textbf{CAT}},\\
    \end{split}
\end{equation}
where 
\begin{equation*}
    \hat{\mathscr G}_{\bf \overline{DEM}}f(r, {\bm x}) := (F(r) - \nu^{-1})\bigg(\beta_1x_1\frac{\partial f(r, {\bm x})}{\partial x_1} +\beta_2x_2\frac{\partial f(r, {\bm x})}{\partial x_2}\bigg). 
\end{equation*}
Then \ref{it:decomposition} can be proved by combining Propositions \ref{prop:preparatoryLemma} and \ref{prop:uniformConvergence} below.
    
    \begin{proposition}[Extended Continuous Mapping Theorem]\label{prop:preparatoryLemma}
        Let \(\Y\) be a Polish space, let \(\{f^{n}\}_{n \geq 1}\) be a sequence of measurable functions \([0,T]\times \hat{E} \to [0,T]\times \Y\), and \(f\) a continuous function such that 
        \begin{equation*}
            \lim_{n \rightarrow \infty} f^{n} = f,
        \end{equation*}
        uniformly on compact sets. 
        Then \(M^n \xRightarrow{\mathscr{P}} M\) implies that \(f^{n}_* M^n\xRightarrow{\mathscr{P}(\Y)} f_*M\), where \(f_*M\) is the pushforward of \(M\) by \(f\), in the topology of weak convergence of measures on $[0,T]\times \hat E$.
    \end{proposition}

    \begin{proposition}\label{prop:uniformConvergence}
        For every \(f \in C^2(\hat{E})\), we have that 
        \begin{equation*}
            \lim_{\gamma \uparrow \infty}\sup_{(r,\bm x)\in K}\left|\hat{\bf G}_0^{\gamma}f - \hat{\bf G}_0 f\right|= 0,
        \end{equation*}
        for any compact \(K\subseteq \hat E\), where 
        \begin{equation}\label{eq:gZeroFormula}
            \hat{\bf G}_0f(r, {\bm x}) = \frac{F(r)+\nu^{-1}}{2}\sum_{i=1,2}\beta_ix_i \frac{\partial^2f(r,{\bm x})}{\partial x_i^2} + \kappa\int_{[0,1]^2} \bigg[f(r, {\bm u} \odot {\bm x}) - f(r, {\bm x})\bigg]\Upsilon(d{\bf u}).
        \end{equation}
    \end{proposition}

    \begin{proof}[Proof of Proposition \ref{prop:preparatoryLemma}]
        According to \cite[Theorem 1.11.1]{wellner1996weak}, it suffices to show that for any deterministic sequence \(\{m^n\}_{n \geq 1} \subseteq \mathscr{P}\) such that \(m^n \xrightarrow{\mathscr{P}}m\) one has \(f^n_* m^n \xrightarrow{\mathscr{P}(\Y)}f_*m\).
        
        The pushforward through a continuous map is a continuous functional \(\mathscr{P} \to \mathscr{P}(\Y)\), hence by the ordinary continuous mapping theorem \(f_*m^n \xrightarrow{\mathscr{P}(\Y)} f_*m\).

        Let \(F : [0,T] \times \Y \to \RR\) be a continuous bounded functional, and fix any \(\varepsilon>0\). 
        By Prokhorov's theorem, \(m^n \xrightarrow{\mathscr{P}} m\) implies there exist a compact set \(C \subseteq [0,T] \times \Y \) such that \(\sup_{n \geq 1}m^n(C^{\complement}) \leq \varepsilon\). 
        Furthermore, by uniform convergence there exists \(n_0\) large enough so that \\
        \({\sup_{n \geq n_0} \sup_{x \in C}|F \circ f(x) - F \circ f_n(x)| \leq \varepsilon}\). Hence for all \(n \geq n_0\)
        \begin{align*}
            \int_{[0,T] \times \hat{E}} (F\circ f_n - F \circ f)dm^n &\leq \int_{C} (F\circ f_n - F \circ f)dm^n + \int_{C^{\complement}} (F\circ f_n - F \circ f)dm^n\\
            &\leq \varepsilon + \varepsilon\sup_{[0,T] \times \hat{E}}|F|.
        \end{align*}
        As \(\varepsilon\) was arbitrary, it follows that \(\int F\, d(f_*^nm^n) - \int F\, d(f_*m^n) \to 0\). 
        Combining this with $m^n \xrightarrow{\mathscr{P}} m$, it follows that $\int F\, d(f_*^nm^n) - \int F\, d(f_*m) \to 0$.
        Since $F$ was arbitrary, we have shown that $f^n_* m^n \xrightarrow{\mathscr{P}(\Y)}f_*m$, completing the proof.
    \end{proof}

    \begin{proof}[Proof of Propositon \ref{prop:uniformConvergence}] 
    Let \(f \in C^2(\hat E)\) and let \(K \subseteq \hat E\) be a compact set. Then by a second-order Taylor expansion 
    \begin{equation}\label{eq:finiteDifferenceEstimate}
        [\hat{\mathscr{G}}^{\gamma}_{\bf DEM} - \gamma\hat{\mathscr{G}}_{\bf\overline {DEM}}]f(r, {\bm x})  = \frac{F(r)+\nu^{-1}
        }{2}\bigg(\beta_1x_1\frac{\partial^2 f(r,{\bm x})}{\partial x_1^2}  + \beta_2x_2\frac{\partial^2f(r,{\bm x})}{\partial x_2^2}\bigg) + o(1),
    \end{equation}
    where the error term vanishes as \(\gamma \uparrow \infty\) uniformly in \(r\) and \(\bm{x}\) as the second-order partial derivatives of \(f\) are uniformly continuous on \(K\). 
    
    Similarly, as the derivative of \(f\) is uniformly bounded on \(K\),  there exists a constant \(C_{K,f}\) depending only on \(f\) and \(K\) such that   
    \begin{align*}
        & \bigg|\hat{\mathscr{G}}_{\bf CAT}^{\gamma}f(r, {\bm x}) - \kappa\int\displaylimits_{[0,1]^2} \bigg(f(r, {\bf u} \odot {\bm x}) - f(r, {\bm x})\bigg)\Upsilon(d{\bf u})\bigg|\\ 
        &\qquad\qquad \leq \kappa\int\displaylimits_{[0,1]^2} \left|\EE{\lfloor \gamma{\bm x}\rfloor, {\bf u}}{f(r, \gamma^{-1}{\bf B})} - f (r, {\bf u} \odot {\bm x})\right|\Upsilon(d{\bf u})\\
        &\qquad\qquad\leq \frac{C_{K,f}\kappa}{\gamma}\int\displaylimits_{[0,1]^2} \EE{\lfloor \gamma{\bm x}\rfloor, {\bf u}}{|{\bf B} - \gamma{\bf u} \odot {\bm x}|}\Upsilon(d{\bf u})\\
        &\qquad\qquad\leq \frac{C_{K,f} \kappa}{\gamma} \sum_{i=1}^2\int\displaylimits_{[0,1]^2} \sqrt{\VVar{\lfloor \gamma x_i\rfloor, u_i}{B_i}}\Upsilon(d{\bf u})\\
        &\qquad\qquad\leq \frac{C_{K,f} \kappa}{\sqrt{\gamma}} \sup_{{\bm x} \in K} (\sqrt{x_1} + \sqrt{x_2}),
    \end{align*}
    for all \((r, {\bm x}) \in K\), where we have used Cauchy--Schwarz in the penultimate step.  
    \end{proof}

    Let \(h \in C^2(\Gamma)\). By applying the chain rule to \(\hat{\bf G}^{\gamma}(h \circ {\bm \pi})\), Proposition \ref{prop:piCancellation} below implies \ref{it:piCancellation}.
    
    \begin{proposition}\label{prop:piCancellation}
        Let \({\bm \pi}\) be defined as in \eqref{eq:piLimit}. Then \(\bm \pi \in C^2(\hat{E} \setminus \{{\bf 0}\})\) and
        \begin{equation*}
            \hat{\bf G}^{\gamma}{\pi}_1 \equiv \hat{\bf G}^{\gamma}_0 {\pi}_1, \quad \hat{\bf G}^{\gamma}{\pi}_2 \equiv \hat{\bf G}^{\gamma}_0 {\pi}_2; 
        \end{equation*}
    \end{proposition}

    \begin{proof}[Proof of Proposition \ref{prop:piCancellation}]
        The regularity follows from Theorem 7.1 of \cite{falconer1983differentiation}. Using Proposition \ref{prop:firstDerivatives} we compute that 
        \begin{equation}\label{eq:piGradientZero}
            \bigg(\nabla {\pi}_1({\bm x}) \cdot ({\bm \beta} \odot {\bm x}),  \nabla {\pi}_2({\bm x}) \cdot ({\bm \beta} \odot {\bm x})\bigg) = 0,
        \end{equation}
        for all \({\bm x} \in \hat{E} \setminus \{\bf 0\}\). We conclude by comparing \eqref{eq:piGradientZero} with the expression for \(\hat{\bf G}_1\) given in \eqref{eq:generatorDecompositionFormulae}. 
    \end{proof}
\subsection{Proof of \ref{it:psiTightness}}

By the characterization of compact pseudopaths sets given in Proposition \ref{prop:psiCompactSets}, \ref{it:psiTightness} follows from the following proposition: 

\begin{proposition}\label{prop:compactConfinment}
    For every compact \(K\subseteq \hat{E}\) and for every positive \(\varepsilon\) there exists a positive \(M_{\varepsilon}\) such that 
    \begin{equation*}
        \sup_{\gamma > 1}\sup_{(r,{\bm x})\in K}\PP{(r,{\bm x})}{\sup_{t \in [0,T]}|\hat{\bf X}_t^{\gamma}| + \hat{R}^{\gamma}_t \geq M_{\varepsilon}} \leq \varepsilon. 
    \end{equation*}
\end{proposition}

\begin{proof}[Proof of Proposition \ref{prop:compactConfinment}]
First notice that for \(M_\varepsilon>R_{\max}\)
    \begin{align*}
        \PP{(r,{\bm x})}{\sup_{t \in [0,T]}|\hat{\bf X}_t^{\gamma}| + \hat{R}^{\gamma}_t \geq M_{\varepsilon}} &\leq \PP{(r,{\bm x})}{\sup_{t \in [0,T]}|\hat{\bf X}_t^{\gamma}| +R_{\max}\geq M_{\varepsilon} } \\
        &\leq \PP{(r,{\bm x})}{\sup_{t \in [0,T]} V^{\gamma}(\hat{R}^\gamma_t,\hat{\bf X}_t^\gamma) \geq M_{\varepsilon} - R_{\max}}
    \end{align*}
    for any \(\gamma> 1\), where
    \begin{equation*}
        V^{\gamma}(r, {\bm x}) =|{\bm x}| + \frac{r}{\gamma}. 
    \end{equation*}
    Then by a direct computation
    \begin{align}\label{eq:driftBoundV}
        \hat{\bf G}^{\gamma} V^{\gamma}(r, {\bm x}) &= \gamma[R_{\max} - r - \nu^{-1} {\bm \beta} \cdot {\bm x}] + \kappa \int\displaylimits_{[0,1]^2} ({\bf u} \cdot {\bm x} - |{\bm x}|)\Upsilon(d{\bf u})\nonumber\\
        &\leq \gamma R_{\max} - c_0\gamma V^{\gamma}(r, {\bm x}),
    \end{align}
    for a constant \(c_0 > 0\). Define \(V_t^{\gamma} = V^{\gamma}(\hat{R}^{\gamma}_t, \hat{\bf X}^{\gamma}_t)\) and \(\hat{\bf G}^{\gamma}V^{\gamma}_t = (\hat{\bf G}^{\gamma} V^{\gamma})(\hat{R}^{\gamma}_t, \hat{\bf X}^{\gamma}_t)\). We can write 
    \begin{equation*}
        V_t^{\gamma} = V_0^\gamma + \int_0^t \hat{\bf G}^{\gamma} V_s^{\gamma} ds + A^{\gamma}_t,
    \end{equation*}
    where \(\{A^{\gamma}_t\}_{t \geq 0}\) is a martingale such that \(A^{\gamma}_0 = 0\). Using the integrating factor \(e^{-c_0\gamma t}\) and Ito integration by parts, we can rewrite this as 
    \begin{align}\label{eq:VGammaIntegrationByParts}
        V_t^{\gamma} &= e^{-c_0\gamma t}  \cdot e^{c_0\gamma t}V_t^{\gamma}\nonumber\\
        &= e^{-c_0\gamma t}V_0^\gamma + e^{-c_0\gamma t}\int_0^t ({\bf G}^{\gamma}V_s^{\gamma} + c_0\gamma V_s^{\gamma})e^{c_0\gamma s}ds + e^{-c_0 \gamma t}\int_0^t e^{c_0\gamma s } dA^{\gamma}_s.
    \end{align}
    Using \eqref{eq:driftBoundV} yields the bound 
    \begin{equation}\label{eq:VGammaUsefulUpperBound}
        V_t^{\gamma} \leq V_0^\gamma + \frac{R_{\max}}{c_0} + 2\sup_{0 \leq s \leq t}|A^{\gamma}_s|.
    \end{equation}
    Furthermore, as \(\{A^{\gamma}_t\}_{t \geq 0}\) is a pure jump process, we can express the expectation of its quadratic variation as 
    \begin{align*}
        \EE{(r,{\bm x})}{\langle A^{\gamma}\rangle_t} &= \EE{(r,{\bm x})}{\int_0^t (F(\hat{R}^{\gamma}_s) + \nu^{-1}){\bm\beta} \cdot \hat{\bf X}^{\gamma}_s ds} + \kappa\EE{(r, {\bm x})}{\int_0^t \int_{[0,1]^2}\EE{\gamma\hat{\bf X}_s^{\gamma}, {\bf u}}{(\gamma^{-1}|{\bf B}| - |\hat{\bf X}^{\gamma}_s|)^2}\Upsilon(d{\bf u})ds}\\
        &\leq \EE{(r,{\bm x})}{\int_0^t (F(\hat{R}^{\gamma}_s) + \nu^{-1}){\bm\beta} \cdot \hat{\bf X}^{\gamma}_s ds} + 2\kappa\EE{(r, {\bm x})}{\int_0^t \int_{[0,1]^2}{(|{\bf u } \odot \hat{\bf X}^{\gamma}_s| - |\hat{\bf X}^{\gamma}_s|)^2}\Upsilon(d{\bf u})ds}\\
        &\qquad + \frac{2\kappa}{\gamma^2}\EE{(r, {\bm x})}{\int_0^t \int_{[0,1]^2}\VVar{\gamma\hat{\bf X}_s, {\bf u}}{|{\bf B}|}\Upsilon(d{\bf u})ds}\\
        &\leq C_t\int_0^t \EE{(r, {\bm x})}{{W}^{\gamma}_s} ds, 
    \end{align*}
    where \(C_t\) is a constant that does not depend on \((\hat{R}^{\gamma}, \hat{\bf X}^{\gamma})\) (in particular it does not depend on \((r, {\bm x})\)), and \(W_t^{\gamma} := (V_t^{\gamma})^2\). By Dynkin's formula 
    \begin{align*}
        \EE{(r, {\bm x})}{W_t^{\gamma}} &= W_0^{\gamma} + \int_0^t \EE{(r,{\bm x})}{\hat{\bf G}^\gamma W_s^{\gamma}}ds\\
        &\leq (V_0^\gamma)^2 + \gamma\int_0^t\bigg(V_0^{\gamma}R_{\max} + \frac{R_{\max}^2}{c_0}  - c_0\EE{(r, {\bm x})}{W_s^{\gamma}}\bigg)ds
    \end{align*}
    where we have used the fact that \(\EE{(r,{\bm x})}{V^{\gamma}_t} \leq V_0 + R_{\max}/c_0\), which can be read from the expectation of \eqref{eq:VGammaIntegrationByParts}. By Gr\"onwall's inequality it follows that 
    \begin{equation*}
        \EE{(r, {\bm x})}{W_t^{\gamma}} \leq (V_0^{\gamma})^2 + \frac{V_0R_{\max}}{c_0} + \frac{R_{\max}^2}{c_0^2} =: \overline{W}. 
    \end{equation*}
    Hence by Doob's martingale inequality, for any positive \(\rho\), one has  
    \begin{align*}
        \PP{(r, {\bm x})}{\sup_{t \in [0,T]}|A^{\gamma}_t| \geq \rho } &\leq \frac{4\EE{(r, {\bm x})}{\langle A^{\gamma}\rangle_T}}{\rho^2}\\
        &\leq \frac{4TC_T}{\rho^2}\overline{W}. 
    \end{align*}
    Since by \eqref{eq:VGammaUsefulUpperBound}
    \begin{align*}
        \PP{(r, {\bm x})}{\sup_{t \in [0,T]}V_t^\gamma \geq M_{\varepsilon} - R_{\max}} &\leq \PP{(r, {\bm x})}{V_0^{\gamma} + \frac{R_{\max}}{c_0} \geq (M_{\varepsilon} - R_{\max})/2}\\ &\qquad + \PP{(r, {\bm x})}{\sup_{t \in [0,T]}|A_t^{\gamma}| \geq (M_{\varepsilon} - R_{\max})/4}, 
    \end{align*}
    choosing \(M_{\varepsilon}\) large enough proves the result. 
\end{proof}

\subsection{Proof of \ref{it:nonExtinction}}\label{subsec:proofNonExtinction}

We will start by introducing an auxiliary process, which we dub the \textit{jump-free process}, as it behaves like \((\hat{R}^{\gamma}, \hat{\bf X}^{\gamma})\) with the rate of catastrophic events set to zero (\(\kappa =0\)). Write \((\bar{R}^{\gamma}, \bar{\bf X}^{\gamma})\) for the jump-free process, with generator 
\begin{equation*}
    \bar{\bf G}^{\gamma} := \hat{\mathscr G}^{\gamma}_{\bf RES} + \hat{\mathscr G}^{\gamma}_{\bf DEM}. 
\end{equation*}
We write \(\Gamma_c\) for the \(c\)-neighborhood of \(\Gamma\), i.e.,
\begin{equation*}
    \Gamma_c = \{{\bm x} \in \RR_{\geq 0}^2 : d({\bm x}, \Gamma) < c\}. 
\end{equation*}
Also recall that we use \({\bf U}\) to denote a random variable with law \(\Upsilon\). As \(\gamma\) increases, the jump-free process is very rapidly pulled towards the stable manifold \(\Gamma\). Proposition \ref{prop:convergenceToManifold} makes this idea rigorous. 
\begin{proposition}\label{prop:convergenceToManifold}
    Define
    \begin{equation*}
        D_{\Gamma}({\bm x}) := (F \circ R^*({\bm \beta} \cdot {\bm x}) - \nu^{-1})^2,
    \end{equation*}
    as a measure of the distance between \(\bm x\) and \(\Gamma\). 
    For every positive \(c\) and \(\rho\), one has 
    \begin{equation}\label{eq:longTimeScaleConvergence}
        \lim_{\gamma \uparrow \infty}\sup_{r\in [0,R_{\max}]}\sup_{{\bm x} \in \Gamma_c}\PP{{(r,\bf U} \odot {\bm x})}{\sup_{t \in [T \log \gamma/\gamma, T]} D_{\Gamma}(\bar{\bf X}^{\gamma}_t) \geq \rho} = 0, 
    \end{equation}
    and
    \begin{equation}\label{eq:shortTimeScaleConvergence}
        \lim_{\gamma \uparrow \infty}\sup_{r\in [0,R_{\max}]} \sup_{{\bm x} \in \Gamma_c}\PP{(r,{\bf U} \odot {\bm x})}{\sup_{t \in [0, T]}D_{\Gamma}(\bar{\bf X}^{\gamma}_t) \geq D_{\Gamma}(\bar{\bf X}^{\gamma}_0) + \rho} = 0. 
    \end{equation} 
\end{proposition}
Notice the difference in the time interval considered between \eqref{eq:longTimeScaleConvergence} and \eqref{eq:shortTimeScaleConvergence}. This is because, even for very large \(\gamma\), the reversion to \(\Gamma\) does not occur instantaneously, and so we must allow a small amount of time to pass before considering arbitrarily small bounds. Near \(t = 0\), the best we can hope to show is that the distance between \(\bar{\bf X}\) and \(\Gamma\) does not increase much as time increases. 

Proposition \ref{prop:convergenceToManifold} is all that is needed to prove \ref{it:nonExtinction}. 

\begin{proof}[Proof of \ref{it:nonExtinction}]
Let \(\tau_1, \tau_2, \tau_3, \ldots \) be the occurrence times of the catastrophes in \([0,T]\). As the rate of catastrophes is the same for all \(\gamma > 1\) and is independent of the current state, we can consider a coupling such that these times are the same for all \(\gamma > 1\). 

Let \(\varepsilon > 0\). Then there exists \(\delta_{\varepsilon} > 0\) such that \(\P{\inf_{j > 1} |\tau_{j+1} - \tau_j| \leq \delta_{\varepsilon}} \leq \varepsilon\). 
As $\{\hat{\mathbf{X}}_0^\gamma\}_{\gamma>1}$ converges, hence is bounded, 
there also exists \(c > 0\) (not depending on $\gamma$) such that \(\hat{\bf X}^\gamma_0 \in \Gamma_{c}\) for all $\gamma>1$. 
By Assumption \ref{ass:survivalProbability} we can choose \(m_0\) such that \(\sup_{{\bm x \in \Gamma_c}}\P{|{\bf U} \odot {\bm x}| \leq m_0} = 0\). Then, taking advantage of the Markov property 
\begin{align*}
    &\P{\inf_{[0,T]}|\hat{\bf X}^{\gamma}_t| \leq m_0} \\
    &\quad\leq \P{\bigcup_{j= 1}^{\lfloor T/\delta_{\varepsilon}\rfloor} \{\inf_{[\tau_{j}, \tau_{j+1})}|\hat{\bf X}^{\gamma}_{t }| \leq m_0  \} \cup \{\hat{\bf X}^{\gamma}_{\tau_j^- } \not\in \Gamma_c\} \:\big|\: \inf_{j \geq 1} |\tau_{j+1} - \tau_j| > \delta_{\varepsilon}} 
    + \P{\inf_{j \geq 1} |\tau_{j+1} - \tau_j| \leq \delta_{\varepsilon}} \\
    &\qquad\leq C_0^{\varepsilon}\sup_{r\in [0,R_{\max}]}\sup_{{\bm x} \in \Gamma_c}\PP{(r,{\bf U} \odot {\bm x})}{\inf_{t \in [0,T]}|\bar{\bf X}^{\gamma}_{t}| \leq m_0}
    + C_0^{\varepsilon}\sup_{r\in [0,R_{\max}]}\sup_{{\bm x} \in \Gamma_c}\PP{(r,{\bf U} \odot {\bm x})}{\bar{\bf X}^{\gamma}_{T} \not\in \Gamma_c } + \varepsilon, 
\end{align*}
where \(C_0^{\varepsilon}\) is a constant depending only on \(\varepsilon\) and \({\bf U}\) has law \(\Upsilon\). Notice the change from \(\hat{\bf X}^{\gamma}\) to \(\bar{\bf X}^{\gamma}\) on the last line, which is possible as we are looking at an event occurring between jumps; as the process is restarted after the jump with a resource level that could be anywhere in $[0,R_{\max}]$, we need to add the supremum over $r$. 
By our choice of \(m_0\), there exists \(\rho > 0\) such that 
\begin{align*}
    \sup_{r\in [0,R_{\max}]}\PP{(r,{\bf U} \odot {\bm x})}{\inf_{t \in [0,T]}|\bar{\bf X}^{\gamma}_{t}| \leq m_0} 
    &\leq \sup_{r\in [0,R_{\max}]}\PP{(r,{\bf U} \odot {\bm x})}{\sup_{t \in [0,T]} D_{\Gamma}(\bar{\bf X}^{\gamma}_t) \geq D_{\Gamma}(\bar{\bf X}^{\gamma}_0) + \rho}. 
\end{align*}
For the second term we only need the long-term behavior: 
\begin{align*}
    \sup_{r\in [0,R_{\max}]}\PP{(r,{\bf U} \odot {\bm x})}{\bar{\bf X}^{\gamma}_{T} \not\in \Gamma_c} &\leq \sup_{r\in [0,R_{\max}]}\PP{(r,{\bf U} \odot {\bm x})}{D_{\Gamma}(\bar{\bf X}^{\gamma}_T) \geq \rho},
\end{align*}
for any choice of \(\rho\) that is large enough. By Proposition \ref{prop:convergenceToManifold}, both vanish as \(\gamma \uparrow \infty\). As \(\varepsilon\) was arbitrary, this concludes the proof.
\end{proof}

\subsection{Proof of Proposition \ref{prop:convergenceToManifold}}\label{subsec:convergenceInProbability}

As the proof of Proposition \ref{prop:convergenceToManifold} is both fundamental and involved, we dedicate an entire section to it. Let us start with a simple lemma which concerns Poisson point processes. 

\begin{lemma}[Sliding Window Lemma]\label{lem:slidingWindow}
    Let \(P\) be a homogeneous Poisson point process on \([0,T]\) with intensity \(I\). Then 
    \begin{equation*}
        \E{\sup_{0 \leq t \leq T-\delta} P[t, t+\delta]} \leq 4I\delta + \log \lfloor T/\delta \rfloor. 
    \end{equation*}
\end{lemma}

    Recall the definitions of \(R^*\) 
    \begin{equation*}
        R^*({\mathcal E}) = \frac{R_{\max} - a- \mathcal{E}}{2} + \frac{\sqrt{4aR_{\max} + (R_{\max} -a - {\mathcal E})^2}}{2},
    \end{equation*}
    and \(F\)
    \begin{equation*}
        F(r) = \frac{r}{a+r}. 
    \end{equation*}

\begin{proof}[Proof of Lemma \ref{lem:slidingWindow}]
   Every interval of the form \([t, t+\delta]\) is contained in an interval of the form \([k\delta, (k+2)\delta]\), for \(0 \leq k \leq \lfloor T/\delta \rfloor\). Therefore
   \begin{equation*}
       \sup_{0 \leq t \leq T-\delta} P[t, t+\delta] \leq \max_{0 \leq k \leq \lfloor T/\delta \rfloor} P[\delta k, \delta(k+2)]. 
   \end{equation*}
   Applying Jensen's inequality and computing the Poisson moment-generating function, we find 
   \begin{align*}
       \E{\max_{0 \leq k \leq \lfloor T/\delta \rfloor} P[\delta k, \delta(k+2)]} &\leq \log \E{\max_{0 \leq k \leq \lfloor T/\delta \rfloor} \exp P[\delta k, \delta(k+2)]}\\
       &\leq \log \lfloor T/\delta \rfloor + \log \E{\exp P[0, 2\delta]}\\
       &\leq \log \lfloor T/\delta \rfloor + 4I\delta. 
   \end{align*}
\end{proof}

The following facts about the functions \(F\) and \(R^*\) are immediate from differentiation:

\begin{lemma}\label{lem:growthRateDerivative}
    The following properties hold: 
    \begin{enumerate}
        \item\label{it:lipschitzF} \(F\) is Lipschitz over \([0,\infty)\); 
        \item\label{it:lipschitzR} \(R^*\) is Lipschitz over \([0,\infty)\);
        \item\label{it:rStarNegativeDerivative} For all compact sets \(K \subseteq \RR_{\geq 0}^2\), there exists a positive constant \(C\) such that 
        \begin{equation*}
            \sup_{{\bm x} \in K} (F \circ R^*)'({\bm\beta} \cdot{\bm x}) \leq -C;\text{ and}
        \end{equation*}
        \item \(D_\Gamma\) is Lipschitz on \([0,\infty)^2\).
    \end{enumerate}
\end{lemma}

\begin{proposition}\label{prop:compactConfinmentJumpFree}
    For every compact \(K\subseteq \hat E\) and every positive \(\varepsilon\) there exists a positive real number \(M_{\varepsilon}\), not depending on \(\gamma\), such that 
    \begin{equation*}
        \sup_{\gamma \geq 1}\sup_{(r,{\bm x})\in K}\PP{(r,\bm x)}{\sup_{t \in [0,T]}|\bar{\bf X}_t^{\gamma}| + \bar{R}^{\gamma}_t \geq M_{\varepsilon}} \leq \varepsilon. 
    \end{equation*}    
\end{proposition}

\begin{proof}[Proof of Proposition \ref{prop:compactConfinmentJumpFree}]
    Same proof as Proposition \ref{prop:compactConfinment} with \(\kappa = 0\).
\end{proof}

For the remainder of this section for each \(\epsilon>0\), with \(M_\epsilon\) as in Proposition \ref{prop:compactConfinmentJumpFree} applied to the compact set $K_c:=[0,R_{\max}]\times \{{\bm u}\odot{\bm x}: {\bm u}\in [0,1]^2, {\bm x}\in \overline{\Gamma_c}\}$, define
\begin{equation}
    \label{E:tauepsilon}
    \tau^\varepsilon:= \inf\bigl\{ t\ge 0\, : \, |\bar{\bf X}^\gamma_t | + \bar R^\gamma_t \ge M_\varepsilon \bigr\},
\end{equation}
so that, by Proposition \ref{prop:compactConfinmentJumpFree}, \(\sup_{\gamma\ge 1} \sup_{(r,{\bm x})\in K_c}\PP{(r,{\bm x})}{ \tau^\varepsilon \le T} \le \varepsilon\).
Note that \(\tau^\varepsilon\) depends on \(\gamma\) through \(\bar{\bf X}^\gamma\), though we suppress this from the notation; \(M_\epsilon\) may be chosen without reference to \(\gamma\).

Using Lemma \ref{lem:slidingWindow}, we show that the resource level rapidly approaches its equilibrium value. 

\begin{proposition}\label{prop:resourceLevelMomentBound}
    For every positive \(\delta_1, \delta_2,\varepsilon,c\), we have   
    \begin{equation*}
        \lim_{\gamma \uparrow \infty}\sup_{r\in [0,R_{\max}]}\sup_{\bm x \in \Gamma_c}\PP{(r,{\bf U} \odot {\bm x})}{\sup_{t \in [T/\gamma^{2-\delta_2}, T \wedge \tau^{\varepsilon}]}|\bar R_t^{\gamma} - R^*({\bm \beta} \cdot \bar{\bf X}^{\gamma}_t)| \geq \gamma^{-1/2 +\delta_1}} = 0. 
    \end{equation*}
\end{proposition}

\begin{corollary}\label{cor:resourceLevelMomentBound}
    There exists \(\alpha > 0\) such that for all positive \(\rho\)
    \begin{equation*}
        \lim_{\gamma \uparrow \infty}\sup_{r\in [0,R_{\max}]}\sup_{{\bm x} \in \Gamma_c}\PP{(r,{\bf U} \odot {\bm x})}{\sup_{t \in [0,T]} e^{-\alpha \gamma t} \int_0^t e^{\alpha \gamma s}[\alpha \gamma D_{\Gamma}(\bar{\bf X}^{\gamma}_s) + (\bar{\bf G}^{\gamma} D_{\Gamma})(\bar{R}^{\gamma}_s, \bar{\bf X}^{\gamma}_s)]ds \geq \rho} = 0.
    \end{equation*}
\end{corollary}

\begin{proof}[Proof of Corollary \ref{cor:resourceLevelMomentBound}]
    Let \(\varepsilon\) be positive. Define \(\eta^{\varepsilon, \gamma} = \tau^{\varepsilon} \wedge\inf\{t \geq T/\gamma^{3/2} : |\bar{R}^{\gamma}_t - R^*({\bm \beta} \cdot \bar{\bf X}^{\gamma}_t)| \geq \gamma^{-1/4}\}\). 
    By Proposition \ref{prop:resourceLevelMomentBound} we have \(\lim_{\gamma \uparrow \infty}\sup_{r\in [0,R_{\max}]}\sup_{{\bm x} \in \Gamma_c}\PP{(r,{\bf U} \odot {\bm x})}{\eta^{\varepsilon, \gamma} \leq T} = 0 \). The choice of exponents \(3/2\) and \(-1/4\) is arbitrary. What matters is that we apply Proposition \ref{prop:resourceLevelMomentBound} with \(\delta_1 \in (0, 1/2)\) and \(\delta_2 \in (0,1)\). Define 
    \begin{equation*}
        N^{\gamma}(r, {\bm x}) := \alpha\gamma D_{\Gamma}({\bm x}) + (\bar{\bf G}^{\gamma}D_{\Gamma})(r, {\bm x})
    \end{equation*}
    We distribute the supremum over the sub-intervals \([0,T/\gamma^{3/2} \wedge \tau^{\varepsilon}]\) and \([T/\gamma^{3/2} \wedge \tau^{\varepsilon}, T \wedge \eta^{\varepsilon, \gamma}]\) by splitting the integral, as follows 
    \begin{align*}
    \sup_{t \in [0,T \wedge \eta^{\varepsilon, \gamma}]}&e^{-\alpha \gamma t}\int_0^t e^{\alpha \gamma s} N^{\gamma}(\bar{R}^{\gamma}_s, \bar{\bf X}^{\gamma}_s)ds \\
        &\leq \sup_{t \in [0,T \wedge \eta^{\varepsilon, \gamma}]}e^{-\alpha \gamma t} \int_{0}^{t \wedge T/\gamma^{3/2}} e^{\alpha \gamma s} N^{\gamma}(\bar{R}^{\gamma}_s, \bar{\bf X}^{\gamma}_s)ds + \sup_{t \in [0,T \wedge \eta^{\varepsilon, \gamma}]}e^{-\alpha \gamma t}\int_{t \wedge T/\gamma^{3/2}}^{t} e^{\alpha \gamma s}N^{\gamma}(\bar{R}^{\gamma}_s, \bar{\bf X}^{\gamma}_s)ds\\
        &\leq \sup_{t \in [0,T/\gamma^{3/2}  \wedge \tau^{\varepsilon}]}e^{-\alpha \gamma t} \int_{0}^{t} e^{\alpha \gamma s} N^{\gamma}(\bar{R}^{\gamma}_s, \bar{\bf X}^{\gamma}_s)ds
        + \sup_{t \in [T/\gamma^{3/2} \wedge \tau^{\varepsilon}, T \wedge \eta^{\varepsilon, \gamma}]}e^{-\alpha \gamma t} \int_{T/\gamma^{3/2}}^{t} e^{\alpha \gamma s} N^{\gamma}(\bar{R}^{\gamma}_s, \bar{\bf X}^{\gamma}_s)ds.
    \end{align*}
     By \eqref{eq:finiteDifferenceEstimate} from the proof of Proposition \ref{prop:uniformConvergence} we have
     \begin{equation}\label{eq:deltaDerivativeUpperBound}
         (\bar{\bf G}^{\gamma} D_{\Gamma})(r, {\bm x}) = \gamma (\hat{\mathscr{G}}_{\overline{\text{DEM}}}D_{\Gamma})(r, {\bm x}) + o(\gamma),
     \end{equation}
    where the $o(\gamma)$ error term is uniform in \(\bm x\) such that \(|\bm x| \leq M_{\varepsilon}\) and \(\gamma > 1\), and we can evaluate
    \begin{equation}\label{eq:demDGamma}
        (\hat{\mathscr{G}}_{\overline{\text{DEM}}}D_{\Gamma})(r, {\bm x}) = 2(F(r) - \nu^{-1})(F \circ R^*({\bm \beta}\cdot{\bm x})- \nu^{-1})(F \circ R^*)'({\bm \beta} \cdot {\bm x})\sum_{i=1,2}\beta_i^2x_i.
    \end{equation}
    (Here we rely on the facts that \(D_\Gamma\) has bounded second derivatives, and that \(\hat{\mathscr{G}}_{\text{RES}}D_{\Gamma}\equiv 0\).)
    
    For the supremum over \([0, T/\gamma^{3/2} \wedge \tau^{\varepsilon}]\), we note that by \ref{it:lipschitzF} and \ref{it:lipschitzR} of Lemma \ref{lem:growthRateDerivative} we can infer from \eqref{eq:deltaDerivativeUpperBound} and \eqref{eq:demDGamma} that there exists a constant \(C_1\) such that 
    \begin{equation}
    \label{E:Ggamma_bound}
       \bigl|(\bar{\bf G}^{\gamma} D_{\Gamma})(r,{\bm x})\bigr| \leq C_1\gamma,\, \text{for all } r\in [0,R_{\max}], \, \bm x \text{ s.t. } |\bm x|\le M_\epsilon. 
    \end{equation}
    Hence for any choice of \(\alpha\), we have
    \begin{align*}
        e^{-\alpha \gamma t}\int_0^t e^{\alpha \gamma s}N^{\gamma}(\bar{R}^{\gamma}_s, \bar{\bf X}^{\gamma}_s)ds &\leq (\alpha + C_1)\gamma e^{-\alpha\gamma t}\int_0^t e^{\alpha \gamma s}ds\\
        &= \frac{\alpha + C_1}{\alpha} \bigg(1 -e^{-\alpha \gamma t}\bigg), 
    \end{align*}
    for all \(t \leq \tau^{\varepsilon}\). Hence 
    \begin{equation*}
        \sup_{t \in [0, T/\gamma^{3/2} \wedge \tau^{\varepsilon}]} e^{-\alpha \gamma t}\int_0^t e^{\alpha \gamma s}N^{\gamma}(\bar{R}^{\gamma}_s, \bar{\bf X}^{\gamma}_s)ds \leq \frac{\alpha + C_1}{\alpha}\bigg(1-e^{-\alpha T/\sqrt{\gamma}}\bigg),
    \end{equation*}
    which vanishes as \(\gamma \uparrow \infty\).
    This constant, and future constants denoted ``\(C_\cdot\)'', will be understood to be allowed to depend on \(\epsilon\), but not on \(\gamma\).

    For the supremum over \([T/\gamma^{3/2} \wedge \tau^{\varepsilon}, T \wedge \eta^{\varepsilon, \gamma}]\), we must work out a different bound for \(\bar{\bf G}^{\gamma} D_{\Gamma}\). Starting from \eqref{eq:demDGamma} we write
    \begin{align*}
    (\hat{\mathscr{G}}_{\overline{\text{DEM}}}D_{\Gamma})(r, {\bm x}) 
        = 2\sum_{i=1,2} \beta_i^2x_i\bigg[D_{\Gamma}&({\bm x}) (F \circ R^*)'({\bm \beta} \cdot {\bm x}) \\ 
    &+ (F(r) - F \circ R^*({\bm \beta} \cdot {\bm x}))(F \circ R^*({\bm \beta \cdot {\bm x}})-\nu^{-1})(F \circ R^*)'({\bm \beta} \cdot {\bm x})\bigg].
    \end{align*}
     For the first summand on the right hand side, using \ref{it:rStarNegativeDerivative} of Lemma \ref{lem:growthRateDerivative}, we can see there exists a positive constant \(C\) such that 
     \begin{equation*}
         D_{\Gamma}({\bm x}) (F \circ R^*)'({\bm \beta} \cdot {\bm x}) \leq - CD_{\Gamma}({\bm x}).
     \end{equation*}
    For the second summand, taking absolute value and using \ref{it:lipschitzF} of Lemma \ref{lem:growthRateDerivative}, there is a constant \(C\) such that  
    \begin{align*}
        (F(r) - F \circ R^*(\bm \beta\cdot {\bm x}))(F \circ R^*({\bm \beta \cdot {\bm x}})-\nu^{-1})(F \circ R^*)'({\bm \beta} \cdot {\bm x})
        &\leq C|r - R^*({\bm \beta} \cdot {\bm x})|.
    \end{align*}
    Hence it follows from \eqref{eq:deltaDerivativeUpperBound} that
    \begin{equation*}
        (\bar{\bf G}^{\gamma}D_{\Gamma})(r, {\bm x}) \leq -C_2 \gamma D_{\Gamma}({\bm x}) + l(\gamma)
    \end{equation*}
    where \(C_2\) is a positive constant and \(l \in o(\gamma)\) is uniform in \(r\) and \(\bm x\) such that \(|{\bm x}| \leq M_{\varepsilon}\) and \(|r - R^*({\bm \beta} \cdot {\bm x})| \leq \gamma^{-1/4}\). We choose \(\alpha = C_2\), so that the negative term on the right-hand side cancels out. 
    We then have
    \begin{align*}
        e^{-\alpha \gamma t}\int_{T/\gamma^{3/2}}^{t} e^{\alpha \gamma s} N^{\gamma}(\bar{R}^{\gamma}_s, \bar{\bf X}^{\gamma}_s)ds &\leq l(\gamma) e^{-\alpha\gamma t} \int_{T/\gamma^{3/2}}^t e^{\alpha\gamma s}ds\\
        &\leq \frac{1}{\alpha} \frac{l(\gamma)}{\gamma},
    \end{align*}
    for all \(t \leq \eta^{\varepsilon, \gamma}\), which vanishes as \(\gamma \uparrow \infty\). 
\end{proof}

\begin{proof}[Proof of Proposition \ref{prop:resourceLevelMomentBound}]
    Define \(Q(r, {\bm x}) = (r - R^*({\bm \beta} \cdot {\bm x}))^2\), and set \(Q_\gamma= \gamma^{1-\delta_1}Q\). By using the fact that \(R_{\max} - R^*({\bm \beta} \cdot {\bm x}) = {\bm \beta} \cdot {\bm x} F \circ R^*({\bm \beta} \cdot {\bm x})\), we see that   
    \begin{equation*}
        \begin{split}
            \hat{\mathscr G}_{\bf RES}^{\gamma} Q_{\gamma}(r, {\bm x}) 
            &= 2\gamma^{3-\delta_1} [R^*({\bm \beta} \cdot{\bm x}) - r + [F \circ R^*({\bm \beta} \cdot {\bm x}) - F(r)]{\bm \beta} \cdot {\bm x}](r - R^*({\bm \beta} \cdot{\bm x}))\\
            &= -2\gamma^{2} Q_{\gamma}(r,{\bm x}) + 2\gamma^{3-\delta_1}{\bm \beta} \cdot {\bm x}(F \circ R^*({\bm\beta} \cdot {\bm x})-F(r))(r-R^*({\bm \beta} \cdot {\bm x})).
        \end{split}
    \end{equation*}
    Since \(F\) is an increasing function, \((F \circ R^*({\bm\beta} \cdot {\bm x})-F(r))(r-R^*({\bm \beta} \cdot {\bm x})) \leq 0\). And so 
    \begin{equation}\label{eq:lyapunovFosterConditionResources}
        \hat{\mathscr G}_{\bf RES}^{\gamma} Q_{\gamma}(r, {\bm x}) \leq - 2\gamma^2 Q_{\gamma}(r, {\bm x}). 
    \end{equation}
    Applying Ito integration by parts to the product of \(Q_{\gamma}(\bar{R}_t^{\gamma}, \bar{\bf X}^{\gamma}_t)\) and \(e^{2\gamma^2t}\), we find 
    \begin{align*}
        Q_{\gamma}(\bar{R}_t^{\gamma}, \bar{\bf X}^{\gamma}_t) &= e^{-2\gamma^2t}\bigg(e^{2\gamma^2 t}Q_{\gamma}(\bar{R}_t^{\gamma}, \bar{\bf X}^{\gamma}_t)\bigg)\\
        &= \gamma^{1-\delta_1}e^{-2\gamma^2t}Q(\bar{R}_0^{\gamma}, \bar{\bf X}^{\gamma}_0)\\
        &\qquad+ e^{-2\gamma^2t}\int_0^t e^{2\gamma^2s}[2\gamma^2 Q_{\gamma}(\bar{R}^{\gamma}_s, \bar{\bf X}^{\gamma}_s) + \hat{\mathscr G}^{\gamma}_{\bf RES} Q_{\gamma}(\bar{R}^{\gamma}_s, \bar{\bf X}^{\gamma}_s)]ds\\
        &\qquad\qquad+ e^{-2\gamma^2 t}\sum_{0 \leq s \leq t} e^{2\gamma^2 s} \Delta Q_{\gamma}(\bar{R}^{\gamma}_s, \bar{\bf X}^{\gamma}_s),
    \end{align*}
    where \(\Delta\) denotes the jump component, i.e., \(\Delta Q_{\gamma}(\bar{R}^{\gamma}_t, \bar{\bf X}^{\gamma}_t) = Q_{\gamma}(\bar{R}^{\gamma}_t, \bar{\bf X}^{\gamma}_t) - Q_{\gamma}(\bar{R}^{\gamma}_{t^-}, \bar{\bf X}^{\gamma}_{t^-})\).
    The first term vanishes as \(\gamma \uparrow \infty\) when \(t \geq T/\gamma^{2-\delta_2} \), since \(0\le Q\le R_{\max}^2\) on \(\hat E\), as \(r,R^*({\bm \beta} \cdot {\bm x})\in [0,R_{\max}]\). The second term is non-positive by \eqref{eq:lyapunovFosterConditionResources}. For the third term, choosing an arbitrary \(\delta\) such that \(0 < \delta < \delta_1\), we write: 
    \begin{align*}
        \sup_{0 \leq t \leq T \wedge \tau^{\varepsilon}}&e^{-2\gamma^2 t}\sum_{0 \leq s \leq t} e^{2\gamma^2 s} \Delta Q_{\gamma}(\bar{R}^{\gamma}_s, \bar{\bf X}^{\gamma}_s) \leq \sup_{0 \leq t \leq T \wedge \tau^{\varepsilon  }}e^{-2\gamma^2 t}\sum_{0 \leq s \leq t} e^{2\gamma^2 s} |\Delta Q_{\gamma}(\bar{R}^{\gamma}_s, \bar{\bf X}^{\gamma}_s)|\\
        &\leq \sup_{0 \leq t \leq T \wedge \tau^{\varepsilon}} e^{-2\gamma^2t}\bigg( \sum_{0 \leq s \leq t-\gamma^{-2+\delta}}e^{2\gamma^2 s} |\Delta Q_{\gamma}(\bar{R}^{\gamma}_s, \bar{\bf X}^{\gamma}_s)|
        + \sum_{t-\gamma^{-2+\delta} \leq s \leq t}e^{2\gamma^2 s} |\Delta Q_{\gamma}(\bar{R}^{\gamma}_s, \bar{\bf X}^{\gamma}_s)|  \bigg)\\
        &\leq \sup_{\gamma^{-2+\delta} \leq t \leq T \wedge \tau^{\varepsilon}}\sum_{t-\gamma^{-2+\delta} \leq s \leq t}|\Delta Q_{\gamma}(\bar{R}^{\gamma}_s, \bar{\bf X}^{\gamma}_s)|
         + \sup_{0 \leq t \leq T \wedge \tau^{\varepsilon}} e^{-2\gamma^{\delta}} \sum_{0 \leq s \leq t } |\Delta Q_{\gamma}(\bar{R}^{\gamma}_s, \bar{\bf X}_s^{\gamma})|\\
        &\leq \sup_{\substack{t_1, t_2 \leq \tau^{\varepsilon}\\ |t_1-t_2| \leq \gamma^{-2+\delta}}} \sum_{t_1 \leq s \leq t_2}|\Delta Q_{\gamma}(\bar{R}^{\gamma}_s, \bar{\bf X}^{\gamma}_s)|
        + e^{-2\gamma^{\delta}}\sum_{0 \leq t \leq T \wedge \tau^{\varepsilon}} |\Delta Q_{\gamma}(\bar{R}^{\gamma}_s, \bar{\bf X}^{\gamma}_s)|. 
    \end{align*}
As \(R^*\) is Lipschitz continuous (\ref{it:lipschitzR} of Lemma \ref{lem:growthRateDerivative}) and all jumps are of size \(\gamma^{-1}\) and jump rates are never larger than \(C_1\gamma^2\), for some constant \(C_1\) depending on \(M_{\varepsilon}\), there exists a coupling such that 
\begin{equation*}
    \sum_{t_1 \leq s \leq t_2} |\Delta Q_{\gamma}(\bar{R}^{\gamma}_s, \bar{\bf X}^{\gamma}_s)| \leq \gamma^{-\delta_1}C_2 P^{\gamma}([t_1, t_2]),
\end{equation*}
for all \(0 < t_1 \leq t_2 \leq T \wedge \tau^{\varepsilon}\), where \(P^{\gamma}\) is a homogeneous Poisson process of intensity \(C_1\gamma^2\), and \(C_2\) is another constant. From this, its immediately clear that \(e^{-2\gamma^{\delta}}\sum_{0 \leq t \leq T}|\Delta Q_{\gamma}(\bar{R}^{\gamma}_t, \bar{\bf X}^{\gamma}_t)|\) goes to zero in probability in the manner desired. For the other term, we can use the sliding window bound given by Lemma \ref{lem:slidingWindow} along with Markov's inequality to obtain that 
\begin{equation*}
    \sup_{(r,{\bm x})\in K_c}\PP{(r,{\bm x})}{\gamma^{-\delta_1}C_2\sup_{0 \leq t \leq T-\gamma^{-2+\delta}}P^{\gamma}([t, t+\gamma^{-2+\delta}]) \geq \zeta} \leq \frac{C_2}{\zeta \gamma^{\delta_1}} \left(4C_1\gamma^{\delta}  + \log (T \gamma^{2-\delta})\right),
\end{equation*}
which vanishes as \(\gamma \uparrow \infty\). 

As \(\zeta\) was arbitrary, this concludes the proof. 
\end{proof}

\begin{proposition}\label{prop:martingaleTermVanishing}
    Define the martingale 
    \begin{equation*}
        A_t^{\gamma} := D_{\Gamma}(\bar{\bf X}^{\gamma}_t) - D_{\Gamma}(\bar{\bf X}^{\gamma}_0) - \int_0^t (\bar{\bf G}^{\gamma} D_{\Gamma})(\bar{R}^{\gamma}_s, \bar{\bf X}^{\gamma}_s) ds.
    \end{equation*}
    Then for all positive \(\rho\)
    \begin{equation*}
        \lim_{\gamma \uparrow \infty} \sup_{r\in [0,R_{\max}]}\sup_{{\bm x} \in \Gamma_c} \PP{(r,{\bf U} \odot {\bm x})}{\sup_{0 \leq t \leq T}e^{-\alpha\gamma t}\int_0^t e^{\alpha\gamma s}dA^{\gamma}_s \geq \rho } = 0.
    \end{equation*}
\end{proposition}

\begin{proof}[Proof of Proposition \ref{prop:martingaleTermVanishing}]
    Fix $(r,\bm x)\in K_c$.
    For the remainder of the proof we understand $\mathbb{P}$ to mean $\mathbb{P}_{(r,{\bm x})}$ and $\mathbb{E}$ to mean $\mathbb{E}_{(r,{\bm x})}$, and interpret all bounds without further comment to be uniform in $(r,{\bm x})$.

    We use Lemma 5.2 of~\cite{katzenberger1991solutions}, according to which it suffices to verify that the sequence \(\{A^{\gamma}\}_{\gamma > 1}\) is tight for the Skorokhod topology. To do so, we use the Aldous--Rebolledo criterion~\cite{aldous1978stopping,rebolledo1980central}. 
    
   Let \(\mathscr{T}\) be the set of stopping times bounded above by \(T\). We first observe that for any \(\tau\in \mathscr{T}\) the process \(\{D_{\theta}^{\gamma, \tau}\}_{\theta \geq 0}\) defined as \(D^{\gamma,\tau}_{\theta} := A_{\tau + \theta}^{\gamma} - A^{\gamma}_{\tau}\) is also a martingale, with respect to the shifted filtration \(\{\mathcal{F}_{\tau + \theta}\}_{\theta \geq 0}\). Since \(\bar{\bf X}^\gamma\) is a pure jump process and \(D_\Gamma\) has no dependence on \(r\), the optional quadratic variation of \(\{D_{\theta}^{\gamma, \tau}\}_{\theta \geq 0}\) is 
    \begin{equation*}
        [D^{\gamma, \tau}]_{\theta} = \sum_{\tau \leq s \leq \tau + \theta} (\Delta A^{\gamma}_s)^2 = \sum_{\tau \leq s \leq \tau + \theta} (\Delta D_{\Gamma}(\bar{\bf X}^{\gamma}_s))^2,
    \end{equation*}
    where \(\Delta f(t) = f(t) - f(t^-)\) is the jump component. Note that \(D^{\gamma,0}_t=A_t^{\gamma}\).

    Now fix any \(\varepsilon>0\). 
    For times \(t\le \tau^\varepsilon\) we have \(|\bar{\bf X}_t^\gamma|\le M_\varepsilon\), hence by \eqref{E:Ggamma_bound} \(A^\gamma\) is bounded for \(t\le \tau^\varepsilon\), hence is {\em a fortiori} a square-integrable martingale.
    By Proposition I.4.50c of \cite{jacod2013limit} it follows that \((D^{\gamma,\tau}_{\cdot\wedge (\tau^\varepsilon-\tau)})^2 - [D^{\gamma,\tau}]_{\cdot\wedge (\tau^\varepsilon-\tau)}\) is a martingale.
    Since \(D^{\gamma,\tau}_0=0\),
    \begin{equation}
        \label{E:second_moment_qv}
      \E{(D^{\gamma,\tau}_{\theta\wedge (\tau^\varepsilon-\tau)})^2 } = \E{ [D^{\gamma,\tau}]_{\theta \wedge (\tau^\varepsilon-\tau)}}.
    \end{equation}
    
    The jumps in \(\bar{\bf X}_t^\gamma\) are of size exactly \(1/\gamma\), hence those of \(A_t^{\gamma}\) are bounded by \(L/\gamma\), where \(L\) is a Lipschitz constant for \(D_\Gamma\) (cf. Lemma \ref{lem:growthRateDerivative}). 
    The total jump rate at any \((r,\bm x)\) with \(|{\bm x}|\le M_\varepsilon\) is bounded by \({\gamma^2 M_\varepsilon(\beta_1\vee\beta_2)(1+\nu^{-1})}\); by the definition of \(\tau^\varepsilon\) this bound applies at \((\bar R_t^\gamma,\bar {\bf X}_t^\gamma)\) for all \(t\le \tau^\varepsilon\).
    Hence \(\E{[ D^{\gamma,\tau}]_{\theta}} \leq C^*_\varepsilon \theta\) for all \(\gamma>1\) on the event \(\{\tau + \theta \leq \tau^{\varepsilon}\}\), where \({C^*_\varepsilon:= L^2 M_\varepsilon (\beta_1\vee \beta_2)(1+\nu^{-1})}\).    
    By \eqref{E:second_moment_qv}, for any positive \(\delta\) 
    \begin{equation*}
        \sup_{\gamma > 1}\sup_{\tau \in \mathscr{T}}\P{|D^{ \gamma,\tau}_{\theta\wedge (\tau^{\varepsilon} - \tau)}| \geq \delta} \leq \frac{C^*_\varepsilon\theta}{\delta^2}
    \end{equation*}
    by Chebyshev's inequality. We then have
    \begin{equation} \label{E:Dsupbound}
    \P{|D^{\gamma,\tau}_{\theta}| \geq \delta} \le 
    \sup_{\gamma > 1}\P{|D^{\gamma,\tau}_{\theta\wedge (\tau^{\varepsilon} - \tau)}| \geq \delta} +\P{\tau^\varepsilon \le T} 
    \le
    \frac{C^*_\varepsilon \theta}{\delta^2} +\varepsilon 
      \text{ for all }\gamma>1,\, \tau\in \mathscr{T}.
    \end{equation}
     Applying \eqref{E:Dsupbound} with \(\tau\equiv 0\) implies tightness of the fixed-time marginals. Indeed, we have for each fixed \(t\) and \(\epsilon>0\)
     \begin{equation*}
       \lim_{\delta\to \infty}\sup_{\gamma > 1}\sup_{t\in [0,T]}\P{|A_t^{\gamma}| \geq \delta} \le \lim_{\delta\to \infty}\sup_{\gamma > 1}\sup_{t\in [0,T]} \frac{C^*_\varepsilon T}{\delta^2} + \varepsilon = \varepsilon
     \end{equation*}
     hence it is 0 since \(\varepsilon>0\) is arbitrary.
    
    Taking the supremum over all \(\theta\in [0,\theta^*]\) we have for all \(\delta>0\) and \(\epsilon>0\)
    \begin{equation*}
    \sup_{\gamma > 1}\sup_{\tau \in \mathscr{T}} \sup_{\theta\in[0,\theta^*]}\P{|D^{\gamma,\tau}_{\theta}| \geq \delta} 
    \le
    \frac{C^*_\varepsilon \theta^*}{\delta^2} +\varepsilon .
    \end{equation*}
    Letting \(\theta^*\) go to 0 this is bounded by \(\varepsilon\), hence is 0 since \(\varepsilon\) is arbitrary.
    This is Aldous's tightness criterion, which completes the proof.
\end{proof}

\begin{proof}[Proof of Proposition \ref{prop:convergenceToManifold}]
    Pick \(\alpha\) as in Corollary \ref{cor:resourceLevelMomentBound}. Using Ito integration by parts
    \begin{equation*}
        D_{\Gamma}(\bar{\bf X}^{\gamma}_t) = e^{-\alpha \gamma t}D_{\Gamma}(\bar{\bf X}^{\gamma}_0) + e^{-\alpha\gamma t } \int_0^t e^{\alpha \gamma s}[\alpha \gamma D_{\Gamma}(\bar{\bf X}^{\gamma}_s) + (\bar{\bf G}^{\gamma}D_{\Gamma})(\bar{R}_s^{\gamma},\bar{\bf X}_s^{\gamma})]ds + e^{-\alpha \gamma t}\int_0^t e^{\alpha \gamma s}dA^{\gamma}_s. 
    \end{equation*}
    Fix \(\rho>0\). By Corollary \ref{cor:resourceLevelMomentBound} and Proposition \ref{prop:martingaleTermVanishing} the second and third terms are each at most \(\rho/3\), with probability tending to 1 as \(\gamma\to\infty\), uniformly over \(r\in[0,R_{\max}]\) and \(\bm x\in \Gamma_c\).
    The first term is bounded above by \(D_{\Gamma}(\bar{\bf X}^{\gamma}_0)\) for every \(t\), which gives \eqref{eq:shortTimeScaleConvergence};
    and since \(D_\Gamma\) is bounded on the compact set \(K_c\) it is at most \(\rho/3\) for all \(t\ge T\log\gamma/\gamma\) once \(\gamma\) is large enough, yielding \eqref{eq:longTimeScaleConvergence}. 
\end{proof}

\subsection{Proof of \ref{it:skorokhodTightness}}

We use the Aldous--Rebolledo criterion \cite{aldous1978stopping,rebolledo1980central}. As the target space of \(\bm \pi\) is compact, tightness of the finite-dimensional marginals is immediate. Hence it suffices to verify Aldous's condition to show tightness in the \(J_1\) topology of \(\hat{\bf Z}^{\gamma}\), meaning \ref{it:skorokhodTightness} follows from the following proposition: 

\begin{proposition}\label{prop:aldousCriterion}
    For every \(\delta > 0\), one has 
    \begin{equation*}
        \lim_{\theta^* \downarrow 0} \limsup_{\gamma \rightarrow \infty} \sup_{\tau \in \mathscr{T}}\sup_{0 \leq \theta \leq \theta^*} \P{|\hat{\bf Z}^{\gamma}_{\tau + \theta} - \hat{\bf Z}^{\gamma}_{\tau}| > \delta} = 0. 
    \end{equation*}
    where \(\mathscr{T}\) is the set of stopping times bounded above by \(T\).
\end{proposition}

\begin{proof}[Proof of Proposition \ref{prop:aldousCriterion}]
For all \(t \in [0,T]\) and \(i \in \{1,2\}\), define the martingale 
\begin{equation*}
    \{A_t^{\pi_i, \gamma}\}_{t \geq 0} := \bigg\{(\hat{Z}_t^{\gamma})_i - \int_{0}^{t} (\hat{\bf G}^{\gamma}\pi_i)(\hat{R}^{\gamma}_s, \hat{\bf X}^{\gamma}_s)ds\bigg\}_{t \geq 0}. 
\end{equation*}
For any \(t \in [0,T]\) and \(\theta > 0\), also define 
\begin{equation*}
    D_{\theta}^{t, \gamma} := A_{t+\theta}^{\pi_i, \gamma} - A^{\pi_i, \gamma}_t. 
\end{equation*}
Let \(\tau\) be a stopping time bounded above by \(T\). As \(A_t^{\pi_i, \gamma}\) enjoys the strong Markov property, we have that \(\{D_{\theta}^{\tau, \gamma}\}_{\theta  > 0}\) is a martingale with respect to the shifted filtration \(\{\mathcal{F}_{\tau + \theta}\}_{\theta >0}\). As \(\{{\bf Z}_t^{\gamma}\}_{t \geq 0}\) is a pure jump process, the quadratic variation of \(\{D^{\gamma, \tau}\}_{\theta \geq 0}\) is 
\begin{equation*}
    [D^{\tau, \gamma}]_{\theta} = \sum_{\tau \leq s \leq \tau+\theta} (\Delta A_{s}^{\pi_i, \gamma})^2 = \sum_{\tau \leq s \leq \tau + \theta } (\Delta (\hat{Z}^{\gamma}_s)_i)^2, 
\end{equation*}
where \(\Delta f(t) = f(t) - f(t^-)\) is the jump component. Let \(\varepsilon\) be a positive real number and define the compact set \(K^{\varepsilon} = \{(r, {\bm x}) \in \hat E : r + |{\bm x}| \leq M_{\varepsilon} \text{ and } |{\bm x}| \geq m_0\}\) and the stopping time \(\eta^{\varepsilon, \tau} = \inf\{t \geq \tau : (\hat{R}_t^{\gamma}, \hat{\bf X}_t^{\gamma}) \not\in K^{\varepsilon}\}\), where \(M_{\varepsilon}\) is chosen as in Proposition \ref{prop:compactConfinment} and \(m_0\) is chosen as in \ref{it:nonExtinction}, i.e., such that \(\P{\eta^{\varepsilon, \tau} \leq T} < \varepsilon\), for all \(\gamma\) large enough. On the event \(\{\eta^{\varepsilon, \tau} > T\}\), the rate of the demographic jumps of size \(1/\gamma\) for \(\hat{\bf X}^{\gamma}\) is at most \(\gamma^2 M_{\varepsilon} (\beta_1 \vee \beta_2)(1 + \nu^{-1})\). The displacement of \(\hat{\bf X}^{\gamma}\) when a catastrophe occurs is no larger than \(M_{\varepsilon}\). 
Furthermore, by Proposition \ref{prop:piCancellation}, \(\pi_i\) must be Lipschitz on any compact set; choose \(L\) to be a Lipschitz constant for both functions on \(\{(r, {\bm x}) \in \hat E : r + |{\bm x}| \leq M_{\varepsilon} \text{ and } |{\bm x}| \geq a m_0\}\). (We choose this set so that it contains \({\bf u} \odot K_{\varepsilon}\) for all all \({\bf u}\) satisfying Assumption \eqref{ass:survivalProbability}; this is necessary for the last inequality of \eqref{eq:lipschitzInequalityGZero} to be true). And so in expectation
\begin{align*}
    \E{[D^{\tau, \gamma}]_{\theta \wedge \eta^{\varepsilon, \tau}} } &\leq C^*_{\varepsilon} \theta,
\end{align*}
where we have the positive constant \(C^*_{\varepsilon} = L^2 M_{\varepsilon} (\beta_1 \vee \beta_2)(1 + \nu^{-1}) + \kappa L^2 M^2_{\varepsilon}\). Let \(\delta\) be any small positive quantity. Then 
\begin{align*}
\P{|(\hat{Z}^{\gamma}_{\tau + \theta})_i - (\hat{Z}^{\gamma}_{\tau})_i| > \delta} &\leq \P{|(\hat{Z}^{\gamma}_{\tau + \theta \wedge \eta^{\varepsilon, \tau}})_i - (\hat{Z}^{\gamma}_{\tau})_i| > \delta} + \P{\eta^{\varepsilon, \tau} \leq T}\\
&\leq \P{|D_{\theta \wedge \eta^{\varepsilon, \tau}}^{\tau, \gamma}| > \delta/2} + \P{\int_{\tau}^{\tau + \theta \wedge \eta^{\varepsilon, \tau}}(\hat{\bf G}^{\gamma}\pi_i)(\hat{R}^{\gamma}_s,\hat{\bf X}_s^{\gamma})ds > \delta/2} + \varepsilon.
\end{align*}
For the first term, the bounded martingale \(\{D^{\tau, \gamma}_{\theta \wedge \eta^{\varepsilon, \tau}}\}_{\theta \geq 0}\) has second moment bounded by the expected quadratic variation, hence by Chebyshev's inequality 
\begin{align*}
    \sup_{0 \leq \theta \leq \theta^*}\P{|D_{\theta \wedge \eta^{\varepsilon, \tau}}^{\tau, \gamma}| > \delta/2} &\leq \frac{\E{[D^{\tau, \gamma}]_{\theta \wedge \eta^{\varepsilon, \tau}}}}{\delta^2/4}\\
    &\leq \frac{C^*_{\varepsilon}\theta^*}{\delta^2/4}. 
\end{align*}
For the second term, using Proposition \ref{prop:piCancellation} then Proposition \ref{prop:uniformConvergence}, we break down the computation as follows: 
\begin{align}
    \sup_{(r, {\bm x}) \in K^{\varepsilon}}|(\hat{\bf G}^{\gamma}\pi_i)(r, {\bm x})| &= \sup_{(r, {\bm x}) \in K^{\varepsilon} }|(\hat{\bf G}^{\gamma}_0\pi_i)(r, {\bm x})|\nonumber\\
    &\leq \sup_{(r,{\bm x}) \in K^{\varepsilon}}|(\hat{\bf G}_0 {\pi}_i)(r, {\bm x})| + o(1) \quad (\text{as } \gamma \uparrow \infty)\nonumber\\
    &\leq C^{\star}_{\varepsilon} + c_{\varepsilon},\label{eq:lipschitzInequalityGZero}
\end{align}
where \(C^{\star}_{\varepsilon} = \frac{1+\nu^{-1}}{2} M_{\varepsilon} \sup_{(r, {\bm x}) \in K^{\varepsilon}}\sum_{j=1,2}\beta_j\big|\frac{\partial^2 \pi_i}{\partial x_j^2}({\bm x})\big| + \kappa L M_{\varepsilon} < \infty\), and \(c_{\varepsilon}\) is the maximum value of the uniform remainder \(o(1)\) over \(K^{\varepsilon}\) and \(\gamma > 1\). Hence 
\begin{align*}
    \sup_{0 \leq \theta \leq \theta^*}\P{\int_{\tau}^{\tau + \theta \wedge \eta^{\varepsilon, \tau}}(\hat{\bf G}^{\gamma} \pi_i) (\hat{R}^{\gamma}_s,\hat{\bf X}^{\gamma}_s)ds > \delta/2} \leq \P{(C^{\star}_{\varepsilon} + c_{\varepsilon})\theta^* > \delta/2}, 
\end{align*}
which is equal to zero whenever \(\theta^* \leq \delta/2(C^{\star}_{\varepsilon} + c_{\varepsilon})\). As \(\varepsilon\) and \(\delta\) are arbitrary, taking \(\theta^*\) then \(\varepsilon\) to zero finishes the proof. 
\end{proof}

\section{Proof of Lemma \ref{lem:fixationProbabilityApproximation}}\label{sec:fixationProof}

\subsection{Proof of \ref{it:operatorDifferenceEstimate}}

Under the assumption of \ref{it:diagonalDistribution}, after a first order Taylor expansion of the coefficients of the first and second derivative in \({\bf A}^{\varepsilon}\), we have 
\begin{align*}
    \frac{{\bf A}^{\varepsilon}f(w)}{w(1-w)} &= \frac{\tilde{\bf A}^{\varepsilon}f(w)}{w(1-w)} + c_1(w,\varepsilon)f'(w) + c_2(w,\varepsilon)f''(w)\\ 
    &+ \frac{\kappa}{w(1-w)}\E{f \circ \Phi \circ {\bm \pi}(U^*\Phi^{-1}(w)) - f(w)} + \frac{\varepsilon\kappa}{\beta}\E{\log \frac{1}{U^*}}f'(w), 
\end{align*}
where \(c_1\) and \(c_2\) are \(o(\varepsilon)\) uniformly in \(w\). Hence \(\ref{it:operatorDifferenceEstimate}\) boils down to proving Proposition \ref{prop:smallJumpsEstimate}, stated below. 

\begin{proposition}\label{prop:smallJumpsEstimate}
    Let \(a > 0\). There exists a constant \(C\) such that
    \begin{equation*}
        \int_0^1\bigg|\frac{f \circ \Phi \circ {\bm \pi}(u\Phi^{-1}(w)) - f(w)}{w(1-w)}  +  \frac{\varepsilon}{\beta}\log \frac{1}{u}f'(w)\bigg|dw \leq C\varepsilon^2\int_0^1 |f'(w)| + |f''(w)|dw, 
    \end{equation*}
    for all \(f \in C^2[0,1]\), all \(u \geq a\) and all \(\varepsilon\) small enough. 
\end{proposition}

\begin{proof}[Proof of Proposition \ref{prop:smallJumpsEstimate}]
    Define \(H_u(w, \varepsilon) = \Phi \circ {\bm \pi}(u \Phi^{-1}(w))\). Using Taylor's Theorem, we write 
    \begin{equation*}
        f \circ H_u(w, \varepsilon) 
        = f(w) - \frac{\varepsilon}{\beta} w(1-w) \log \frac{1}{u}f'(w) + R_{f,\varepsilon}(w), 
    \end{equation*}
    where we write the Taylor remainder in integral form 
    \begin{equation*}
        R_{f,\varepsilon}(w) = \int_0^{\varepsilon} (\varepsilon-\xi)\bigg[f'' \circ H_{u}(w,\xi) \left(\frac{\partial H_u(w,\xi)}{\partial \varepsilon}\right)^2 + f' \circ H_u(w,\xi) \frac{\partial^2H_u(w,\xi)}{\partial \varepsilon^2}\bigg]d\xi. 
    \end{equation*}
    Using Proposition \ref{prop:H1DerivativeEstimates} and Fubini's Theorem, there exists a constant \(C\) such that, for all \(f\) and \(u\),  
    \begin{equation*}
        \int_0^1 \frac{|R_{f, \varepsilon}(w)|}{w(1-w)}dw \leq C\int_0^{\varepsilon}(\varepsilon-\xi) \int_0^1 |f'' \circ H_u(w,\xi)| + |f' \circ H_u(w,\xi)|dw d\xi. 
    \end{equation*}
    Consider the change of variable \(v = H_u(w,\xi)\). By Proposition \ref{prop:nonDegenerateCoordinateTransform} it is non-degenerate and the derivative \(\frac{\partial H_u(w,\xi)}{\partial w}\) is uniformly bounded from below for all \(\xi\) small enough. 
    Hence there exists a constant \(C\) such that  
    \begin{align*}
        \int_0^1 |f'' \circ H_u(w,\xi)|dw &\leq Cu^{-\frac{2|\varepsilon|}{\beta+|\varepsilon|}} \int_0^1 |f''(v)|dv. 
    \end{align*}
    And so there exists a constant \(C\) such that
    \begin{equation*}
        \int_0^1 \frac{|R_{f, \varepsilon}(w)|}{w(1-w)}dw \leq C\varepsilon^2 \int_0^1 
        |f'(v)|+|f''(v)|dv, 
    \end{equation*}
    for all \(u \geq a\) and \(\varepsilon\) small enough. 
\end{proof}

\subsection{Proof of \ref{it:stabilityBound}}

Using integration by parts and the condition \(f(0) = f(1) = 0\), we have 
\begin{equation*}
    \int_0^1 \frac{f(w)\tilde{\bf A}^{\varepsilon}f(w)}{w(1-w)}dw = -(a_{1,2} + \varepsilon a_{2,2})\int_0^1|f'(w)|^2 dw. 
\end{equation*}
Therefore, using the Cauchy--Schwarz inequality
\begin{align*}
    \left(\sup_{w \in [0,1]}|f(w)|\right)^2 &= \left(\sup_{w \in [0,1]} \left| \int_0^w f'(v)dv\right|\right)^2\\
    &\leq \left(\int_0^1 |f'(w)|dw\right)^2\\
    &\leq \int_0^1|f'(w)|^2dw\\
    &\leq \frac{1}{|a_{1,2} + \varepsilon a_{2,2}|}\int_0^1\frac{|f(w)||\tilde{\bf A}^{\varepsilon}f(w)|}{w(1-w)}dw\\
    &\leq \frac{1}{|a_{1,2} + \varepsilon a_{2,2}|} \sup_{w \in [0,1]}|f(w)|\int_0^1\frac{|\tilde{\bf A}^{\varepsilon} f(w)|}{w(1-w)}dw.
\end{align*}
The conclusion follows by dividing each side by \(\sup_{w \in [0,1]}|f(w)|\).

\subsection{Proof of \ref{it:aPrioriEstimate}}

We use a coupling argument. Denote by \(W^{w}\) the frequency process such that \(W^w_0 = w\). For any \(w,w' \in [0,1]\), there exists a coupling such that \(W^w\) and \(W^{w'}\) are driven by the same Brownian motion, and the discontinuous jump times are driven by the same Poisson process (this is possible because \(\kappa\) does not depend on the current state of \(W\)). That is, we can write 
\begin{align*}
    W^w_t = w + \int_0^t \mu(W_{s^-}^w)ds + \int_0^t \sigma(W_{s^-}^w)dB_s + \sum_{1 \leq j \leq N_t} \Delta W^w_{\tau_j},\\
    W^{w'}_t = w' + \int_0^t \mu(W_{s^-}^{w'})ds + \int_0^t \sigma(W_{s^-}^{w'})dB_s + \sum_{1 \leq j \leq N_t} \Delta W^{w'}_{\tau_j},
\end{align*}
where \(\{B_t\}_{t\geq0}\) is a one-dimensional Brownian motion and \(\tau_1, \tau_2, \ldots, \tau_{N_t}\) are the occurrence times of catastrophes, up to time \(t\). By Proposition 5.2.18 of \cite{karatzas2014brownian}, \(W_t^w \leq W_t^{w'}\) for all \(t \in [0,\tau_1)\) provided \(w \leq w'\). Let us also choose the coupling such that, at each catastrophic event, the survival probabilities for \(W^w\) and \(W^{w'}\) are driven by the same random variable \({\bf U}\). Furthermore \(W^w_{\tau_1} \leq W^{w'}_{\tau_1}\), as \(\Phi \circ {\bm \pi} \circ {\bf u} \odot {\bm \Phi}^{-1}\) is increasing for all \({\bf u} \in (0,1]^2\) by Proposition \ref{prop:nonDegenerateCoordinateTransform}. By using the strong Markov property and performing induction over the remaining intervals \([\tau_1, \tau_2], [\tau_2, \tau_3], \ldots\) we conclude that \(W^w_t \leq W_t^{w'}\) for all \(t \in [0, \infty)\). As such, \(\theta^{\varepsilon}(w) \leq \theta^{\varepsilon}(w')\).  

\section{Discussion}\label{sec:discussion}

We investigated the competition of two genetic types or species with differing life-history strategy for a single growth-limiting resource. We showed that, in the regime of abundant resources (or equivalently large carrying capacity), the relative abundance of each type evolves according to the stochastic process with generator \eqref{eq:relativeAbundanceGenerator}. Working from this approximation, we estimated a measure of evolutionary success in the fixation probability of one type against another (Equation~\eqref{eq:H2FixationProbability} and Theorem~\ref{thm:fixationProbability}).

\subsection{An Evolutionary Risk-Dominant Strategy} 
\label{sec:risk}
To describe the effect of \(\beta\) on the evolutionary success of populations, we borrow some terminology from evolutionary game theory \cite{nowak2004emergence}. 

We say that a certain phenotype \emph{A} is favored by selection over another type \emph{B} when the probability of a small number of invading mutants of type \emph{A} successfully replacing a population of type \emph{B} is greater than the probability of a small number of invading mutants of type \emph{B} successfully replacing a population of type \emph{A}. If \(\theta(w)\) denotes the fixation probability of \emph{A} given that its initial frequency in the population is \(w\), this condition can be expressed as 
\begin{equation}\label{eq:condFavoredBySelection}
    \theta(w) \geq 1 - \theta(1-w).
\end{equation}
for all sufficiently small \(w\). When \emph{A} is favored by selection over all other types, we say that it is a \textit{risk-dominant strategy}.
In the model we have studied, the possible phenotypes are indexed by the birth rate parameter \(\beta\). Let us first study the case of \ref{it:diagonalDistribution}. According to Theorem \ref{thm:fixationProbability}, the condition \eqref{eq:condFavoredBySelection} for \(\beta\) to be favored by selection over \(\beta + \varepsilon\), when \(\varepsilon\) is sufficiently small, takes the form
\begin{align*}
    \beta^2 &\geq \nu \mathcal{E}_{\Gamma} \kappa\E{\log \frac{1}{U^*}} \text{ if \(\varepsilon > 0\)},\\
    \beta^2 &\leq  \nu \mathcal{E}_{\Gamma} \kappa\E{\log \frac{1}{U^*}} \text{ if \(\varepsilon < 0\)}. 
\end{align*}
Hence the unique risk-dominant strategy corresponds to 
\begin{equation*}
    \beta_{\text{RDS}}^2 = \nu \mathcal{E}_{\Gamma} \kappa\E{\log \frac{1}{U^*}}. 
\end{equation*}
Without carrying out the calculations as we have done in Theorems~\ref{thm:convergenceToManifold} and~\ref{thm:fixationProbability}, it is not obvious at all that there should exist a non-trivial (different from \(\beta = 0\) and \(\beta = +\infty\)) RDS.

In the case of \ref{it:proportionalDistribution}, \eqref{eq:H2FixationProbability} implies that \eqref{eq:condFavoredBySelection} is satisfied if and only \(\varepsilon \geq 0\). So that \(\beta_{\mathrm{RDS}} = 0\). Of course, \(\beta_{\mathrm{RDS}}=0\) is a singular solution, not corresponding to a feasible model; it must be understood to suggest that the main resource-constraint fitness pressure analyzed here pushes toward ever-slower life-histories until physical constraints block further evolution in this direction.

In all cases, these results establish a strong link between the nature and intensity of those mass-mortality events, and the selective pressures leading to longer and slower life cycles.

\subsection{On \emph{r}/\emph{K} and Density-Dependent Selection}

The case \(\kappa = 0\) recovers many of the predictions of \emph{r}/\emph{K}-selection theory. Suppose we introduce a small quantity of individuals of type \(\beta\) and \(\beta + \varepsilon\) in a population. At first, \eqref{eq:phiPiFirstOrderExpansion} predicts that we will observe an increase in the frequency of type \(\beta + \varepsilon\) up until the point where the population reaches the carrying capacity dictated by available resources. Once carrying capacity is reached, we find ourselves in the setting of Theorem \ref{thm:fixationProbability}, and \eqref{eq:condFavoredBySelection} states that type \(\beta\) is favored over \(\beta + \varepsilon\), meaning there is a reversal of the direction of selection. 

It should be said that the theory of \emph{r}/\emph{K}--selection has lost currency, in favor of models which incorporate more complex patterns of reproduction and mortality, such as age-specific mortality and fecundity, growth, complex life cycles and inter-generational resource transfers \cite{reznick2002r, lande2017evolution, wright2019life, wright2020contrasting, travis2023density, stott2024life}, and a complete treatment of the questions we raised would need to incorporate these features. Nonetheless, we found this simple framework to be sufficient for the emergence of complex evolutionary dynamics, as evidenced by the existence of a risk-dominant strategy.  Furthermore, separation of time scales arguments such as the one we rigorously established also play a very important role in the analysis of those more sophisticated models. 

\appendix

\section{Stability Analysis of the Reduced Chemostat Equation}\label{sec:stability}

In this section we will use Poincar\'e-Bendixson theory to support the claims made in Section \ref{sec:deterministicPrologue} regarding the reduced chemostat equation \eqref{eq:chemostatEquationReduced}.  

If \(\beta_1F(R_{\max}) < \mu_1\) (or \(\beta_2F(R_{\max}) < \mu_2\)), type 1 (respectively type 2) goes extinct irrespective of the behavior of the other type. Let us assume neither is the case. There exists a positive constant \(C\) such that 
\begin{equation*}
    \begin{split}
        \frac{d}{dt} \bm{\beta} \cdot {\bf X} &\leq C\bigg(R_{\max} - \frac{1}{\nu_1}{\bm \beta} \cdot {\bf X}\bigg), 
    \end{split}
\end{equation*}
from which we conclude that, given an initial condition, the trajectory never leaves a bounded set \(D \subseteq \RR_{\geq 0}^2\). Recall the Poincar\'e--Bendixson theorem: given a differentiable real dynamical system, any compact \(\omega\)-limit set contains only fixed points, periodic orbits and homoclinic or heteroclinic orbits joining fixed points. The fixed points of our system can be easily identified as \({\bf a}_1 = ((F \circ R^*)^{-1}(1/\nu_1)/\beta_1, 0)\), \({\bf a}_2 = (0, (F \circ R^*)^{-1}(1/\nu_2)/\beta_2)\) and \(\bf 0\). Hence to reach the desired conclusion it suffices to prove the following:
\begin{enumerate}
    \item There are no flow-invariant closed curves in \(\RR_{\geq 0}^2\); and
    \item \(\nu_1 > \nu_2\) implies that \({\bf a}_1\) is stable and \({\bf a}_2\) is a saddle point;
    \item \(\nu_1 > 1/F(R_{\max})\) implies that \(\bf 0\) is unstable.
\end{enumerate}
For 1. we use the Bendixson-Dulac criterion. Defining the Dulac function \(\varphi({\bm x}) = \frac{1}{x_1x_2}\) for \({\bm x} \in \RR_{>0}^2\), such that 
\begin{align*}
    \sum_{i=1,2}\frac{\partial}{\partial x_i}\left[\varphi({\bm x})\beta_i x_i\bigg(F\circ R^*(\bm{\beta} \cdot {\bm x})-\nu_i^{-1}\bigg)\right]  &=  \frac{x_1\beta_1^2 + x_2\beta_2^2}{x_1x_2}(F \circ R^*)'({\bm \beta} \cdot {\bm x}),
\end{align*}
which is always strictly negative. Hence there are no invariant closed curves wholly contained in \(\RR_{> 0}^2\). As the axes are flow invariant and the solution is unique, this can be extended to \(\RR_{\geq 0}^2\). 

For 2. and 3., we simply compute the Jacobian of the flow at those points: 
\begin{align*}
    J({\bf{a}}_1) &= \begin{pmatrix}
        \beta_1^2({\bf a}_1)_1 (F \circ R^*)'({\bm \beta}  \cdot {\bf a}_1) & \beta_1\beta_2\nu_1^{-1}({\bf a}_1)_1 \\
        0 &  \beta_2(\nu_1^{-1} - \nu_2^{-1})
    \end{pmatrix},\\
    J({{\bf a}}_2) &= \begin{pmatrix}
        \beta_1 (\nu_2^{-1} - \nu_1^{-1})& 0\\
        \beta_1\beta_2\nu_2^{-1}({\bf a}_2)_2 & \beta_2^2({\bf a}_2)_2 (F \circ R^*)'({\bm \beta}  \cdot {\bf a}_2)
    \end{pmatrix},\\
    J({\bf 0})&= \begin{pmatrix}
        \beta_1F(R_{\max}) - \mu_1 & 0 \\
        0 & \beta_2F(R_{\max}) - \mu_2
    \end{pmatrix}.
\end{align*}
The first matrix has all negative eigenvalues, the second has eigenvalues of opposite signs, and the third has all positive eigenvalues. Hence, the only attractive limit point is \({\bf a}_1\). It follows that for any \({\bf X}(0) \in \RR_{\geq 0} \times \RR_{> 0} \setminus \{(0,0)\}\), one has
\begin{equation*}
    \lim_{t \rightarrow \infty} {\bf X}(t) = {\bf a}_1. 
\end{equation*}

\section{Calculus of the Projection Map}

\subsection{Derivatives with Respect to the Initial Conditions}

\begin{proposition}\label{prop:firstDerivatives}
    Let \(J_{\bm \pi}\) be the Jacobian of \(\bm \pi\), i.e., \(J_{\bm \pi} = \begin{pmatrix}
        \frac{\partial\pi_1}{\partial x_1} & \frac{\partial\pi_1}{\partial x_2} \\
        \frac{\partial\pi_2}{\partial x_1} & \frac{\partial\pi_2}{\partial x_2}
    \end{pmatrix}\), and let \(J^2_{\pi_1} = \begin{pmatrix}
        \frac{\partial^2 \pi_1}{\partial x_1^2} & \frac{\partial^2 \pi_1}{\partial x_1 \partial x_2}\\ \frac{\partial^2 \pi_1}{\partial x_1 \partial x_2}& \frac{\partial^2 \pi_1}{\partial x_2^2}
        \end{pmatrix}\)
        and 
    \(J^2_{{\pi}_2} =
        \begin{pmatrix}
        \frac{\partial^2 \pi_2}{\partial x_1^2} & \frac{\partial^2 \pi_2}{\partial x_1 \partial x_2}\\ \frac{\partial^2 \pi_2}{\partial x_1 \partial x_2}& \frac{\partial^2 \pi_2}{\partial x_2^2}
        \end{pmatrix}\) be the Hessian matrices of each component. We can compute 
    \begin{equation}\label{eq:jacobianOfPi}
        J_{\bm \pi} = \frac{\pi_1\pi_2}{B}\begin{pmatrix}
            \beta_2^2x_1^{-1}& - \beta_1\beta_2x_2^{-1}\\
            -\beta_1\beta_2x_1^{-1} & \beta_1^2x_2^{-1}
        \end{pmatrix},
    \end{equation}
    where \(B=\beta_1^2\pi_1 + \beta_2^2\pi_2\), and
    \begin{equation}\label{eq:projectionHessian}
    \begin{split}
        \beta_1J_{{\pi}_1}^2 = -\beta_2J_{{\pi}_2}^2 =  &
        \frac{\beta_1\beta_2\pi_1\pi_2(\beta_2^3 \pi_2^2-\beta_1^3\pi_1^2)}{B^3}
        \begin{pmatrix}
            \beta_2^2 x_1^{-2} & -\beta_1\beta_2 x_1^{-1}x_2^{-1}\\
            -\beta_1\beta_2 x_1^{-1}x_2^{-1} & \beta_1^2 x_2^{-2}
        \end{pmatrix}\\
        &- \frac{\beta_1\beta_2\pi_1\pi_2}{B}\begin{pmatrix}
            \beta_2 x_1^{-2} & 0 \\
            0 & -\beta_1x_2^{-2}
        \end{pmatrix}.
    \end{split}
    \end{equation}
\end{proposition}

\begin{corollary}\label{cor:gradientOnManifold}
    Evaluated on the manifold \(\Gamma\) (such that \((\pi_1, \pi_2) = (x_1, x_2)\)), the Jacobian and the Hessian simplify to 
    \begin{align*}
        J_{\bm \pi} &= \frac{1}{\beta_1^2x_1 + \beta_2^2x_2}\begin{pmatrix}
            \beta_2^2x_2 & -\beta_1\beta_2x_1\\
            -\beta_1\beta_2x_2 & \beta_1^2x_1
        \end{pmatrix},\\
        J^2_{{\pi}_1} &= \frac{\beta_1\beta_2}{(\beta_1^2 x_1 + \beta_2^2x_2)^3}\begin{pmatrix}
            - \beta_1\beta_2x_2(\beta_1(\beta_1 + \beta_2)x_1 + 2\beta_2^2x_2) & \beta_1^3\beta_2x_1^2 - \beta_2^4x_2^2\\
            \beta_1^3\beta_2x_1^2 - \beta_2^4x_2^2 & \beta_2^2x_1(2\beta_1^2x_1 + \beta_2(\beta_1+\beta_2)x_2)
        \end{pmatrix}.
    \end{align*}
\end{corollary}

\begin{proof}[Proof of Proposition \ref{prop:firstDerivatives}]
    Knowing \(\bm \pi\) must satisfy \eqref{eq:implicitSystem}, we implicitly differentiate the system to obtain the linear equation
    \begin{equation*}
        \begin{pmatrix}
            \beta_1 & \beta _2\\
            -\beta_2\pi_1^{-1} & \beta_1\pi_2^{-1}
        \end{pmatrix}J_{\bm \pi} = \begin{pmatrix}
            0 & 0\\
            -\beta_2x_1^{-1}& \beta_1x_2^{-1}. 
        \end{pmatrix}
    \end{equation*}
    Solving gives \eqref{eq:jacobianOfPi}. Considering the second derivatives of the same system we obtain 
    \begin{align*}
        -\frac{\beta_2}{\pi_1}
        J^2_{\pi_1} + \frac{\beta_1}{\pi_2}J^2_{\pi_2} &= - 
            \frac{\beta_2}{\pi_1^2} \begin{pmatrix}
                \left(\frac{\partial \pi_1}{\partial x_1}\right)^2 & \frac{\partial \pi_1}{\partial x_1}\frac{\partial \pi_1}{\partial x_2}\\
                \frac{\partial \pi_1}{\partial x_2}\frac{\partial \pi_2}{\partial x_1} & \left(\frac{\partial \pi_1}{\partial x_2}\right)^2
            \end{pmatrix} 
            + \frac{2\beta_1\beta_2}{\pi_1\pi_2}\begin{pmatrix}
                \frac{\partial \pi_1}{\partial x_1}\frac{\partial \pi_2}{\partial x_1} & \frac{\partial \pi_1}{\partial x_1}\frac{\partial \pi_2}{\partial x_2}\\
                \frac{\partial \pi_1}{\partial x_1}\frac{\partial \pi_2}{\partial x_2} & \frac{\partial \pi_1}{\partial x_2}\frac{\partial \pi_2}{\partial x_2}
            \end{pmatrix}\\
        &-\frac{\beta_1}{\pi_2^2}\begin{pmatrix}
                \left(\frac{\partial \pi_2}{\partial x_1}\right)^2 & \frac{\partial \pi_2}{\partial x_1}\frac{\partial \pi_2}{\partial x_2}\\
                \frac{\partial \pi_2}{\partial x_2}\frac{\partial \pi_2}{\partial x_1} & \left(\frac{\partial \pi_2}{\partial x_2}\right)^2
            \end{pmatrix} + \begin{pmatrix}
            \beta_2x_1^{-2} & 0 \\
            0 & \beta_1(\beta_1-1)x_2^{-2}
        \end{pmatrix}
    \end{align*}
    with the constraint \(\beta_1J^2_{{\pi}_1} = - \beta_2J^2_{{\pi}_2}\). 
    After substituting \eqref{eq:jacobianOfPi}, it can be simplified to \eqref{eq:projectionHessian}. 
\end{proof}

\subsection{Intermediate Generator Computations}\label{sec:intermediateGenerator}

Applying \({\bf L}\) to \(f \circ \Phi\) yields the expression
    \begin{align*}
        {\bf L} f({\bm z}) &= \frac{1}{\nu}\sum_{i,j=1}^2\beta_iz_i \frac{\partial^2 \pi_j}{\partial x_i^2}\frac{\partial \Phi}{\partial z_j}f'\circ \Phi({\bm z}) + \frac{1}{\nu} \sum_{i,j,k=1}^2\beta_i z_i \frac{\partial \pi_j}{\partial x_i}\frac{\partial \pi_k}{\partial x_i} \frac{\partial^2 \Phi}{\partial z_k \partial z_j}f' \circ \Phi({\bm z})\\
        &+ \frac{1}{\nu} \sum_{i,j,k=1}^2\beta_i z_i \frac{\partial \pi_j}{\partial x_i}\frac{\partial \pi_k}{\partial x_i} \frac{\partial \Phi}{\partial z_k}\frac{\partial \Phi}{\partial z_j}f'' \circ \Phi({\bm z})\\
        &+ \kappa\int\displaylimits_{[0,1]^2} [f \circ \Phi \circ {\bm \pi}({\bf u} \odot {\bm z}) - f \circ \Phi \circ {\bm \pi}({\bm z})] \Upsilon(d{\bf u}).
    \end{align*}
    Using the relationship \(\beta_1 z_1 + \beta_2 z_2 = \mathcal{E}_{\Gamma}\), we can write 
    \begin{align*}
        z_1 &= \frac{w\mathcal{E}_{\Gamma}}{\beta_1 w + \beta_2(1-w)},\\
        z_2 &= \frac{(1-w)\mathcal{E}_{\Gamma}}{\beta_1 w + \beta_2(1-w)},
    \end{align*}
    where \(w = \Phi(z_1, z_2)\). Substituting this into the expressions of Corollary~\ref{cor:gradientOnManifold} we find
    \begin{align*}
        J_{\bm \pi} &= \frac{1}{\beta_1^2w + \beta_2^2(1-w)}\begin{pmatrix}
            \beta_2^2(1-w) & -\beta_1\beta_2w\\
            -\beta_1\beta_2(1-w) & \beta_1^2w
        \end{pmatrix},\\
        J^2_{{\pi}_1} &= \frac{\beta_1\beta_2(\beta_1 w + \beta_2(1-w))}{\mathcal{E}_{\Gamma}(\beta_1^2 w + \beta_2^2(1-w))^3}\begin{pmatrix}
            - \beta_1^2\beta_2(\beta_1 + \beta_2)w(1-w) - 2\beta_1\beta_2^3(1-w)^2 & \beta_1^3\beta_2w^2 - \beta_2^4(1-w)^2\\
            \beta_1^3\beta_2w^2 - \beta_2^4(1-w)^2 & 2\beta_1^2\beta_2^2w^2 + \beta_2^3(\beta_1+\beta_2)w(1-w)
        \end{pmatrix},\\
        J^2_{{\pi}_2} &= \frac{\beta_1\beta_2(\beta_1 w + \beta_2(1-w))}{\mathcal{E}_{\Gamma}(\beta_1^2 w + \beta_2^2(1-w))^3}\begin{pmatrix}
             \beta_1^3(\beta_1 + \beta_2)w(1-w) + 2\beta_1^2\beta_2^2(1-w)^2 & -\beta_1^4w^2 + \beta_1\beta_2^3(1-w)^2\\
            -\beta_1^4w^2 + \beta_1\beta_2^3(1-w)^2 & -2\beta_1^3\beta_2w^2 - \beta_1\beta_2^2(\beta_1+\beta_2)w(1-w)
        \end{pmatrix}
    \end{align*}
    Similarly we can express the partial derivatives of \(\Phi\) as
    \begin{align*}
        \begin{pmatrix}\frac{\partial \Phi}{\partial z_1} &\frac{\partial \Phi}{\partial z_2} \end{pmatrix} &= \frac{\beta_1w + \beta_2(1-w)}{\mathcal{E}_{\Gamma}}\begin{pmatrix}
            1-w & -w\\
        \end{pmatrix},\\
        \begin{pmatrix}
            \frac{\partial^2 \Phi}{\partial z_1^2} & \frac{\partial^2 \Phi}{\partial z_1 \partial z_2}\\
            \frac{\partial^2 \Phi}{\partial z_1 \partial z_2}& \frac{\partial^2 \Phi}{\partial z_2^2} \\
        \end{pmatrix} &= \frac{(\beta_1w + \beta_2(1-w))^2}{\mathcal{E}_{\Gamma}^2}\begin{pmatrix}
            -2(1-w) & -(1-2w)\\
            -(1-2w) & 2w\\
        \end{pmatrix}. 
    \end{align*}
    It then suffices to add up the coefficients 
    \begin{align*}
        \sum_{i,j=1}^2\beta_iz_i \frac{\partial^2 \pi_j}{\partial x_i^2}\frac{\partial \Phi}{\partial z_j} &= \frac{\beta_1\beta_2 (\beta_2-\beta_1)w(1-w)(\beta_1 w + \beta_2(1-w))^2}{\mathcal{E}_{\Gamma}(\beta_1^2w + \beta_2^2(1-w))^3}\bigg[\beta_2^2(2\beta_1 + \beta_2)(1-w)  + \beta_1^2(\beta_1 + 2\beta_2)w \bigg],\\
        \sum_{i,j,k=1}^2\beta_i z_i \frac{\partial \pi_j}{\partial x_i}\frac{\partial \pi_k}{\partial x_i} \frac{\partial^2 \Phi}{\partial z_k \partial z_j} &= (-2)\frac{\beta_1\beta_2(\beta_2- \beta_1)w(1-w)(\beta_1 w + \beta_2(1-w))^2}{\mathcal{E}_{\Gamma}(\beta_1^2 w + \beta_2^2(1-w))^3}(\beta_1 w + \beta_2(1-w))(\beta_1^2 w + \beta_2^2(1-w)),\\
        \sum_{i,j,k=1}^2\beta_i z_i \frac{\partial \pi_j}{\partial x_i}\frac{\partial \pi_k}{\partial x_i} \frac{\partial \Phi}{\partial z_k}\frac{\partial \Phi}{\partial z_j} &= \frac{\beta_1\beta_2w(1-w)(\beta_1 w + \beta_2(1-w))^4}{\mathcal{E}_{\Gamma}(\beta_1^2 w + \beta_2^2(1-w))^2}.
    \end{align*}

\subsection{Bounds on the Growth Factor}

\begin{proposition}\label{prop:growthFactorBound}
    Suppose \({\bm \beta} \cdot {\bm x} \leq \mathcal{E}_{\Gamma}\), \(\beta_1 = \beta\) and \(\beta_2 = \beta + \varepsilon\) with \(\varepsilon\) positive. Then the following inequality holds: 
    \begin{align}
            \left(\frac{\mathcal{E}_{\Gamma}}{{\bm \beta} \cdot {\bm x}}\right)^{\frac{\beta}{\beta + \varepsilon}} &\leq \frac{\pi_1({\bm x})}{x_1} \leq \frac{\mathcal{E}_{\Gamma}}{{\bm \beta} \cdot {\bm x}} \leq \frac{\pi_2({\bm x})}{x_2} \leq  \left(\frac{\mathcal{E}_{\Gamma}}{{\bm \beta} \cdot {\bm x}}\right)^{\frac{\beta+ \varepsilon}{\beta}}\label{eq:growthFactorBound}.
    \end{align}
    All inequalities are reversed if \(\varepsilon\) is negative. 
\end{proposition}

\begin{proof}[Proof of Proposition \ref{prop:growthFactorBound}]
    Let \(\hat{\bf X}\) solve \eqref{eq:chemostatEquationReducedDegenerate} with initial condition \({\bm x}\). 
    Introduce the auxiliary variable \(S = {\bm \beta} \cdot \hat{\bf X}\). One sees that 
    \begin{equation*}
        \frac{dS}{dt} = (F \circ R^*({\bm \beta} \cdot \hat{\bf X}) - 1/\nu) \sum_{i=1,2}\beta_i^2 \hat{X}_i.
    \end{equation*}
    Since 
    \begin{align*}
        \log \hat{X}_1(T) 
        &= \log x_1 + \int_0^T \frac{1}{\hat{X}_1(t)} \frac{d\hat{X}_1(t)}{dt}dt\\
        &= \log x_1 + \int_0^T \beta_1(F \circ R^*({\bm\beta} \cdot \hat{\bf X}(t))- \nu^{-1})dt, 
    \end{align*}
    taking the limit as \(T \uparrow \infty\) yields
    \begin{align*}
        \log\pi_1({\bm x}) &= \log x_1 + \int_0^{\infty} \beta_1 \bigg(F \circ R^*({\bm \beta} \cdot \hat{\bf X}(t)) - 1/\nu\bigg)dt\\
        &= \log x_1 + \int_{\bm\beta\cdot{\bm x}}^{\mathcal{E}_{\Gamma}} \frac{\beta_1}{\sum_{i=1,2}\beta_i^2 \hat{X}_i} dS.
    \end{align*}
    Notice the bounds
    \begin{equation*}
        \frac{\beta}{\beta + \varepsilon}\int^{\mathcal{E}_{\Gamma}}_{{\bm \beta}\cdot{\bm x}} \frac{dS}{S} \leq \int_{{\bm \beta}\cdot{\bm x}}^{\mathcal{E}_{\Gamma}} \frac{\beta}{\sum_{i=1,2}\beta_i^2 \hat{X}_i} dS \leq \int^{\mathcal{E}_{\Gamma}}_{{\bm \beta}\cdot{\bm x}} \frac{dS}{S}. 
    \end{equation*}
    A similar calculation can be done for \(\pi_2\). Evaluating the integrals and exponentiating yields \eqref{eq:growthFactorBound}. 
\end{proof}

\subsection{Non-Degeneracy of the Change of Coordinates}

    \begin{proposition}\label{prop:coordinateTransformDerivative}
    Define \(\varphi_{\bf u}(w) = \Phi \circ {\bm \pi}({\bf u} \odot \Phi^{-1}(w))\). Then
    \begin{equation*}
        \frac{d}{dw}\varphi_{\bf u}(\varepsilon, w) = \frac{1}{w(1-w)}\frac{\pi_1 \pi_2}{(\pi_1 + \pi_2)^2}({\bf u} \odot \Phi^{-1}(w)). 
    \end{equation*}
    Note this applies to \(H_u(\varepsilon, w)\) with \({\bf u} = (u,u)\). 
    \end{proposition}
    \begin{proposition}\label{prop:nonDegenerateCoordinateTransform}
        Suppose \(\beta_1 = \beta\) and \(\beta_2 = \beta + \varepsilon\). For all \(\varepsilon\) small enough, 
        \begin{equation*}
            \frac{d}{dw}\varphi_{\bf u}(\varepsilon, w) \geq \frac{(\beta-|\varepsilon|)^6}{4\mathcal{E}_{\Gamma}^{3}(\beta + |\varepsilon|)} \frac{u_1u_2}{(u_1 \vee u_2)^{\frac{2\beta}{\beta + |\varepsilon|}}}. 
        \end{equation*}
        Note this applies to \(H_u(\varepsilon, w)\) with \({\bf u} = (u,u)\).  
    \end{proposition}

    \begin{proof}[Proof of Proposition \ref{prop:coordinateTransformDerivative}]
    Recall that \(\Phi^{-1}\) is defined as 
    \begin{equation*}
        \Phi^{-1}(w) = \frac{\mathcal{E}_{\Gamma}}{\beta_1w + \beta_2 (1-w)}(w, 1-w).
    \end{equation*}
    This is a straightforward application of the chain rule. Let us start by evaluating 
    \begin{align*}
        \frac{d}{dw}\big(\pi_1({\bf u } \odot\Phi^{-1}(w))\big) &= \frac{\mathcal{E}_{\Gamma}}{[\beta_1 w + \beta_2(1-w)]^2} \bigg(\beta_2u_1\frac{\partial\pi_1}{\partial x_1}({\bf u} \odot\Phi^{-1}(w)) - \beta_1u_2\frac{\partial \pi_1}{\partial x_2}({\bf u} \odot\Phi^{-1}(w))\bigg)\\
        &= \frac{\beta_2\pi_1\pi_2}{B}\frac{\beta_1^2w + \beta_2^2(1-w)}{w(1-w)(\beta_1w + \beta_2(1-w))},\\
\frac{d}{dw}\big(\pi_2({\bf u} \odot \Phi^{-1}(w))\big) 
        &= \frac{\mathcal{E}_{\Gamma}}{[\beta_1 w + \beta_2(1-w)]^2} \bigg(\beta_2u_1\frac{\partial\pi_2}{\partial x_1}({\bf u} \odot \Phi^{-1}(w)) - \beta_1u_2\frac{\partial \pi_2}{\partial x_2}({\bf u} \odot \Phi^{-1}(w))\bigg)\\
        &= -\frac{\beta_1\pi_1\pi_2}{B}\frac{\beta_1^2w + \beta_2^2(1-w)}{w(1-w)(\beta_1w + \beta_2(1-w))}.
    \end{align*}
    Thus
    \begin{align*}
        \frac{d}{dw}\varphi_{\bf u}(w) &= \frac{\pi_2}{(\pi_1 + \pi_2)^2}\frac{d}{dw}\big(\pi_1(u\Phi^{-1}(w))\big) - \frac{\pi_1}{(\pi_1 + \pi_2)^2}\frac{d}{dw}\big(\pi_2(u\Phi^{-1}(w))\big)\\
        &= \frac{\pi_1 \pi_2}{(\pi_1 + \pi_2)^2}\frac{\mathcal{E}_{\Gamma}(\beta_1^2w + \beta_2^2(1-w))}{Bw(1-w)(\beta_1w + \beta_2(1-w))}.
    \end{align*}
    \end{proof}

    \begin{proof}[Proof of Proposition \ref{prop:nonDegenerateCoordinateTransform}]
        From Proposition \ref{prop:coordinateTransformDerivative} recall that 
        \begin{equation*}
            \frac{d}{dw}\varphi_{\bf u}(w) = \frac{\mathcal{E}_{\Gamma}(\beta_1^2w + \beta_2^2(1-w))/(\beta_1w + \beta_2(1-w))}{(\pi_1({\bf u} \odot\Phi^{-1}(w)) + \pi_2({\bf u}\odot\Phi^{-1}(w)))^2}\frac{\pi_1(u\Phi^{-1}(w))\pi_2(u\Phi^{-1}(w))}{w(1-w)}. 
        \end{equation*}
        The first factor can easily been seen to be bounded below by \(\frac{\beta^4 \wedge (\beta + \varepsilon)^4}{\beta \vee (\beta + \varepsilon)} \mathcal{E}_{\Gamma}^{-1}\). Using Proposition \ref{prop:growthFactorBound} we write
        \begin{align*}
            \frac{\pi_1({\bf u} \odot\Phi^{-1}(w))}{w} &= \frac{u_1[\beta_1 w + \beta_2(1-w)]}{\mathcal{E}_{\Gamma}} \frac{\pi_1({\bf u} \odot\Phi^{-1}(w))}{u_1\Phi_1^{-1}(w)}\\
            &\geq \frac{(\beta  - |\varepsilon|)u_1}{\mathcal{E}_{\Gamma}}\left(\frac{\mathcal{E}_{\Gamma}}{(u_1 \vee u_2){\bm \beta} \cdot \Phi^{-1}(w)}\right)^{\frac{\beta}{\beta+|\varepsilon|}}\\
            &= \frac{(\beta -|\varepsilon|)}{2\mathcal{E}_{\Gamma}} \frac{u_1}{(u_1 \vee u_2)^{\frac{\beta}{\beta+|\varepsilon|}}}
        \end{align*}
        and likewise for \(\pi_2(u\Phi^{-1}(w))/(1-w)\). 
    \end{proof}

\subsection{Derivatives with Respect to \texorpdfstring{\(\varepsilon\)}{epsilon}}

\begin{proposition}\label{prop:piEpsilonDerivatives}
    Define \({\bm J}(\varepsilon) := {\bm \pi}(u\Phi^{-1}(w))\). Let us write \(f \ll g\) to mean there exists a constant \(C\) depending only on \(u\) and \(\beta\) such that for all \(|\varepsilon| \leq \beta/2\) it holds that \(f(\varepsilon) \leq Cg(\varepsilon)\). Then     
    \begin{equation*}
        \bigg|\frac{dJ_1}{d\varepsilon}\bigg| \ll |J_1||J_2|, \quad \bigg|\frac{dJ_2}{d\varepsilon}\bigg| \ll |J_1||J_2| + |J_2|^2,
    \end{equation*}
    and
    \begin{equation*}
        \bigg|\frac{d^2J_1}{d\varepsilon^2}\bigg| \ll |J_1||J_2|, \quad \bigg|\frac{d^2J_2}{d\varepsilon^2}\bigg| \ll |J_1||J_2|+|J_2|^2.
    \end{equation*}
\end{proposition}

\begin{proof}
    Recall that \(\bm \pi\) must satisfy \eqref{eq:implicitSystem}. Substituting \({\bm x} = u\Phi^{-1}(w)\), and applying the logarithm to the second line, \eqref{eq:implicitSystem} takes the form 
    \begin{align*}
        \beta J_1(\varepsilon) + (\beta+ \varepsilon) J_2(\varepsilon)  &= \mathcal{E}_{\Gamma},\\
        \beta \log{J_2(\varepsilon)} - (\beta+\varepsilon)\log{J_1(\varepsilon)} &= \beta \log \big(u\Phi^{-1}(w)\big)_2 - (\beta + \varepsilon) \log \big(u\Phi^{-1}(w)\big)_1.    
    \end{align*}
    Implicitly differentiating the system with respect to \(\varepsilon\) yields 
    \begin{align*}
        \beta \frac{dJ_1}{d\varepsilon}(\varepsilon) + (\beta + \varepsilon) \frac{dJ_2}{d\varepsilon}(\varepsilon) + J_2(\varepsilon) &= 0,\\
        \frac{\beta}{J_2(\varepsilon)}\frac{dJ_2}{d\varepsilon}(\varepsilon) - \frac{\beta + \varepsilon}{J_1(\varepsilon)}\frac{dJ_1}{d\varepsilon}(\varepsilon) - \log J_1(\varepsilon)  &= \zeta(\varepsilon),
    \end{align*}
    where \(\zeta(\varepsilon) = -\log \mathcal{E}_{\Gamma }u w + \frac{d}{d\varepsilon} \bigg(\varepsilon \log (\beta + \varepsilon(1-w))\bigg) =-\log \mathcal{E}_{\Gamma }u w +  \log(\beta + \varepsilon(1-w)) + \frac{\varepsilon(1-w)}{\beta + \varepsilon(1-w)}\). Solving this linear system for \(\frac{dJ_1}{d\varepsilon}\) and \(\frac{dJ_2}{d\varepsilon}\) yields 
    \begin{align}
        \label{eq:J1Derivative}\frac{dJ_1}{d\varepsilon}(\varepsilon) &= -\frac{J_1(\varepsilon)J_2(\varepsilon)}{B(\varepsilon)}\bigg(\beta + (\beta + \varepsilon)(\log J_1(\varepsilon) + \zeta(\varepsilon))\bigg),\\  
        \frac{dJ_2}{d\varepsilon}(\varepsilon) &= \frac{J_1(\varepsilon)J_2(\varepsilon)}{B(\varepsilon)}\bigg(\beta(\log J_1(\varepsilon) + \zeta(\varepsilon))\bigg)\nonumber\\
        &\label{eq:J2Derivative}\qquad -  \frac{J_2(\varepsilon)^2}{B(\varepsilon)}\bigg(\beta+\varepsilon\bigg),
    \end{align}
    where \(B(\varepsilon) = \beta^2 J_1(\varepsilon) + (\beta+\varepsilon)^2J_2(\varepsilon)\). Since 
    \begin{align*}
        \log J_1(\varepsilon) + \zeta(\varepsilon) &= \log \frac{\pi_1(u\Phi^{-1}(w))}{(u\Phi^{-1}(w))_1} + \frac{\varepsilon(1-w)}{\beta + \varepsilon(1-w)},
    \end{align*}
    we have \(|\log J_1(\varepsilon) + \zeta(\varepsilon)| \leq c_u = 1 + |\log \mathcal{E}_{\Gamma}| - \log u\) by Proposition \ref{prop:growthFactorBound}. Using the fact that \(B(\varepsilon) \geq \beta \mathcal{E}_{\Gamma}\) and \(|\varepsilon| \leq \beta/2\) it follows that 
    \begin{align*}
        \bigg|\frac{dJ_1}{d\varepsilon}(\varepsilon)\bigg| &\leq \frac{2+2c_u}{\mathcal{E}_{\Gamma}}|J_1(\varepsilon)||J_2(\varepsilon)|,\\  
        \bigg|\frac{dJ_2}{d\varepsilon}(\varepsilon)\bigg| &\leq \frac{2+2c_u}{\mathcal{E}_{\Gamma}}|J_1(\varepsilon)||J_2(\varepsilon)| + \frac{2}{\mathcal{E}_{\Gamma}}|J_2(\varepsilon)|^2.
    \end{align*}
    Differentiating \eqref{eq:J1Derivative} yields 
    \begin{align*}
        \frac{d^2J_1}{d\varepsilon^2}(\varepsilon) &= - \frac{J_1(\varepsilon)J_2(\varepsilon)}{B(\varepsilon)}\bigg(\zeta(\varepsilon) + \log J_1(\varepsilon) + (\beta + \varepsilon)\frac{d\zeta}{d\varepsilon}(\varepsilon)\bigg) - \frac{\frac{dJ_1(\varepsilon)}{d\varepsilon}J_2(\varepsilon)}{B(\varepsilon)}\bigg(\beta + \varepsilon\bigg)\\
        &\qquad+ \frac{J_1(\varepsilon)J_2(\varepsilon)^2}{B(\varepsilon)^2}\bigg(2\beta(\beta+\varepsilon) + 2(\beta+\varepsilon)^2(\zeta(\varepsilon) + \log J_1(\varepsilon))\bigg)\\
        &\qquad - \frac{\frac{dJ_1}{d\varepsilon}(\varepsilon)J_2(\varepsilon)^2}{B(\varepsilon)^2}\bigg(2\beta(\beta+\varepsilon)^2 - (\beta+\varepsilon)^3(\zeta(\varepsilon) + \log J_1(\varepsilon))\bigg)\\
        &\qquad - \frac{J_1(\varepsilon)^2 \frac{dJ_2}{d\varepsilon}(\varepsilon)}{B(\varepsilon)^2}\bigg(\beta^3 -(\beta+\varepsilon)\beta^2 (\zeta(\varepsilon) + \log J_1(\varepsilon))\bigg).
    \end{align*}
    Furthermore \(|\frac{d\zeta}{d\varepsilon}(\varepsilon)| = |\frac{1-w}{\beta + \varepsilon(1-w)} + \frac{\beta(1-w)}{(\beta + \varepsilon(1-w))^2}|\leq \frac{6}{\beta}\) (using that \(|\varepsilon| < \beta/2\)). Hence, using the previous bounds on \(\frac{dJ_1}{d\varepsilon}\) and \(\frac{dJ_2}{d\varepsilon}\) we find
    \begin{align*}
        \bigg|\frac{d^2J_1}{d\varepsilon^2}(\varepsilon)\bigg| &\leq \frac{c_u + 12}{\beta\mathcal{E}_{\Gamma}}|J_1(\varepsilon)||J_2(\varepsilon)| + \frac{4 + 4c_u}{ \mathcal{E}_{\Gamma}^2}|J_1(\varepsilon)|^2|J_2(\varepsilon)|\\
        &\qquad+ \frac{4+8c_u}{ \mathcal{E}_{\Gamma}}|J_1(\varepsilon)||J_2(\varepsilon)|^2 + \frac{16(1+c_u)^2\beta}{\mathcal{E}_{\Gamma}^3}|J_1(\varepsilon)||J_2(\varepsilon)|^3\\
        &\qquad+ \frac{2\beta(1+c_u)(1+2c_u)}{\mathcal{E}_{\Gamma}^3}|J_1(\varepsilon)|^3|J_2(\varepsilon)| + \frac{2\beta(1+2c_u)}{\mathcal{E}_{\Gamma}^3}|J_1(\varepsilon)|^2|J_2(\varepsilon)|^2.
    \end{align*}
    Similarly differentiating \eqref{eq:J2Derivative} yields 
    \begin{align*}
        \frac{d^2J_2}{d\varepsilon^2}(\varepsilon) &= - \frac{J_2(\varepsilon)^2}{B(\varepsilon)} + \frac{\frac{dJ_1}{d\varepsilon}(\varepsilon)J_2(\varepsilon)^2}{J_1(\varepsilon)B(\varepsilon)}(\beta+\varepsilon) + \frac{J_1(\varepsilon)J_2(\varepsilon)}{B(\varepsilon)}\bigg(\beta \frac{d\zeta}{d\varepsilon}(\varepsilon)\bigg)\\
        &\qquad + \frac{\frac{dJ_1}{d\varepsilon}(\varepsilon)J_2(\varepsilon)}{B(\varepsilon)}\beta - \frac{J_2(\varepsilon)\frac{dJ_2}{d\varepsilon}(\varepsilon)}{B(\varepsilon)}(\beta + \varepsilon)\\
        &\qquad - \frac{J_1(\varepsilon)J_2(\varepsilon)^2}{B(\varepsilon)^2}2\beta(\beta+\varepsilon)(\zeta(\varepsilon) + \log J_1(\varepsilon))\\
        &\qquad + \frac{\frac{dJ_1}{d\varepsilon}(\varepsilon)J_2(\varepsilon)^2}{B(\varepsilon)^2}\beta(\beta+\varepsilon)^2(\zeta(\varepsilon) + \log J_1(\varepsilon))\\
        &\qquad+ \frac{J_1(\varepsilon)^2\frac{dJ_2}{d\varepsilon}(\varepsilon)}{B(\varepsilon)^2}\beta^3(\zeta(\varepsilon) + \log J_1(\varepsilon))\\
        &\qquad+ \frac{J_2(\varepsilon)^3}{B(\varepsilon)^2}\bigg(2(\beta+\varepsilon)^2\bigg) - \frac{\frac{dJ_1}{d\varepsilon}(\varepsilon)J_2(\varepsilon)^3}{J_1(\varepsilon)B(\varepsilon)^2}\bigg((\beta+\varepsilon)^3\bigg)\\
        &\qquad- \frac{J_1(\varepsilon)J_2(\varepsilon)\frac{dJ_1(\varepsilon)}{d\varepsilon}(\varepsilon)}{B(\varepsilon)^2}\beta^2(\beta+\varepsilon)
    \end{align*}
    which entails 
    \begin{align*}
        \bigg|\frac{d^2J_2}{d\varepsilon^2}(\varepsilon)\bigg| &\leq \frac{1}{\beta\mathcal{E}_{\Gamma}}|J_2(\varepsilon)|^2 + \frac{6(1+c_u)}{\mathcal{E}_{\Gamma}^2}|J_2(\varepsilon)|^3 + \frac{6}{\beta \mathcal{E}_{\Gamma}}|J_1(\varepsilon)||J_2(\varepsilon)|\\
        &\qquad+  \frac{6(1+c_u)}{\mathcal{E}_{\Gamma}^2}|J_1(\varepsilon)||J_2(\varepsilon)|^2 + \frac{4}{\mathcal{E}_{\Gamma}^2}|J_2(\varepsilon)|^3\\
        &\qquad+ \frac{8(1+c_u)}{\mathcal{E}_{\Gamma}^2}|J_1(\varepsilon)||J_2(\varepsilon)| + \frac{16\beta(1+c_u)^2}{\mathcal{E}_{\Gamma}^3}\\
        &\qquad+ \frac{4\beta(1+c_u)^2}{\mathcal{E}_{\Gamma}^3}|J_1(\varepsilon)|^3|J_2(\varepsilon)| + \frac{4\beta(1+c_u)}{\mathcal{E}_{\Gamma}^3}|J_1(\varepsilon)|^2|J_2(\varepsilon)|^2\\
        &\qquad+ \frac{8}{\mathcal{E}_{\Gamma}^2}|J_2(\varepsilon)|^3 + \frac{16\beta(1+c_u)}{\mathcal{E}_{\Gamma}^3}|J_2(\varepsilon)|^3\\
        &\qquad+ \frac{4\beta(1+c_u)}{\mathcal{E}_{\Gamma}^3}|J_1(\varepsilon)|^2|J_2(\varepsilon)|^2 + \frac{4\beta}{\mathcal{E}_{\Gamma}^3}|J_1(\varepsilon)||J_2(\varepsilon)|^3.
    \end{align*}
\end{proof}

\subsection{Control of \texorpdfstring{\(H_u\)}{H} with Respect to \texorpdfstring{\(\varepsilon\)}{epsilon}}

\begin{proposition}\label{prop:H1DerivativeEstimates}
    Let \(\beta_1 = \beta\) and \(\beta_2 = \beta + \varepsilon\). For a fixed \(\beta\), the following uniform bounds hold 
    \begin{align*}
        \sup_{|\varepsilon| < \beta/2}\sup_{w \in [0,1]}\frac{1}{w(1-w)}\left|\frac{\partial H_u(w,\varepsilon)}{\partial \varepsilon}\right|^2 &< \infty,\\
        \sup_{|\varepsilon| < \beta/2}\sup_{w \in [0,1]}\frac{1}{w(1-w)}\left|\frac{\partial^2 H_u(w,\varepsilon)}{\partial \varepsilon^2}\right| &< \infty.
    \end{align*}
\end{proposition}

\begin{proof}[Proof of Proposition \ref{prop:H1DerivativeEstimates}]
    As in the statement of Proposition \ref{prop:piEpsilonDerivatives}, we write \({\bm J}\) as a shorthand for \({\bm J}(u\Phi^{-1}(w))\). Define the norm
    \begin{equation*}
        \|f\|_{[0,1]} = \sup_{|\varepsilon| < \beta/2}\sup_{w \in [0,1]} \frac{|f(w)|}{w(1-w)}.
    \end{equation*}
    Using Proposition \ref{prop:growthFactorBound}, as in the proof of Proposition \ref{prop:nonDegenerateCoordinateTransform}, we have
    \begin{equation}\label{eq:productBound}
        \|J_1J_2\|_{[0,1]} < \infty. 
    \end{equation}
    Combining \eqref{eq:productBound} with Proposition \ref{prop:piEpsilonDerivatives} it follows that 
    \begin{equation*}\label{eq:supNormFirstOrder}
        \left\|J_1\frac{d J_2}{d \varepsilon}\right\|_{[0,1]} < \infty, \quad \left\|\frac{d J_1}{d \varepsilon}J_2\right\|_{[0,1]} < \infty.
    \end{equation*}
    Then using the chain rule 
    \begin{align*}
        \frac{\partial H_u(w,\varepsilon)}{\partial \varepsilon}
        &= \frac{J_2}{(J_1 + J_2)^2}\frac{dJ_1}{d\varepsilon}- \frac{J_1}{(J_1 + J_2)^2}\frac{dJ_2}{d\varepsilon}.
    \end{align*}
    From \eqref{eq:supNormFirstOrder} we see that 
    \begin{equation*}
        \left\|\frac{\partial H_u(\cdot,\varepsilon)}{\partial \varepsilon}\right\|_{[0,1]} < \infty. 
    \end{equation*}
    Furthermore, the combination of \eqref{eq:productBound} with Proposition \ref{prop:piEpsilonDerivatives} also implies 
    \begin{align*}
        \left\|J_1\left(\frac{d J_2}{d \varepsilon}\right)^2\right\|_{[0,1]} < \infty,\quad \left\|\frac{d J_1}{d \varepsilon}\frac{d J_2}{d \varepsilon}\right\|_{[0,1]}& < \infty,  \quad \left\|\left(\frac{d J_1}{d \varepsilon}\right)^2J_2\right\|_{[0,1]} < \infty,\\
        \left\|J_1 \frac{d^2 J_2}{d \varepsilon^2}\right\|_{[0,1]} < \infty,& \quad \left\| \frac{d^2 J_1}{d \varepsilon^2}J_2\right\|_{[0,1]} < \infty.
    \end{align*}
    And by the chain rule 
    \begin{align*}
        \frac{\partial^2 H_u(w,\varepsilon)}{\partial \varepsilon^2} &= -\frac{2J_2}{(J_1 + J_2)^3}\left(\frac{dJ_1}{d\varepsilon}\right)^2 + \frac{2(J_1 - J_2)}{(J_1 + J_2)^3} \frac{dJ_1}{d\varepsilon} \frac{dJ_2}{d\varepsilon} + \frac{2J_1}{(J_1 + J_2)^3} \bigg(\frac{dJ_2}{d\varepsilon}\bigg)^2\\
        &+ \frac{J_2}{(J_1 + J_2)^2} \frac{d^2 J_1}{d\varepsilon^2}- \frac{J_1}{(J_1 + J_2)^2}\frac{d^2J_2}{d\varepsilon^2},
    \end{align*}
    so that 
    \begin{equation*}
        \left\|\frac{\partial^2 H_u(\cdot,\varepsilon)}{\partial \varepsilon^2}\right\|_{[0,1]} < \infty. 
    \end{equation*}
\end{proof}

\section{Mathematical Glossary}\label{sec:glossary}

    \resizebox{\textwidth}{!}{
        \centering
        \begin{tabular}{c c c c }
        \hline
        Notation & Name & Formula & Section\\ 
        \hline
    
        \(a\) & Lower Bound on Catastrophic Probabilities & & \ref{sec:ModelDefinition}\\
    
        \({\bf A}\) & Generator of the Frequency Process & cf. \eqref{eq:generatorFormula} & \ref{sec:fixationProbability}\\
    
        \(C^2(\hat{E})\) & & & \ref{sec:intro}\\
        
        \(C_b^2(\hat{E})\) & & &\ref{sec:intro}\\
    
        \(D_{\hat{E}}[0,T]\) & Skorokhod Space & &\ref{sec:topology}\\ 
    
        \(D_{\Gamma}\) & & \(D_{\Gamma}({\bm x}) = (F \circ R^*({\bm \beta} \cdot {\bm x}) - \nu^{-1})^2\) & \ref{subsec:proofNonExtinction}\\
    
        \(E\) & & \([0,R_{\max}] \times \ZZ^2_{\geq 0}\) & \ref{sec:ModelDefinition}\\
        
        \(\hat{E}\) & & \([0,R_{\max}] \times \RR^2_{\geq 0}\) & \ref{sec:ModelDefinition}\\
    
        \(\mathcal{E}_{\Gamma}\) & & \((F \circ R^*)^{-1}(\nu^{-1})\)& \ref{sec:deterministicPrologue}\\
    
        \(F\) & Functional Response & \(F(R) = \frac{R}{a+R}\)& \ref{sec:deterministicPrologue}\\
        
        \(G\) & Resource Growth & \(G(R) = R_{\max} - R\)& \ref{sec:deterministicPrologue}\\
    
        \(H_u\) & & \(\Phi \circ {\bm \pi} \circ (u \odot )\) & \ref{sec:fixationProof}\\
            
        \(K\) & Carrying Capacity & \(\beta^{-1}\mathcal{E}_{\Gamma}\)& \ref{sec:deterministicPrologue}\\
    
        \(\mathscr{P}\) & Probability Measures on \([0,T] \times \hat{E}\) & & \ref{sec:topology}\\
    
        \(R_{\max}\) & Maximal Resource Level & & \ref{sec:deterministicPrologue}\\
    
        \(R^{\gamma}\) & Resource Process & & \ref{sec:ModelDefinition}\\
        
        \(R^*\) & Limit Resource Level & \(R^*(\mathcal{E}) = \frac{R_{\max} - \mathcal{E} - a}{2} + \frac{\sqrt{4aR_{\max} + (R_{\max} -a - \mathcal{E})^2}}{2}\)&\ref{sec:deterministicPrologue}\\
    
         \(\hat{R}^{\gamma}\) & Time-Accelerated Resource Process & \(\hat{R}_t^\gamma = R_{\gamma t} \)& \ref{sec:mainResults}\\
    
         \(\bar{R}^{\gamma}\) & Jump-Free Resource Process &  & \ref{subsec:proofNonExtinction}\\
    
         \({\bf U}\) & Survival Probabilities Sample & \({\bf U} \sim \text{Law}(\Upsilon)\) & \ref{sec:ModelDefinition}\\
    
        \(W\) & Frequency Process & \(\Phi(\hat{\bf Z})\) & \ref{sec:fixationProbability}\\
    
         \({\bf X}^{\gamma}\) & Population Process & & \ref{sec:ModelDefinition}\\
        
        \(\hat{\bf X}^{\gamma}\) & Time-Accelerated Rescaled Population Process & \(\bar{\bf X}^{\gamma}_t = \gamma^{-1}{\bf X}^{\gamma}_{\gamma t} \) & \ref{sec:mainResults}\\
    
        \(\bar{\bf X}^{\gamma}\) & Jump-Free Population Process &  & \ref{subsec:proofNonExtinction}\\
    
        \(\hat{\bf Z}^{\gamma}\) & Projected Process & \({\bm \pi}(\hat{\bf X}^\gamma)\) & \ref{sec:mainResults}\\
        
        \(\hat{\bf Z}\) & Limit Projected Process & & \ref{sec:mainResults}\\
    
        \(\beta\) & Maximal Birth Rate & & \ref{sec:deterministicPrologue}\\
    
        \(\beta_{\textrm{RDS}}\) & Risk-Dominant Strategy & & \ref{sec:discussion} \\ 
    
        \(\gamma\) & Resource Abundance & & \ref{sec:deterministicPrologue}\\
    
        \(\Gamma\) & Stable Manifold & \(\{{\bm x} \in \RR^2_{\geq 0} : {\bm \beta} \cdot {\bm x} = \mathcal{E}_{\Gamma}\}\) & \ref{sec:deterministicPrologue}\\
    
        \(\theta\) & Fixation Probability & \(\theta(w) = \PP{w}{\tau_1 < \tau_0}\) & \ref{sec:fixationProbability}\\
    
        \(\kappa\) & Rate of Catastrophes & & \ref{sec:ModelDefinition}\\
        
        \(\mu\) & Death Rate & & \ref{sec:deterministicPrologue}\\
        
        \(\nu\) & Lifetime Reproductive Success &\(\beta\mu^{-1}\) & \ref{sec:deterministicPrologue}\\
        
        \({\bm \pi}\) & Projection Map & cf. \eqref{eq:implicitSystem} & \ref{sec:deterministicPrologue}\\
    
        \(\tau^{\varepsilon}\) & & cf. \eqref{E:tauepsilon} & \ref{subsec:convergenceInProbability}\\
    
        \(\Upsilon\) & Catastrophic Jump Distribution & & \ref{sec:ModelDefinition}\\
        
        \(\Phi\)& Frequency Map & \(\Phi({\bm x}) = \frac{x_1}{x_1+x_2}\)& \ref{sec:deterministicPrologue}\\

        \(\Phi^{-1}\) & Inverse Frequency Map & \(\Phi^{-1}(w) = \frac{\mathcal{E}_{\Gamma}(w,1-w)}{\beta_1 w + \beta_2(1-w)}\)& \ref{sec:fixationProbability}\\
    
        \({\bm \Psi}\) & Set of Pseudopaths & & \ref{sec:topology}\\
        
        \end{tabular}
    }

\bibliographystyle{abbrv}
\bibliography{refs}

\end{document}